\documentclass[12pt]{article}
\usepackage[margin=1in]{geometry}
\usepackage{mathrsfs, mathtools}
\usepackage[cal=boondoxo]{mathalfa}
\usepackage[title,toc,titletoc]{appendix}
\usepackage{titlesec}
\titleformat{\section}[block]{\large\normalfont\filcenter}{\thesection.}{.5em}{}
\titleformat{\subsection}[hang]{\normalfont \bfseries}{\thesubsection.}{.5em }{}[]
\usepackage[dvipsnames]{xcolor}
\usepackage{bbm}
\usepackage{bm}
\usepackage{amsmath, amssymb}	
\usepackage[shortlabels]{enumitem}
\usepackage[hang,flushmargin]{footmisc}
\usepackage[round]{natbib}
\PassOptionsToPackage{hyphens}{url}
\usepackage[colorlinks=true,linkcolor = blue, citecolor=blue, hyperfootnotes=false, hyperindex,breaklinks]{hyperref}
\usepackage{cleveref}
\usepackage{bigints}

\newtheorem{lema}{Lemma}

\usepackage{pgf}
\usepackage[justification=centering]{caption}
\usepackage{subcaption}
\usepackage{tikz}
\usepackage{tikz-3dplot}
\usetikzlibrary{patterns, decorations.pathreplacing}
\usepackage{hhline}
\usepackage{float}
\usepackage{placeins}
\usepackage{siunitx}
\usepackage{textcomp}
\usepackage[subtle]{savetrees}
\usepackage{graphics}
\usepackage{graphicx}

\newtheorem{theorem}{Theorem}
\newtheorem{proposition}{Proposition}

\newtheorem{corollary}{Corollary}
\newtheorem{definition}{Definition}
\newtheorem{lemma}{Lemma}

\newenvironment{proof}[1][Proof]{\textbf{#1.} }{\  \rule{0.5em}{0.5em}}
\DeclareMathOperator*{\argmax}{arg\,max}
\DeclareMathOperator*{\argmin}{arg\,min}

\newcommand{\norm}[1]{\left\vert\left\vert #1\right\vert\right\vert_1}
\DeclareMathOperator*{\supp}{supp}
\DeclareMathOperator*{\cl}{cl}
\DeclareMathOperator*{\conv}{conv}

\usepackage{setspace}
\usepackage{multirow}
\definecolor{shade}{gray}{.7}
\begin{document}
\title{Competition, Persuasion, and Search\thanks{For helpful comments and discussions, we thank Alessandro Bonatti, Joyee Deb, Th\'{e}o Durandard, Jack Fanning, Daniel Fershtman, Marina Halac, Teddy Kim, Elliot Lipnowski, Erik Madsen, Harry Pei, Larry Samuelson, Quitz\'e Valenzuela-Stookey, Juuso V\"{a}lim\"{a}ki, Rajiv Vohra, Mark Whitmeyer, Kun Zhang, and seminar audiences at Brown, Yale, U Chicago, Virtual Seminar in Economic Theory, SEA, SITE, and Stony Brook. Mekonnen also thanks the Cowles Foundation for Research in Economics, where he spent a majority of his time while conducting this research. }}

\author{Teddy Mekonnen\thanks{Department of Economics, Brown University. Contact: \href{mailto:mekonnen@brown.edu}{mekonnen@brown.edu}}\and Bobak Pakzad-Hurson\thanks{Department of Economics, Brown University. Contact: \href{mailto:bph@brown.edu}{bph@brown.edu}}}
\date{\today}
\maketitle
\thispagestyle{empty}
\setcounter{page}{0}
\begin{abstract}
How does competition in markets for information affect the creation and division of surplus? We study this question in a  search environment in which an agent searches sequentially for a high-quality good and learns about the quality of sampled goods by repeatedly purchasing signals from profit-maximizing information brokers. Brokers design and price signals but can commit only to spot contracts. We characterize the equilibrium payoff set as a function of the market structure---the number of competing brokers. When search costs are low, market structure affects neither surplus generation nor its division. When costs are high, however, competition benefits the agent but reduces total surplus relative to monopoly. Methodologically, we extend repeated-games theory to stopping problems such as sequential search.

\end{abstract}

 \newpage
 \allowdisplaybreaks
\section{Introduction}
Agents in search markets aim to find high-quality goods quickly but face a key hurdle: limited information about the quality of available goods. However, in an increasingly data-driven economy, agents may acquire this information by contracting with brokers who gather, process, and sell data. Examples of such \emph{persuaded search} markets---where agents contract with ``upstream" information brokers to guide their decisions in a ``downstream" sequential search market---include: labor markets in which employers contract with pre-employment testing agencies to screen job candidates; dating markets in which individuals purchase background checks on potential dates; and used automobile markets where buyers rely on vehicle history reports to uncover hidden defects.\footnote{82\% of US private-sector firms use pre-employment testing services to screen candidates (\url{https://www.sparcstart.com/wp-content/uploads/2018/03/2017-CandE-Report.pdf}). In a 2023 survey of U.S. singles, 18\% of respondents say they have solicited a background check on their dates (\url{thrivingcenterofpsych.com/blog/dating-in-2023-here-is-where-most-singles-are-living-in-the-us/}). Similarly, according to an April 2024 survey for \textit{Experian}, 70\% of used car shoppers used a vehicle history report for their last automobile purchase. (\url{https://www.experian.com/content/dam/marketing/na/automotive/infrographic/adding-a-second-vhr-could-help-sales-infographic.pdf}). }

A primary driver of surplus in such agent–broker relationships is the broker's ability to facilitate the agent's search for goods of high match value in the downstream market. Empirical evidence supports the salience of product quality in search markets:  a majority of survey respondents, largely online shoppers, state that finding a desirable good is an important search motive \citep{coey2020discounts}. Indeed, in some search markets, downstream prices are either exogenously fixed or non-existent, making match value the sole channel for surplus creation. In labor markets, pre-employment tests typically affect hiring decisions rather than wages \citep{aut08}; in dating markets, the search goods are potential partners and thus command no price; and used car platforms like \textit{Carvana} strictly anchor prices to \textit{Kelley Blue Book} values, prohibiting price negotiations by design.\footnote{\url{www.prnewswire.com/news-releases/carvana-integrates-kelley-blue-book-suggested-retail-values-into-online-vehicle-inventory-300023390.html}}

To fix a motivating example, consider an agent searching for a used car. Increasingly, search takes place through online marketplaces such as \emph{Carvana},\footnote{\emph{Carvana} was the second-largest U.S. retailer of used cars in 2021 (\url{www.caranddriver.com/features/a60257928/carvana-recovery-progresss}). } where direct inspection of cars is infeasible. Thus, the searching agent may turn to brokers like \emph{CARFAX} and \emph{AutoCheck}, who aggregate and process data from law enforcement, insurance companies, and auto-repair centers into vehicle history reports that they often sell on a per-car basis via spot contracts.  As the reports become more expensive---which reduces the agent's share of the surplus---or less informative---which reduces surplus by diminishing the agent’s ability to distinguish high- from low-quality cars---the agent's incentive to search is also reduced. Therefore, the brokers' design and pricing of information jointly determine both the amount of surplus created in the downstream search market and how that surplus is divided between the agent and the brokers upstream.

The dependence of the downstream search on the upstream contracting environment raises natural questions: How does competition, or lack thereof, among information brokers shape market efficiency and the welfare of the searching agent? What tradeoffs does a regulator face  when deciding whether to break
up a monopolist broker? The policy question is not merely hypothetical. \textit{CARFAX} has faced antitrust litigation alleging monopoly power in vehicle-history certification; the dispute turned partly on whether its contractual practices harmed competition or promoted reliable information provision.\footnote{We briefly discuss this litigation  (\emph{Maxon Hyundai Mazda et al. v Carfax, Inc}) and our model's relationship to it, in the Conclusion.}

To address these questions, we embed a repeated contracting problem between an agent and a finite number of information brokers into a stopping game. The agent (``he") has unit demand for a good and attempts to fulfill this demand in a downstream market by engaging in sequential search, following the canonical framework of \citet{mcc70}. In each period, he incurs a search flow cost and samples one good. However, unlike the standard model, he does not observe the good’s quality directly. Instead, he can acquire a signal about the sampled good's quality from one of the brokers in the upstream information market. Each broker (``she") flexibly designs a signal and sells it through a take-it-or-leave-it spot contract. The agent buys one signal, if any, observes the corresponding signal realization, and then decides whether to stop and match with the sampled good or to continue searching in the next period.

Our model is designed as a minimal departure from the canonical \citet{mcc70} framework. The search environment, including the prior distribution of quality and the agent's search cost, is common knowledge. Thus, the payoff-relevant uncertainty at the contracting stage is object-level: before a signal is produced, neither the agent nor the brokers know the quality of the currently sampled good. We also abstract from hidden agent heterogeneity. The agent has no privately known valuation type or search cost. This abstraction is natural in data-rich digital markets in which many payoff-relevant agent characteristics are observable, for example through cookies, transaction histories, platform records, or other digital traces, and can be used to tailor the terms of the contract. Equivalently, the model studies the contracting problem for a fixed, commonly observed agent type. This allows us to isolate the dynamic distortion generated by repeatedly selling information, rather than adverse-selection distortions generated by hidden agent types.

Repeated contracting within a search framework introduces a fundamental tension. Brokers generate revenue only as long as the agent continues to search, and thus have an incentive to prolong the search process by distorting the informativeness of their signals. On the other hand, the agent is willing to purchase a signal only if it is sufficiently instrumental in distinguishing high- from low-quality goods. As we show, the resolution of this tension depends on the \emph{market structure}---the number of competing brokers.

To analyze the impact of market structure, we first consider the benchmark of history-independent (stationary) strategies. We show that the agent's payoff is higher in a competitive setting than in a monopolistic one, echoing the classic result that Bertrand competition among sellers benefits the buyer. However, we find that the market structure has no impact on total surplus, with both monopoly and competition leading to full efficiency. This benchmark invites a compelling, albeit misleading, intuition: competition serves only to increase the agent's share of the surplus without affecting efficiency.

Restricting attention to history-independent strategies, however, shuts down important regulatory concerns such as (tacit) collusion or loyalty-pricing practices, which can only be sustained through history-dependent strategies. Characterizing the full set of equilibria overturns the intuition that competition serves only to shift surplus to the searching agent. We identify a crucial threshold in the agent's search flow cost that divides the comparison of monopoly and competition into two distinct regimes. When the search cost falls below the threshold, monopoly and competition yield the same range of total surpluses \emph{and} agent payoffs. Above this threshold, however, the highest total surplus attainable under competition coincides with the lowest total surplus under monopoly, while the lowest agent payoff under competition is no smaller than the highest agent payoff under monopoly.  Moreover, for some cost values above the threshold, nearly all competitive equilibria deliver strictly higher payoffs to the agent but strictly lower total surplus than does monopoly. 

In short, our results establish that $(i)$ for low search costs, the creation and division of surplus is independent of market structure, while $(ii)$ for high search costs, competition benefits the agent but reduces total welfare. Our model of an information-rich market without hidden agent heterogeneity invites an analogous comparison of a first-degree price-discriminating monopolist to a Bertrand market, in which competition generates none of the familiar extensive-margin welfare gains from expanding market participation. Thus, one might expect that competition in our setting would merely transfer surplus from brokers to the agent. Our contribution is to show that this intuition misses a distinct intensive-margin channel: in our setting, information is an intermediate good that governs the downstream search process. Its supply and pricing create novel strategic distortions that inextricably link surplus creation in the downstream search market to its division in the upstream information market. Whenever market structure is irrelevant for surplus creation, it is also irrelevant for surplus division; and whenever competition improves the agent's payoff, it has scope to be less efficient than monopoly.

These results also speak to antitrust remedies. A breakup of an information monopolist does not necessarily replace monopoly pricing with Bertrand pricing. Instead, repeated-contracting considerations turn a familiar regulatory question---should a monopolist be broken up?---into a tradeoff between surplus division and efficiency.

At the heart of our main results is the observation that the value of a signal to the agent depends on two distinct factors. The first is standard in any decision problem: a signal is valuable insofar as it helps the agent distinguish goods worth matching with from those that justify continued search. Since brokers design signals each period, they fully control this factor. We show that brokers do not require complex signal structures; binary (``pass-fail'') signals are sufficient to characterize all equilibrium outcomes. Moreover,  these outcomes can be supported by brokers offering identical signals, so the agent derives no marginal value from purchasing multiple signals in a given period. The second factor is the agent’s continuation payoff. When the agent anticipates high future payoffs, he is less likely to stop searching today and thus has less need for information. This factor is unique to stopping games: because the agent has unit demand for the search good, he prefers to continue searching rather than match today and forgo a high continuation payoff. Consequently, a signal's value in a given period declines as the agent’s continuation payoff rises. He can thus credibly reject a contract that offers him a low payoff today if he expects to get a sufficiently high continuation payoff in future periods. This intertemporal dependence of a signal's value and continuation payoffs is the endogenous source of  bargaining power for the agent.\footnote{\cite{dana2011productquality} study a related \emph{repeated sales game} in which, in every period, a buyer purchases a good from sellers who choose both its price and quality. Because their setting is an infinitely repeated game rather than a stopping game, the consumption value of a good in the current period is independent of continuation payoffs. They show that the buyer has bargaining power under monopoly if and only if he has it under competition. By contrast, in our persuaded search setting, the agent’s value of a signal in a given period depends directly on his continuation payoff. As our results will show, this distinction implies that the agent in our model has a strictly stronger bargaining position under competition than under monopoly---he has bargaining power under competition whenever he has it under a monopoly, but the converse is not true.} 

How does the agent's endogenous bargaining power differ by market structure? Consider first the monopolistic case with a single broker. The agent is naturally in a weak bargaining position because he faces a take-it-or-leave-it offer from his only source of information. He may attempt to gain leverage by threatening to reject unfavorable spot contracts in anticipation of high continuation values. In a stationary equilibrium, however, continuation values are fixed, rendering such threats non-credible. As we show, this lack of bargaining power results in a unique stationary equilibrium which features full surplus extraction---an efficient outcome that delivers the agent his lowest possible payoff.

A richer set of outcomes arises under non-stationary equilibria, precisely because continuation payoffs can now be conditioned on the history of play. This flexibility can render the agent’s threats credible, thereby endogenously granting him bargaining power. Importantly, contracts offered by the broker are multidimensional, specifying both a price and a signal structure. Any bargaining power that sustains an efficient equilibrium yielding the agent a given expected payoff (e.g., the broker selling an informative signal at a high price) can equally sustain an inefficient one yielding him the same payoff (e.g., the broker selling a less informative signal at a lower price). As a result, once the agent acquires bargaining power, the equilibrium set expands to include outcomes that are both more favorable to him and less efficient.

Next, consider a competitive market structure with multiple brokers. The unique stationary equilibrium outcome here features Bertrand competition: in each period, at least two brokers offer a fully informative signal for free,\footnote{Or, alternatively, a (binary) signal that induces the same search behavior, which we show always exists.} allowing the agent to both generate and consume the maximum surplus possible in the search market. Similar to the stationary equilibrium under a monopoly, this outcome is efficient, but unlike the monopolistic case, the agent now receives the entire surplus. When we move beyond stationary strategies in competitive settings, brokers may tacitly collude to recover some of the surplus they would otherwise compete away.\footnote{We do not model an official communication stage between brokers. Nonetheless, ``tacit collusion'' arises when brokers coordinate on/select a broker-preferred equilibrium. Alternatively, one can envision such coordination arising through explicit communication channels.}  Yet the agent can blunt such collusion by once again threatening to reject contracts that yield him too little surplus. Relative to monopoly, however, competition yields the agent an advantage: rejecting one broker’s contract does not preclude obtaining information in the current period, since the agent can instead accept another broker’s offer. In this way, competition provides the agent with an ``inside option'' that amplifies his endogenous bargaining power. As a result, competition more readily sustains equilibria that yield high agent payoffs as well as inefficient outcomes. 

Formally, we show that the agent's endogenous bargaining power is closely tied to his search flow cost. This is because the largest feasible agent's continuation payoff increases as the search flow cost decreases, making his threats more credible. Under monopoly, there exists a threshold search cost below which threats are credible, supporting any surplus split in equilibrium. Above the threshold, threats lose credibility, and the broker’s preferred outcome—full surplus extraction—emerges as the unique equilibrium outcome. Thus, the equilibrium payoff set in the monopolistic setting exhibits a \emph{bang-bang} structure: depending on the search cost, either a folk-theorem or a unique equilibrium outcome obtains. With multiple brokers, the agent’s ``inside option" makes  threats credible for a broader range of search costs. Hence, a folk-theorem result obtains under competition whenever it does in a monopolistic setting. Moreover, even for search cost parameters where full surplus extraction is the unique equilibrium outcome in a monopolistic setting, competition can allow for infinitely many equilibrium outcomes, nearly all of which are less efficient than the unique equilibrium outcome of a monopoly. This result holds regardless of the number of competing brokers, implying that the set of all equilibrium outcomes does not shrink to the familiar Bertrand outcome in the limit as competition increases.

We take an agnostic approach when deriving our welfare comparisons across market structures in that we do not rely on a specific equilibrium selection rule.\footnote{Indeed, selecting specific equilibria may lead to conclusions that conflict or do not generalize across the entire set of equilibria. For example, if one restricts attention to the ``best case" (i.e., efficiency-maximizing) equilibrium outcome per market structure, then both monopoly and competition would yield the same total surplus but the latter would yield a higher agent payoff. In contrast, if one selects the ``worst-case" (i.e., efficiency-minimizing) equilibrium outcome, then monopoly dominates competition in terms of efficiency but yields the agent his lowest possible payoff. Neither of these selection rules therefore captures the full nuance of welfare implications across the set of all equilibria.} Rather, we consider the \emph{entire} set of equilibrium payoffs. To our knowledge, this paper is the first to do so for a stopping game with rich action spaces. The classical theory of repeated games primarily focuses on  infinitely repeated, simultaneous-move stage games in which a player's payoff in the  repeated game equals the discounted sum of her stage-game payoffs. However, many of the tools for characterizing  equilibrium payoffs and folk-theorem results do not readily apply to our persuaded search game, which diverges from the classical theory in three key ways: First, stopping games are not repeated games because they lack well-defined stage games.\footnote{As we elaborate in \Cref{sec:single-agent}, this distinction implies that expected payoffs in a stopping game are not additively separable into a flow payoff and a continuation value.} This implies that a player's equilibrium payoff cannot be bounded by the usual stage-game minimax payoff. Second, we study a game in which players do not move simultaneously in each period. As a result, equilibrium strategies must satisfy sequential rationality within each period. Finally, as is typical among search models, our game features a search cost rather than discounting.\footnote{The mechanism behind our main results---the agent's endogenous bargaining power derived from the interplay between a signal's value and continuation payoffs, and the inside option arising from the presence of competing brokers---naturally extends to the case where the agent's search friction is driven by discounting rather than a flow cost.} Thus, the well-known finding that the \emph{one-shot deviation principle} implies sequential rationality does not immediately transfer to our setting. To contend with these differences, we adapt the set-valued dynamic programming approach of  \cite*{abr90} (henceforth, ``APS") to our setting, and further reduce their set-valued fixed-point operator to a single-dimensional fixed-point problem, yielding a tractable characterization of equilibrium payoff sets for stopping games. Furthermore, we complement our APS characterization by constructing a family of reward and punishment schemes that yield all equilibrium outcomes, providing one (of possibly many) equilibrium strategy profiles themselves.

\subsection{Related Literature}

 Our paper relates to a small body of work exploring the role of competitive information provision in search markets. \cite{boa19} and \cite{he23} study how competition among information designers affects outcomes in a random search market, and \cite{au2023attraction, au24} study the case of a directed search market. These papers consider an agent who engages bilaterally with a sequence of designers as in \cite{dia71}, making each designer-agent relationship short-lived and effectively monopolistic. Consequently, equilibria feature a ``Diamond Paradox,"  where the agent terminates his search in the first period and the initial designer captures the full surplus, despite the presence of competition. In contrast, we consider an agent who engages simultaneously with all brokers in each period, thereby eliminating the local monopoly power of any one broker.\footnote{Our model does not involve the agent searching for brokers. Given the dynamic nature of search problems, a natural way to include this feature in a richer model would be to allow the agent to pay a cost to add an additional broker to his ``consideration set'' in all future periods as in \cite{Fershtman_and_pavan}. The agent's ability to simultaneously engage with multiple brokers, albeit at a cost, leads to similar strategic considerations as in our model. In particular, it would not lead to the familiar Diamond Paradox: a broker who currently is the only one in the agent's consideration set can be disciplined to offer some of the surplus from search to the agent by the threat of the agent adding a rival broker to his consideration set and triggering a Bertrand price war in future periods.} This leads to a richer set of equilibrium outcomes, ranging from Bertrand-like competition to collusion among the brokers. 

A larger literature examines the role of information on search market outcomes when a searching agent contracts with a single information designer.  \cite{sato23} characterize the optimal signal for a designer who influences the search behavior of an agent engaged in ordered search \`a la \cite{wei79}. Similarly, \cite{do22} and \cite{hu22} characterize the optimal signal for a designer with intrinsic preferences over the outcomes of a large random search market. \cite{mek23} do the same for a profit-maximizing information broker; they find a unique stationary equilibrium outcome of a persuaded search game between a single broker and a searching agent. Our work refines theirs by showing that in monopoly market structures, their unique stationary equilibrium outcome is in fact the unique equilibrium outcome for high search costs whereas a folk-theorem result obtains for low search costs. 

A complementary set of papers  \citep[see, e.g.,][]{bonatti2020consumer,kayaRoy2024repeated} enriches the downstream physical-good market while simplifying the information environment. The closest comparison along these lines is \citet{dilme2025repeated}, who studies repeated trade between a privately informed buyer and a sequence of price-setting sellers who observe noisy records of the buyer's past purchases. Given the applied setting motivating our analysis, we instead abstract away downstream pricing and private values in order to isolate how competition among upstream information brokers affects information provision, surplus creation, and surplus division in search.

Beyond search markets, the persuasion literature also studies the impact of competition in  both static settings \citep{kamenica2017competition, gen17, au20} and dynamic ones \citep{li18, li21, wu23}. These papers consider the canonical, non-transferable utility persuasion setting in which ``senders" design information so as to influence the action taken by a ``receiver." Because monetary payments are not allowed, receiver's welfare is determined by the amount of information senders provide, and the literature characterizes conditions under which competition increases information provision \citep{kamenica2019bayesian}. In our model, however, welfare depends not just on the level of information provision but also on the monetary transfers from the agent to the brokers.\footnote{Our paper also has more remote connections to work on information acquisition in financial markets. For example, \citet{Wittwer2026} studies financial intermediaries who acquire information about asset returns before investing on behalf of clients, and shows that competition among intermediaries can amplify moral hazard and reduce market efficiency. While the settings and mechanisms differ substantially, both her paper and ours share the finding that competition need not improve welfare in markets where information plays a central role.}

Finally, our work also relates to the literature on the design and pricing of information \citep{adm86, adm90, es07, ber15, ber18, bon24, rod24}. However, all of these papers consider the sale of information by a monopolist in a static market. As far as we are aware, our paper is the first to study the competitive sale of information in a dynamic market.

\section{Model}\label{model}

\subsection{Search environment}\label{sec:search environment}
A risk-neutral agent faces an optimal stopping problem: He has a unit demand, and he sequentially searches without recall for a good to consume. In each period, the agent incurs a search cost $k>0$ and samples a good. The good's quality (or match value), denoted by $\theta\in \Theta=[0,1]$, is identically and independently distributed according to an absolutely continuous distribution $F$ with full support on $\Theta$ and a prior mean of $m^\varnothing \coloneqq \int_\Theta \theta dF(\theta)$. We assume  $m^\varnothing >k$ so that the agent has an incentive to search.

If the agent searches without observing the quality of sampled goods, then he stops his search after the first period, yielding an expected payoff of $\underline u_k\coloneqq m^\varnothing-k$. In contrast, if the agent observes the quality of each sampled good perfectly, his payoff is given by $\bar u_k>\underline u_k$, where $\bar u_k$ is characterized as the unique value of $u$ that solves
\[
k=\int^1_{u}(\theta-u)dF(\theta)
\]
as shown in \cite{mcc70}. 

\subsection{Information}
We depart from the classic theory of sequential search by assuming that the agent cannot directly observe the quality of sampled goods. Instead, he acquires information by repeatedly contracting with $n\geq 1$ risk-neutral brokers, with $N\coloneqq \{1,\ldots, n\}$ denoting the set of brokers. Each broker designs and sells an informative signal about the quality of a sampled good. Because all parties are risk neutral, the only payoff-relevant information is the expected quality of a  sampled good. We therefore model a signal simply as a  posterior-mean distribution $G:\Theta\to[0,1]$. For example, a fully informative signal is given by the distribution  $F$ while an uninformative signal is given by the distribution  $G^\varnothing$ that is degenerate on the prior mean. From \cite{bla53}, the set of feasible posterior-mean distributions is equivalent to the set of mean-preserving contractions of $F$, which we denote by $\mathcal{G}(F)$.

\subsection{Timing}
At the beginning of each period $t=1,2,\ldots$ in which the agent is engaged in search, he incurs a search cost of $k$ and samples a good of quality $\theta_t\sim F$. Neither the agent nor the brokers observe $\theta_t$. Each broker $i\in N$ makes a take-it-or-leave-it (TIOLI) offer $o_{it}\coloneqq (p_{it}, G_{it})$ consisting of a price $p_{it}\in\mathbb{R}_+$ and a posterior-mean distribution $G_{it}\in \mathcal{G}(F)$. Let $O_t\coloneqq (o_{it})_{i\in N}$ denote the profile of period-$t$ offers, and let $O_{-it}\coloneqq (o_{jt})_{j\in N\backslash\{i\}}$. In addition to the offers made by the brokers, the agent has access to a null offer $o_{\emptyset t}\coloneqq (p_{\emptyset t}, G_{\emptyset t})$, with $p_{\emptyset t}=0$ and $G_{\emptyset t}=G^\varnothing$ for all $t$.

Given the profile of offers $O_t$, the agent purchases either from  one of the brokers in $N$ or takes the null offer. We denote the agent's purchase decision as $w_t\in W\coloneqq N\cup\{\emptyset\}$. Given  $w_t$, a realization $m_t\sim G_{w_t t}$ is drawn and publicly observed. Finally, the agent decides whether to consume the sampled good and stop searching $(d_t=1)$, in which case the game ends, or to continue searching $(d_t=0)$, in which case the game continues on to $t+1$. 

Let $\mathcal{O}\coloneqq \mathbb{R}_+\times \mathcal{G}(F)$ be the set of all TIOLI offers. Given a profile of offers $O_{t}\in \mathcal{O}^{n}$, purchase decision $w_t\in W$, and stopping decision $d_t\in \{0,1\}$, the agent's ex-post period-$t$ payoff is  $d_t\theta_t-k-p_{w_t t}$, and Broker $i$'s ex-post period-$t$ payoff is  $p_{it}\mathbbm{1}_{[w_t=i]}$.

\subsection{Discussion of modeling assumptions}
Before we proceed with our analysis, let us remark on some of our modeling choices. Our model implicitly assumes that: $(a)$ brokers lack long-term commitment power and thus sell information via spot rather than long-term contracts; $(b)$ the marginal cost of providing more informative signals in the upstream market is  negligible; and $(c)$ the information the agent acquires in the upstream market does not alter the ex-post match value of a downstream good---for example, by affecting its price. While these assumptions are not suitable for every setting, they are  satisfied in the motivating examples we discuss in the introduction.\footnote{For example, \emph{CARFAX} sells vehicle history reports on a per-car basis, and \emph{ProfileXT}---a pre-employment testing agency---charges per candidate screened. Moreover, pre-employment testing agencies typically provide employers with coarse ``pass/fail" reports despite having more granular data on job candidates' test performance, suggesting that the marginal cost of a more informative test is negligible \citep{aut08}.}

Additionally, we assume that the agent acquires information from at most one broker in a given period. An alternative assumption would be to let the agent acquire information from any subset of brokers. This would allow the agent to combine signals in a similar spirit to \cite{kamenica2017competition} and  learn more than what is possible by observing the signal of any one broker. Finally, we assume that signal realizations are publicly observable. An alternative assumption would be that signal realizations are privately observed by the agent. This alternative would necessitate keeping track of a public history for the brokers, a private history for the agent, and each broker's beliefs over the agent's private history. However, as will become clear in \Cref{prelim}, neither of these alternative assumptions changes our main results.

\subsection{Histories, strategies, payoffs, and equilibrium}

A history at the beginning of period $t$ (for which the agent is still engaged in search) contains all past TIOLI offers, purchase decisions, and signal realizations. Let $h^{t}\coloneqq (O_\tau, w_\tau, m_\tau)_{\tau<t}$ denote an arbitrary period-$t$ history, with $h^1=\emptyset$.\footnote{We do not include $\{d_\tau\}_{\tau<t}$ in the history as it is redundant; the game reaches Period $t$ if and only if $d_\tau=0$ for all $\tau<t$.} We denote the set of all period-$t$ histories by $H^{t}$, and the set of all histories by $H\coloneqq \cup_{t\geq 1}H^t$.

A behavioral strategy for Broker $i\in N$ is given by a universally measurable mapping $\sigma_{i}:  H\to \Delta(\mathcal{O})$, which maps each history $h\in H$ into a distribution over the set of TIOLI offers $\mathcal{O}$.\footnote{Formally, let $\Delta(\Theta)$ denote the space of Borel probability measures on $\Theta$, endowed with the weak-* topology. This makes $\Delta(\Theta)$ a Polish space. Since $\mathcal{G}(F)$ is a closed subset of $\Delta(\Theta)$, it too is a Polish space. We endow all Polish spaces with their Borel sigma-algebras, and all product spaces with their product sigma-algebras. Thus, $\mathcal{O}^n\times W\times \Theta$ and $H^t\equiv(\mathcal{O}^n\times W\times \Theta)^{t-1}$ for all $t=2,3,\ldots$ are standard Borel spaces. Finally, $H$ is also a Borel space as it is the countable union of finite-product Borel spaces. A function is universally measurable if it is Borel-measurable under the completion of every Borel measure.} Let $\Sigma_P$ denote the set of all such strategies (all the brokers have the same strategy space). 

A behavioral strategy for the agent is given by a pair of universally measurable mappings $\sigma_{A}\coloneqq (\sigma^\mathcal{w}_A, \sigma^\mathcal{d}_A)$, where the first mapping $\sigma^\mathcal{w}_{A}: H\times \mathcal{O}^{n}\to \Delta(W)$ captures the purchase decision, and the second mapping $\sigma^\mathcal{d}_{A}: H\times \mathcal{O}^{n}\times W\times \Theta\to [0,1]$ captures the search decision. Specifically, following some history $h\in H$ and a profile of offers $O\in\mathcal{O}^n$, the agent purchases a signal from $w\in W$ with probability $\sigma_A^\mathcal{w}(w|h, O)$ and continues his search following some signal realization $m\in \Theta$ with probability $\sigma_A^\mathcal{d}(h, O, w, m)$. Let $\Sigma_A$ denote the set of all such strategies for the agent.

Each strategy profile  $\sigma\coloneqq \big((\sigma_i)_{i\in N}, \sigma_A\big)\in \Sigma_P^n\times \Sigma_A$ induces a distribution over offers, purchase decisions, and ultimately, a stopping time $\tilde T$ which is the (random) time when the agent ends his search. In all equilibrium constructions below, the stopping time is finite almost surely and the displayed expectations are finite. Thus, given a strategy profile $\sigma$ and a history $h^{t}$, Broker $i$'s expected payoff in Period $t$ is given by
\[
\label{eq:broker-payoff}
\tag{1}
\mathbf V_{it}(\sigma|h^{t})=\mathbb{E}_{\sigma(h^t)}\left[\sum_{\tau=t}^{\tilde T}p_{i\tau}\mathbbm{1}_{[w_\tau=i]}\right],
\]
where the expectation is taken with respect to the distribution induced by $\sigma$ starting from history $h^t$. Similarly, the agent's expected payoff is given by 
\[
\label{eq:agent-payoff}
\tag{2}
\mathbf U_t(\sigma|h^t)=\mathbb{E}_{\sigma(h^t)}\left[\theta_{\tilde T} -\sum_{\tau=t}^{\tilde T}\Big(p_{w_\tau \tau}+k\Big)\right].
\]
When considering payoffs in period $t=1$, we simply write $\mathbf V_i(\sigma)$ and $ \mathbf U(\sigma)$ instead of $\mathbf V_{i1}(\sigma|h^1)$ and $\mathbf U_{1}(\sigma|h^1)$, respectively. 

A strategy profile $\sigma$ constitutes an \emph{equilibrium} if it is sequentially rational, i.e., for all  $t\geq 1$ and all $h^t\in H$,
\begin{enumerate}[$(\alph*)$]
 \item $\mathbf V_{it}(\sigma|h^t)\geq \mathbf V_{it}\big(\tilde \sigma_{i},(\sigma_j)_{j\in N\backslash\{i\}}, \sigma_A|h^t)$ for all $\tilde \sigma_{i}\in \Sigma_P$ and for all  $i\in N$,  and 
    \item $\mathbf U_t(\sigma|h^t)\geq \mathbf U_{t}((\sigma_i)_{i\in N}, \tilde \sigma_A|h^t)$ for all $\tilde \sigma_A\in \Sigma_A$.
    \end{enumerate}

Let $\mathcal{E}\subseteq\mathbb{R}^{n+1}_+$ be the set of equilibrium payoff profiles, that is, $y\in \mathcal{E}$ if and only if there exists an equilibrium $\sigma$ such that $y_i=\mathbf V_i(\sigma)$ for each $i\in N$ and  $y_{n+1}=\mathbf U(\sigma)$. In what follows, it will be useful to parameterize the equilibrium payoff set by the number of competing brokers and the search cost. Thus, we write $\mathcal{E}(n,k)$ to denote the equilibrium payoff set of a game between $n$ brokers and an agent who has a search cost of $k$.

\section{Main Results}\label{main results}
In this section, we present the paper's three main results. After presenting these main results, we discuss the organization of their proofs.

Our first main result compares monopolistic and competitive settings based on the total surplus that can be generated in an equilibrium of each setting. Given a set $X$, we write $2^X$ to denote the power set of $X$. Define the operator $\mathcal{W}:\cup_{n\geq 1} 2^{\mathbb{R}^{n+1}_+}\to 2^{\mathbb{R}_+}$ such that, given $n\geq 1$ and some non-empty set $E\subseteq \mathbb{R}_+^{n+1}$, the corresponding total surplus set is given by
\[
\mathcal{W}(E)\coloneqq\left\{\norm{y}: y\in E\right\}.
\]
Most relevantly for our analysis, $\mathcal{W}(\mathcal{E}(n,k))$ denotes the set of total surplus levels attained in equilibrium with market structure $(n,k)$.

\begin{theorem}\label{thm:ewso}
There exists a unique $k^*\in (0,m^\varnothing)$ such that for all $n\geq 1$, 
\begin{enumerate}[$(\roman*)$]
    \item $\mathcal{W}(\mathcal{E}(n,k))=[\underline u_k, \bar u_k]$ for all $k\leq k^*$, and
    \item  $\sup\hspace{.2em}\mathcal{W}(\mathcal{E}(n,k))=\inf\hspace{.2em} \mathcal{W}(\mathcal{E}(1,k))=\bar u_k$ for all $k> k^*$.
\end{enumerate}
\end{theorem}
\medskip

In words, this theorem states that for small enough search costs, any feasible total surplus of the search market---that is, any surplus $u\in[\underline u_k,\bar u_k]$---can arise in equilibrium, regardless of the number of brokers.\footnote{We provide a characterization of the feasible surplus in \Cref{sec:feasible}.} Hence, the equilibrium total surplus set is equivalent under monopolistic and competitive settings. In contrast, for large enough search costs, the \emph{highest} equilibrium total surplus that can be generated in a competitive setting equals the \emph{lowest} equilibrium total surplus that can be generated in a monopolistic setting, and both coincide with the highest feasible total surplus of the search market.

 Of course, with transferable utility, higher total surplus in one equilibrium than another does not imply that the agent is better off in the former. Our second main result compares  monopolistic and competitive settings based on the agent's payoff. Define the operator $\mathcal{A}:\cup_{n\geq 1} 2^{\mathbb{R}^{n+1}_+}\to 2^{\mathbb{R}_+}$ such that, given $n\geq 1$ and some non-empty set $E\subseteq \mathbb{R}_+^{n+1}$, the corresponding agent's payoff set is given by 
 \[
 \mathcal{A}(E)\coloneqq \{y_{n+1}:y\in E\},
 \]
where $\mathcal{A}(\mathcal{E}(n,k))$ denotes the set of agent payoffs attained in equilibrium with market structure $(n,k)$.
 
\begin{theorem}\label{thm:awso}
Let $k^*$ be the threshold in \autoref{thm:ewso}. For all $n\geq 1$, 
\begin{enumerate}[$(\roman*)$]
    \item $\mathcal{A}(\mathcal{E}(n,k))=[\underline u_k, \bar u_k]$ for all $k\leq k^*$, and
    \item  $\inf \hspace{.2em}\mathcal{A}(\mathcal{E}(n,k))\geq \sup \hspace{.2em}\mathcal{A}(\mathcal{E}(1,k))=\underline u_k$ for all $k> k^*$.
\end{enumerate}
\end{theorem}
\medskip

In words, this theorem establishes that for small enough search costs, the agent can attain any share of the feasible surplus in equilibrium, regardless of the number of brokers. Hence, the agent's payoff set is equivalent under competitive and monopolistic settings. In contrast, for large enough search costs, the agent's \emph{lowest} payoff in a competitive setting is no smaller than his \emph{highest} payoff in a monopolistic setting. In fact, under monopoly, the agent receives no more than the lowest feasible surplus.

Taken together, \autoref{thm:ewso} and \autoref{thm:awso} imply that both competition and search costs qualitatively affect the set of equilibria. Our proofs reveal that for $k> k^*$, there is a unique equilibrium payoff under a monopoly: the highest feasible total surplus $\bar u_k$ is generated in equilibrium but the agent captures only the lowest feasible payoff $\underline u_k$ while the remaining surplus $\bar u_k-\underline u_k$ is extracted by the monopolist (see \autoref{thm:monopoly} in \Cref{sec:proofmain}). For $k\leq k^*$, a folk-theorem result obtains---any feasible surplus creation and division is attained in equilibrium---regardless of the number of brokers (see \autoref{thm:competition} in \Cref{sec:proofmain}).

The contrast in these two theorems suggests that there is a tension between efficiency and the agent's welfare: competition at best does not improve efficiency, and at worst yields more inefficient outcomes. However, from the agent's perspective, competition at worst does not hurt him, and at best gives him a higher share of the surplus. The following result formalizes that these welfare comparisons are strict: there exist cost parameters which admit equilibrium outcomes in a competitive setting that are strictly less efficient, but that yield the agent a strictly higher payoff, than any equilibrium outcome in a monopolistic setting.\medskip

\begin{theorem}\label{thm:strict}    
Let $k^*$ be the same threshold as in \autoref{thm:ewso}. There exists a threshold $k^{**}\in (k^*,m^\varnothing)$ and a constant $\delta>0$ such that for each $k\in (k^*, k^{**}]$ and each $n>1$, there exists a non-empty open set $E\subset \mathcal{E}(n,k)$ with $\inf \hspace{.2em}\mathcal{W}\big(\mathcal{E}(1,k)\big)>\sup \hspace{.2em}\mathcal{W}\big(E)+\delta$ and $\inf \hspace{.2em}\mathcal{A}\big(E\big)>\sup\hspace{.2em} \mathcal{A}\big(\mathcal{E}(1,k)\big)+\delta$.
\end{theorem}
\medskip

We make two observations  when it comes to \autoref{thm:strict}. First, our proof of \autoref{thm:strict} reveals that for $n>1$ and $k\in(k^*,k^{**}]$, the subset of  $\mathcal{E}(n,k)$ that yields weakly higher total welfare than $\inf\hspace{.2em}\mathcal{W}(\mathcal{E}(1,k))$ or that yields weakly lower agent payoff than $\sup\hspace{.2em}\mathcal{A}(\mathcal{E}(1,k))$ is a negligible set. Because the set $E$ described in \autoref{thm:strict} is open and non-empty, it is also non-negligible. Therefore, for $k\in(k^*,k^{**}],$ ``almost all'' equilibrium payoff profiles under competition feature strictly lower total welfare (and strictly higher agent payoff) than does any equilibrium payoff profile under monopoly, implying that ``almost all'' equilibrium payoff profiles under competition are inefficient.\footnote{We say a subset $E'\subseteq\mathcal{E}(n, k)$ is a negligible set if $\mathcal{E}(n,k)$ has positive Lebesgue measure in $\mathbb{R}^{n+1}$ while $E'$ has Lebesgue measure zero. Given any $n>1$ and $k\in(k^*, k^{**}]$, the existence of a non-empty open subset of $\mathcal{E}(n,k)$  implies that the equilibrium payoff set has positive measure in $\mathbb{R}^{n+1}$. By contrast, as we show in the proof of \autoref{thm:strict}, the equilibrium payoff profiles that yield weakly higher total welfare than $\inf\hspace{.2em}\mathcal{W}(\mathcal{E}(1,k))$---i.e., the set of Pareto efficient payoff profiles---lie on a hyperplane in $\mathbb{R}^{n+1}$, and thus have zero measure. The same holds for equilibrium payoff profiles that yield weakly lower agent payoff than $\sup\hspace{.2em}\mathcal{A}(\mathcal{E}(1,k))$.}

Second, notice that $\delta$ is selected independently of $n>1$ and $k\in(k^*,k^{**}],$ further implying that there always exists a non-negligible set of equilibrium outcomes that is uniformly bounded away from efficiency under competition. As we will discuss in \Cref{prelim}, the Bertrand outcome in which brokers give away full information for free is efficient, which implies that the set of equilibrium outcomes does not converge to the Bertrand outcome as $n\to \infty$. 

\subsection{Policy implication: breaking up a monopolist}
The comparison of welfare across different market structures is naturally of interest in certain applied and policy settings. Here we offer one such application of our main results: consider a regulator (it) which must decide whether to break up a monopoly in the market for information brokers. Specifically, given the agent's search cost $k\in (0,m^\varnothing)$, the regulator must choose between a market structure $(1,k)$ and $(n,k)$ for some $n>1$. Which market structure should the regulator choose if it is interested in the efficiency of the search market equilibrium? If it is interested in agent's welfare?

\autoref{thm:ewso}-\autoref{thm:strict} find that the regulator's objective matters. If the regulator maximizes total surplus, monopoly is weakly preferred whenever search costs are high; if the planner maximizes the agent’s payoff, competition is weakly preferred.
 
However, a caveat to this general statement is that there may be many equilibria, and therefore, equilibrium selection rules  matter for some search cost parameters. Our upcoming machinery formalizes results under selection rules that treat market structures symmetrically whenever they generate ``the same'' outcomes, either for total welfare or for agent welfare.
 
 Given a market structure $(\hat n, k)$, suppose payoff profile $y^*(\hat n, k)\in\mathcal{E}(\hat n, k)$ is selected. We refer to $y^*$ as the \emph{equilibrium selection rule}. We say that an equilibrium selection rule is \emph{W-symmetric} if, given $n', n''\geq 1$, the selection rule satisfies $\norm{y^*(n',k)}=\norm{y^*(n'',k)}$ whenever $\mathcal{W}\big(\mathcal{E}(n',k)\big)=\mathcal{W}\big(\mathcal{E}(n'',k)\big)$. Similarly, we say an equilibrium selection rule is \emph{A-symmetric} if $y^*_{n'+1}(n',k)=y^*_{n''+1}(n'',k)$ whenever $\mathcal{A}\big(\mathcal{E}(n',k)\big)=\mathcal{A}\big(\mathcal{E}(n'',k)\big)$. For example, $y^*(\hat n, k)\in\argmin_{y\in \mathcal{E}(\hat n, k)} \norm{y}$ is a W-symmetric selection rule. Such a selection rule would be appropriate for a pessimistic regulator who is interested in maximizing total surplus in a ``worst-case" scenario. In this case, the regulator would solve $\max_{\hat n\in\{1,n\}}\norm{y^*(\hat n, k)}\equiv\max_{\hat n\in\{1,n\}}\min_{y\in \mathcal{E}(\hat n, k)}\norm{y}$.

\begin{corollary} Let $k^*$ be the threshold in \autoref{thm:ewso}, and let $n>1$.
    \begin{enumerate}
        \item For any $k>k^*$, a regulator that seeks to maximize the total surplus (weakly) prefers market structure $(1,k)$ over $(n,k)$ under any equilibrium selection rule.  Moreover, for any $k\in (0,m^\varnothing)$, the regulator (weakly) prefers market structure $(1,k)$ if the selection rule is W-symmetric.

        \item For any $k>k^*$, a regulator that seeks to maximize the agent's payoff (weakly) prefers market structure $(n,k)$ over $(1,k)$ under any equilibrium selection rule.  Moreover, for any $k\in (0,m^\varnothing)$, the regulator (weakly) prefers market structure $(n,k)$ if the selection rule is A-symmetric.
    \end{enumerate}
\end{corollary}

\bigskip

The remainder of the paper is organized around proving the main results. \Cref{prelim} introduces tools which make the analysis more tractable. Then, \Cref{sec:proofmain} provides four propositions which characterize the set of equilibrium payoffs as a function of the number of brokers, $n$, and the search cost, $k$, and shows how these propositions imply our main results. \Cref{sec:conclusion} concludes. Throughout, we present additional lemmas which are useful in the analysis; we relegate proofs of supporting lemmas (found in \Cref{prelim}) to the Online Appendix, while the proofs of the propositions in \Cref{sec:proofmain} are housed in \ref{sec:appendix_proposition}.

\section{Constructing Equilibrium Sets}\label{prelim}

A key departure from standard repeated games is that today’s value of information depends on tomorrow’s anticipated payoff. A higher continuation payoff makes the agent more willing to keep searching, which lowers his willingness to pay for information about today’s draw. This feedback is the source of the agent’s endogenous bargaining power. This section characterizes the agent's stopping problem given beliefs about future payoffs, then embeds this search problem into the broader contracting game to construct equilibrium sets.

\subsection{Single-agent search problem}\label{sec:single-agent}

Let us consider a hypothetical search problem in which there are no brokers but the agent observes a realization from a distribution $G\in\mathcal{G}(F)$ in each period for free. This hypothetical setting is instructive in characterizing the lowest and the highest payoffs possible in the persuaded search game.

Let $u_{t+1}\in\mathbb{R}$ be the agent's expected payoff in period $t+1$. The agent's decision in period $t$ depends on the realized posterior mean $m$: he stops his search if  $m>u_{t+1}$, and he continues his search if $m<u_{t+1}$.\footnote{If $m=u_{t+1}$ then the agent is indifferent between stopping his search or continuing. We specify how the agent breaks this indifference in equilibrium in \Cref{sec:self_gen}.} Therefore, the continuation payoff $u_{t+1}$ also serves as the agent's period-$t$ stopping threshold. The agent's expected payoff in period $t$ is thus  
\begin{align*}
u_t&=\int_\Theta(m-u_{t+1})^+dG(m)-k+u_{t+1}\\[6pt]
&=\int_{u_{t+1}}^1(1-G(m))dm-k+u_{t+1},
\end{align*}
where for any $x\in \mathbb{R}$ we define $x^+:=\max\{x,0\}$, and where the last equality follows from integration by parts. Let $c_G:[0,1]\to \mathbb{R}_+$ be defined as
\begin{align*}
 c_G(u)&\coloneqq \int^{1}_{u}(1-G(m))dm.
\end{align*}
For any $G\in\mathcal{G}(F)$, the mapping $u\mapsto c_G(u)$ is continuous, with  $c_G(0)=m^\varnothing$ and $c_G(1)=0$. It is also convex, with a right derivative of $\partial_+ c_G(u)=G(u)-1$ for all $u\in[0,1)$. Thus, $c_G$ is non-increasing, and strictly decreasing over any interval over which it is positive. Moreover, the mapping $G\mapsto c_G$ is  decreasing in mean-preserving contractions: if $G'$ is a mean-preserving contraction of $G''$, then $c_{G'}\leq c_{G''}$ pointwise. Hence, for any $G\in\mathcal{G}(F)$,  $c_{G^\varnothing }\leq c_G\leq c_F$ pointwise.\footnote{See \cite{mek23} for a discussion of these properties.}

The agent's expected payoff in period $t$ can thus be written as 
\[
u_t=\underbrace{c_G(u_{t+1})-k}_{\substack{\text{within-period}\\ \text{payoff}}}\hspace*{1em}+\underbrace{u_{t+1}}_{\substack{\text{continuation}\\ \text{payoff}}}.
\]
The within-period payoff captures the agent's value from consuming the good today net of the flow search cost he incurs. The continuation payoff captures the agent's value from continuing his search. Notice that the within-period payoff depends directly on the continuation payoff, which implies that expected payoffs are not additively separable across time periods.

Because this hypothetical setting is stationary, the agent's expected payoff in each period is the same. In particular, $u_t=u$  for all $t$, where $u$ is a solution to the equation 
\begin{align*}
 \label{eq:1}
\tag{3}
 k&=c_G(u).
\end{align*}
The expression in \eqref{eq:1}, seen as a function of $u$, has a unique solution because $c_G$ is a continuous and decreasing function with $c_G(0)=m^\varnothing>k>0=c_G(1)$. Given a search cost $k\in (0,m^\varnothing)$, let $\mathcal u_k:\mathcal{G}(F)\to [0,1]$ be the mapping capturing the solution to \eqref{eq:1} so that $k=c_G(\mathcal{u}_k(G))$ for all $G\in\mathcal{G}(F)$, as depicted in \autoref{fig:single-agent}.

\begin{figure}[ht]
 \begin{center}
\begin{tikzpicture}[xscale=8, yscale=4]
 \draw [->, help lines,   thick] (0,0) -- (1.1,0);
 \draw [->, help lines,   thick] (0.03,0) --(.03,1.1);
 \node[right=3pt] at (1.1,0) {$u$};
 \node[left=3pt] at (0.1, 1.1) {$c$};

 \draw [] (1,-.02)node[below]{$1$}--(1,.02);
 \draw [] (.03,-.02) node[below]{$0$}--(.03,.02);
 \node[left]at (0.03,.945){$m^\varnothing $};
 \draw [dotted] (.03, 0)--(.03,.945);
 \draw [] (.3,-0.01) node[below]{$m^\varnothing $}--(.3,.02);

\draw[cyan, opacity=.5, domain=0.03:1,smooth,variable=\x, thick] plot ({\x},{(1-\x)^1.65});
\draw[thick, cyan, opacity=.15](.03, .945)--(.3,0)--(1,0);
\draw[thick, red, opacity=.6](.03, .945)--(.3,0)--(1,0);
\draw[ thick] plot [smooth, tension=2] (0.03,.945)to [in=130, out=295](.37, .33) to [in=175, out=310](1,0);

\draw[thick, dashed] (0.03,0.25)node[left]{\footnotesize{$k$}}--(1,.25);

\draw [dotted, thick] (.57,-0.02) node[below]{\footnotesize{$\bar u_k$}}--(.57,.25);
\draw [dotted, thick] (.23,-.02) node[below]{\footnotesize{$\underline u_k$}}--(.23,.25);
\draw [dotted, thick] (.44,-.02) node[below]{\footnotesize{$\mathcal u_k(G)$}}--(.44,.25);
\end{tikzpicture}
\end{center}
\captionsetup{oneside,margin={0cm,0cm},justification=justified, singlelinecheck=false}
\caption{This figure depicts $c_F$ in blue, $c_{G^\varnothing }$ in red, and a function $c_G$ in black for some arbitrary $G\in\mathcal{G}(F)$. The continuation values $\bar u_k$, $\underline u_k$, and $\mathcal{u}_k(\cdot)$ correspond to the points where each respective curve intersects with the dashed line representing search cost $k$. }\label{fig:single-agent}
\end{figure}
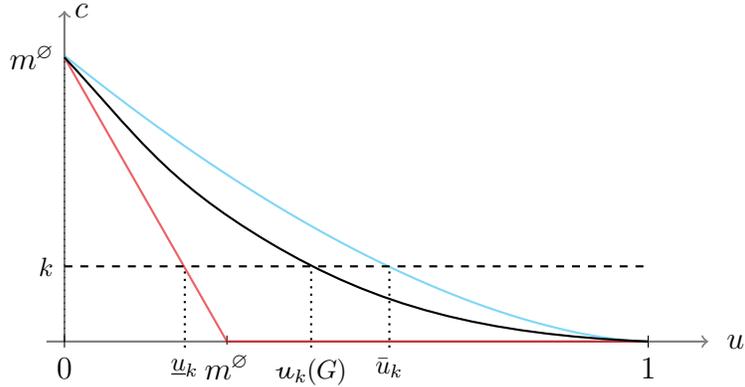

For a fixed $k\in (0, m^\varnothing )$, $\mathcal{u}_k(G')\leq \mathcal{u}_k(G'')$ for any $G', G''\in\mathcal{G}(F)$ such that $G'$ is a mean-preserving contraction of $G''$.\footnote{ This is because  $G''$ is a Blackwell more informative signal than $G'$, and more informative signals are (weakly) more valuable in single-agent decision problems.} Consequently, the lowest payoff the agent can earn is $\underline u_k\coloneqq \mathcal u_k(G^\varnothing )=m^\varnothing -k$, while the highest payoff is $\bar u_k\coloneqq \mathcal u_k(F)>\underline u_k$. We refer to $\bar u_k$ as the \emph{McCall payoff}, to $\underline u_k$ as the \emph{autarky payoff}, and to $\bar u_k-\underline u_k$ as the \emph{full surplus}. The following lemma describes how payoffs change as a function of the agent's search cost.

\begin{lemma}\label{lemma:surplus-compstat-cost}
For any $G\in\mathcal{G}(F)$, the expected payoff $\mathcal{u}_k(G)$ is strictly decreasing in $k$ over the interval $(0,m^\varnothing)$. Furthermore, the full surplus, $\bar u_k-\underline u_k$ is strictly decreasing in $k$ over the interval $(0,m^\varnothing)$.
\end{lemma}

\subsection{Feasible payoffs}\label{sec:feasible}
We now return to the contracting setting between the agent and the brokers. In any equilibrium, each broker's payoff must be non-negative as prices are non-negative, the agent's payoff must weakly exceed his autarky payoff since he can always guarantee this payoff by taking the null offer and immediately ending his search, and the total surplus generated in any equilibrium cannot exceed the McCall payoff as this is the highest possible surplus in the search market. Thus, for $n\geq 1$ brokers and a search cost $k\in (0, m^\varnothing )$, the feasible payoff set is given by 
\[
\mathcal{F}(n,k)\coloneqq \left\{y\in\mathbb{R}^{n+1}_+: y_{n+1}\geq \underline u_k \text{ and } \norm{y}\leq \bar u_k \right\}.
\]
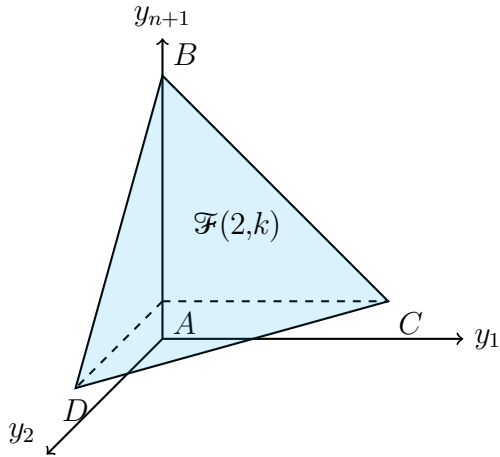
\begin{figure}[ht]
 \begin{center}
\begin{tikzpicture}[scale=.9]
     \coordinate (A) at (0, .5, 0);
     \coordinate (B) at (3, .5, 0);
     \coordinate (C) at (0, 3.5, 0);
     \coordinate (D) at (0,.5, 3);

     \draw[thick] (B) -- (C) --(D) -- cycle;

    \draw[thick, dashed] (A) -- (B);
    \draw[thick, dashed] (A) -- (D);

     \draw[->, thick] (0,0,0) -- (4,0,0) node[right] {$y_1$};
     \draw[->, thick] (0,0,0) -- (0,4,0) node[above] {$y_{n+1}$};
     \draw[->, thick] (0,0,0) -- (0,0,4) node[ left] {$y_2$};

     \node at (A) [below right=-1.5] {$A$};
     \node at (B) [below right=-1.5] {$C$};
     \node at (C) [above right] {$B$};
     \node at (D) [left] {$D$};

     \fill[cyan, opacity=0.15] (B) -- (C) -- (D) -- cycle;
     \node [align=center, above] at (1,1.15) {$\mathcal{F}(2,k)$};
\end{tikzpicture}
\end{center}
\captionsetup{oneside,margin={0cm,0cm},justification=justified, singlelinecheck=false}
\caption{This figure depicts the feasible set $\mathcal{F}(n,k)$ for $n=2$ and $k\in(0,m^\varnothing).$ Vertices $A, B, C$, and $D$ correspond to, respectively,  the autarky payoff profile $(0,0,\underline u_k)$, the ``Bertrand'' payoff profile $(0,0,\bar u_k)$, and the broker optimal payoff profiles $(\bar u_k-\underline u_k,0,\underline u_k)$ and $(0,\bar u_k-\underline u_k,\underline u_k)$.  The face $BDC$ depicts all efficient outcomes, i.e., payoff profiles with $\norm{y}=\bar u_k$.} 
 \label{fig:prism}
\end{figure}

The shaded blue simplex in \autoref{fig:prism} depicts the feasible set $\mathcal{F}(2,k)$ in a duopoly market for some search cost $k\in(0,m^\varnothing)$. For any $n\geq 1$ and $k\in (0,m^\varnothing)$, it is straightforward to see that $\mathcal{F}(n,k)$ is an $(n+1)$-dimensional simplex. Importantly, it has a non-empty interior because $\bar u_k>\underline u_k>0$. Naturally, $\mathcal{E}(n,k)\subseteq \mathcal{F}(n,k)$.

\subsection{Equilibrium outcomes and self-generating sets} \label{sec:self_gen}
In this subsection, we  characterize equilibrium payoff sets by taking a  dynamic programming approach. This is tractable in our model since each good is drawn i.i.d., making all continuation periods identical. Specifically, we study \emph{self-generating} subsets of $\mathcal{F}(n,k)$ by adapting APS to a sequential-move stopping game with no discounting.

To formally describe our approach, we first introduce history-independent \emph{simple strategies}. A simple strategy for Broker $i\in N$, denoted by $\alpha_i\in \Delta( \mathcal O)$, is a randomization over TIOLI offers. Let $\alpha\coloneqq (\alpha_i)_{i\in N}$ and let $\alpha_{-i}\coloneqq (\alpha_j)_{j\in N\backslash\{i\}}$. Similarly, the agent's simple strategy is given by $\beta\coloneqq (\beta^\mathcal{w}, \beta^\mathcal{d})$ where $\beta^\mathcal{w}:\mathcal{O}^n\to \Delta(W)$ is a randomization over purchase decisions, and $\beta^\mathcal{d}:\mathcal{O}^n\times W\times \Theta\to [0,1]$ is the probability with which the agent continues to search. 

We also introduce mappings of the form $\psi:\mathcal{O}^n\times W\times \Theta\to\mathbb{R}^{n+1}$, which we refer to as a \emph{reward function}. The vector $\psi(O, w, m)$ represents the continuation payoff profile if the agent continues searching after the brokers make TIOLI offers $O\in\mathcal{O}^n$, the agent accepts the offer from $w\in W$, and a signal realization of $m\in \Theta$ is observed. In this case, $\psi_i(O,w,m)$ for $i\in N$ represents Broker $i$'s continuation payoff and  $\psi_{n+1}(O,w,m)$ represents the agent's continuation payoff. We denote the image of the mapping $\psi$ by $Im_\psi$.

Given a profile of offers $O$, agent's simple strategy $\beta$, and a reward function $\psi$, Broker $i$'s payoff is given by 
\begin{align*}
V_i(O, \beta, \psi)\coloneqq p_i\beta^\mathcal{w}(i|O)+\sum_{w\in W}\beta^\mathcal{w}(w|O)\int_\Theta \psi_i(O, w, m)  \beta^\mathcal{d}(O, w, m) dG_w(m)
\end{align*}
and the agent's payoff is given by
\begin{align*}
U(O, \beta, \psi)\coloneqq-k&+\sum_{w\in W} \beta^\mathcal{w}(w|O)\left[-p_w+\int_\Theta m\big(1-\beta^\mathcal{d}(O,w,m)\big) dG_w(m)\right]\\
&+\sum_{w\in W} \beta^\mathcal{w}(w|O)\int_\Theta \psi_{n+1}(O, w, m)\beta^\mathcal{d}(O,w,m) dG_w(m).
\end{align*}

\begin{definition}\label{def:self-generation}
A payoff profile $y\in\mathcal{F}(n,k)$ is supported by simple strategies $(\alpha, \beta)$ and a reward function $\psi$ if the following hold:
\begin{enumerate}[$(\alph*)$]
\item $y$ is generated by $(\alpha, \beta, \psi)$: For all $i\in N$, 
\[
y_i=\mathbb{E}_{\alpha}\left[ V_i(O, \beta, \psi) \right],
\]and 
\[
y_{n+1}=\mathbb{E}_{\alpha}\left[ U(O, \beta, \psi) \right],
\]

where $\mathbb{E}_\alpha$ is the expectation over $\mathcal{O}^n$ with respect to the probability induced by $\alpha$.

\item No broker has a profitable deviation: For all $i\in N$ and all  $\tilde o_i\in \mathcal{O}$,
\[
y_i\geq \mathbb{E}_{\alpha_{-i}}\big[ V_i\left(\tilde o_i, O_{-i}, \beta,\psi \right) \big],
\]
where $\mathbb{E}_{\alpha_{-i}}$ is the expectation over $\mathcal{O}^{n-1}$ with respect to the probability induced by $\alpha_{-i}$.
\item $\beta$ is sequentially rational with respect to $\psi$: For all $O\in\mathcal O^n$ and all $\tilde \beta\coloneqq (\tilde \beta^\mathcal{w}, \tilde \beta^\mathcal{d})$,
\[
U(O, \beta, \psi)\geq U(O, \tilde \beta, \psi).
\]
\end{enumerate}
\end{definition}\

Simply put, a feasible payoff profile $y$ is supported by a tuple $(\alpha, \beta,\psi)$ if $y$ is the outcome of the prescribed simple strategies and continuation payoffs, and  if these simple strategies are optimal given the continuation payoffs. Optimality here means that, taking the continuation payoffs as given, the agent's simple strategy is a best response to any profile of TIOLI offers. Furthermore, taking the reward function and the agent's simple strategy as given, each Broker $i$ makes a TIOLI offer that is a best response against the simple strategies of the other brokers. 

\begin{definition}\label{def:gen}
 A set $E\subseteq \mathcal{F}(n,k)$ is self-generating if each $y\in E$ is supported by a tuple $(\alpha,\beta,\psi)$ with $Im_\psi\subseteq E$.
\end{definition}

In other words, a self-generating set $E$  can ``support itself" in the sense that each payoff in $E$ is an outcome of optimal simple strategies when continuation values are chosen from $E$ itself.

We denote the union of all self-generating sets by $\mathcal{S}(n,k)$. By construction, $\mathcal{S}(n,k)$ is itself self-generating, which makes it the largest self-generating set. In the discounted repeated-game environment of APS, there is an equivalence between $\mathcal{S}(n,k)$ and the equilibrium payoff set $\mathcal{E}(n,k)$, with the former offering a recursive characterization of the latter. 

The same conclusion does not apply immediately in our persuaded-search game. First, payoffs in our persuaded search game are not discounted sums of stage-game payoffs. Payoffs are instead accumulated until an endogenous stopping time, and there is no discount factor. Second, the per-period interaction in our model is sequential: brokers first choose TIOLI offers, and the agent then chooses an offer and decides whether to continue after observing the signal realization. We thus require the agent's strategy to be optimal following any profile of offers and signal realization. An analogous within-period sequential rationality condition is not needed in APS because they consider a setting where the per-period interaction is a simultaneous-move game. Finally, APS consider finite action sets, whereas brokers in our persuaded-search game choose from a continuum of offers. This requires a larger class of simple strategies and reward functions,  raising measurability issues of the kind emphasized by \cite{aumann1961borel}. We therefore use the APS recursive idea but separate the classical notion of self-generation from its measurable implementation.

\begin{definition}\label{def:selfgenerate_measurable}
A Borel set $E\subseteq\mathcal{F}(n,k)$ is measurably self-generating if there exist universally measurable functions $(\bm{\alpha}, \bm{\beta}, \bm{\psi})$ with $\bm\alpha:E\to\Delta(\mathcal{O})^n$, $\bm{\beta^\mathcal{w}}:E\times \mathcal{O}^n\to \Delta(W)$, $\bm{\beta^\mathcal{d}}:E\times \mathcal{O}^n\times W\times \Theta\to[0,1]$, and $\bm{\psi}:E\times \mathcal{O}^n\times W\times \Theta\to\mathbb{R}^{n+1}$ such that
\begin{enumerate}[$(\alph*)$]
    \item Each $y\in E$ is supported by $(\bm{\alpha}[y], \bm{\beta}[y], \bm{\psi}[y])$, and \footnote{Here, $\bm{\alpha}[y]=\bm{\alpha}(y)$, $\bm{\beta}[y]\coloneqq (O,w,m)\mapsto\big(\bm{\beta^\mathcal{w}}(y, O),  \bm{\beta^\mathcal{d}}(y, O, w, m)\big)$, and $\bm{\psi}[y]\coloneqq (O,w,m)\mapsto\bm{\psi}(y, O,w,m)$.} 
    \item $Im_{\bm{\psi}} \subseteq E$.
\end{enumerate}
\end{definition}

The main distinction between \autoref{def:gen} and \autoref{def:selfgenerate_measurable} is that for the latter, the supporting tuple is selected through a single collection of universally measurable maps defined on $E$. Let $\mathcal{S}^{m}(n,k)$ be the union of all measurably self-generating sets. Since any measurably self-generating set $E$ is, by definition, self-generating, we have $\mathcal{S}^m(n,k)\subseteq \mathcal{S}(n,k)$. The following lemma relates these two sets to the equilibrium payoff set.

\begin{lemma}\label{lemma:equilibrium-characterization}
For all $n\geq 1$ and $k\in (0, m^\varnothing)$, $\mathcal{S}^m(n,k)\subseteq\mathcal{E}(n,k) \subseteq \mathcal{S}(n,k)$.
\end{lemma}

Therefore, any Borel set $E$ that is measurably self-generating is automatically a subset of the equilibrium payoff set. Furthermore, if we show that $\mathcal{S}(n,k)$ is  itself Borel and measurably self-generating, then the chain of set inclusions in \autoref{lemma:equilibrium-characterization} reduces to a chain of set equalities, giving us an APS-style equivalence in our setting.

While we can use \autoref{lemma:equilibrium-characterization} to recursively characterize (subsets of) the equilibrium payoff set, it is nevertheless a daunting task to check whether some given set $E$ is self-generating. We thus first cut down the difficulty of this task: given a threshold $x\in\Theta$, let $G^x\in\mathcal{G}(F)$ be a posterior-mean distribution given by
\[
 \label{eq:pass-fail-distibution}
\tag{4}
G^{x}(m)=\begin{cases}
     0 &\mbox{if}  ~~~m<\mathbb{E}_F[\theta|\theta\leq  x]  \\
       F(x) &\mbox{if} ~~~ \mathbb{E}_F[\theta|\theta\leq  x]\leq m<\mathbb{E}_F[\theta|\theta\geq  x]  \\
     1 &\mbox{if} ~~~m\geq \mathbb{E}_F[\theta|\theta\geq  x]
     \end{cases}
\]
for any $m\in[0,1]$, and 
\[
 \label{eq:pass-fail-function}
\tag{5}
c_{G^x}(u)=\max\left\{c_{G^\varnothing }(u), c_F(x)+\big(1-F(x)\big)(x-u)\right\}
\]
for any $u\in[0,1].$
We refer to $G^x$ as a \emph{pass-fail} signal (with threshold $x$) because it represents the posterior-mean distribution induced by revealing if a good's quality is above the threshold $x$---in which case the good ``passes"---or if it is below the threshold---in which case the good ``fails." The following lemma shows that any surplus $u\in [\underline u_k, \bar u_k]$ can be generated by some {pass-fail} signal.

\begin{lemma}\label{lemma:pass-fail-sufficiency}
For each $k\in (0,m^\varnothing )$, there exists a continuous and strictly decreasing function $\mathbf x_k:[\underline u_k, \bar u_k]\to \Theta$  with  $\bar u_k=\mathbf x_k(\bar u_k)<\mathbf x_k(\underline u_k)<1$ such that for all $u\in [\underline u_k,\bar u_k]$, $u=\mathcal{u}_k(G^{\mathbf x_k(u)})$.
\end{lemma}

In fact, not only can we generate any total surplus $u \in [\underline u_k, \bar u_k]$, but as our next lemma shows, we can also divide this surplus arbitrarily among the $n$ brokers and the agent (subject to feasibility constraints) in order to generate any payoff profile $y$ with $\norm{y}=u$. To state the lemma, we first introduce a class of TIOLI offers: Given a payoff profile $y\in\mathcal{F}(n,k)$, define the offer $o^*(y)\coloneqq \big(P(y), G^{\mathbf x_k(\norm{y})}\big)\in\mathcal{O}$ where 
\[
\label{eq:price}
\tag{6}
P(y)\coloneqq c_{G^{\mathbf x_k(\norm{y})}}(y_{n+1})-c_{G^{\mathbf x_k(\norm{y})}}(\norm{y}).
\]Let $O^*(y)\in \mathcal{O}^n$ be the profile of TIOLI offers in which each Broker $i\in N$ offers $o^*(y)$, and let $O^*_{-i}(y)$ be the profile in which each Broker $j\in N\backslash\{i\}$ offers $o^*(y)$. 

\begin{lemma}\label{lemma:stationary-implement}
Fix $n\geq 1$ and $k\in (0,m^\varnothing)$. Each $y\in\mathcal{F}(n,k)$ is generated by a tuple $(\alpha, \beta,\psi)$ such that
\begin{enumerate}[$(\roman*)$]
    \item $\alpha_i\big(o^*(y)\big)=1$ for all $i\in N$,
    \item $\beta$ is sequentially rational with respect to $\psi$, and 
    \item $\psi\left(O^*(y), w, m\right)=y$ for all $(w, m)\in W\times \Theta$.
\end{enumerate}
\end{lemma}

\autoref{lemma:stationary-implement} has three important implications. First, for any subset $E\subseteq\mathcal{F}(n,k)$, each payoff profile $y\in E$ is associated with a tuple $(\alpha, \beta, \psi)$ that satisfies Points $(a)$ and $(c)$ of \autoref{def:self-generation}. Hence, to show that $E\subseteq \mathcal{S}(n,k)$, we need only show that the tuple $(\alpha, \beta, \psi)$ also satisfies Points $(b)$ of \autoref{def:self-generation} and that $Im_\psi\subseteq E$. We will use this approach to characterize (subsets of) $\mathcal{S}(n,k)$. Second, \autoref{lemma:stationary-implement} implies that even if the agent could purchase signals from multiple brokers in a given period, he would find it sub-optimal to do so.\footnote{Suppose the agent could buy a signal from more than one broker in a given period. Consider an on-path strategy profile in which all brokers offer the same deterministic pass/fail signal. Combining signals in this case would only give redundant information, so the agent would buy from at most one broker. Now consider an off-path strategy in which some Broker $i\in N$ offers a deviation contract $o_i\neq o^*(y)$. The agent may now want to combine the Broker $i$'s deviation signal with the pass/fail signal of another Broker $j\neq i$. However, in this case, Broker $i$ would have been better off by deviating to sell the combined signal, in which case, the agent would again have no need to buy from multiple brokers. Therefore, any equilibrium we construct in our model remains an equilibrium of this alternative game.} Third, because the reward function described in \autoref{lemma:stationary-implement} does not depend on the realized posterior means, our results on surplus generation and division extend to settings where the agent privately observes signal realizations.

\begin{figure}[ht]
 \begin{center}
\begin{tikzpicture}[xscale=8, yscale=4]
 \draw [->, help lines,   thick] (0,0) -- (1.05,0);
 \draw [->, help lines,   thick] (0.03,0) --(.03,1.05);
 \node[right] at (1.05,0) {$u$};
 \node[] at (0, 1) {$c$};

 \draw [dotted] (.03, 0)--(.03,.945);

  \draw[cyan, opacity=.5, domain=0.03:1,smooth,variable=\x, thick] plot ({\x},{(1-\x)^1.65});
   \draw[thick, cyan, opacity=.15](.03, .945)--(.3,0)--(1,0);
 \draw[thick, red, opacity=.6](.03, .945)--(.3,0)--(1,0);

\draw[domain=0.19:.925,smooth,variable=\x,  thick] plot ({\x},{-.52*\x+.482});
\draw[thick](.0327, .945)--(.195,.378);
\draw[thick](.92,.005)--(.965,.005);

\draw[thick, dashed] (0.03,0.17)node[left]{\footnotesize{$k$}}--(1,.17);
\draw [dotted, thick] (.6,0) node[below]{\footnotesize{$\norm{y}$}}--(.6,.17);
 \draw [dotted, thick] (.35,0) node[below]{\footnotesize{$y_{n+1}$}}--(.35,.3);
 \node [below] at (.92,0){\footnotesize{$m_2$}};
 \draw [dotted, thick] (.195,0) node[below]{\footnotesize{$m_1$}}--(.195,.383);
  \draw [dotted, thick] (.83,0) node[below]{\footnotesize{$\mathbf x_k$}}--(.83,.055);
\draw [decorate,decoration={brace,amplitude=5pt, mirror}]
(0.35,0.3) -- (0.35,.17) node [black, midway,xshift=-0.65cm]
{\footnotesize $P(y)$};
 \end{tikzpicture}
 \captionsetup{oneside,margin={0cm,0cm},justification=justified, singlelinecheck=false}
\caption{This figure depicts $c_F$ in blue, $c_{G^\varnothing }$ in red, and  $c_{G^{\mathbf x_k(\norm{y})}}$ in black. The price $P(y)$ is given by the difference between $c_{G^{\mathbf x_k(\norm{y})}}(y_{n+1})$ and $c_{G^{\mathbf x_k(\norm{y})}}(\norm{y})=k$. The figure also depicts the support of $G^{\mathbf x_k(\norm{y})}$,  $m_1\coloneqq\mathbb{E}_F[\theta|\theta\leq \mathbf x_k(\norm{y})]$ and  $m_2\coloneqq\mathbb{E}_F[\theta|\theta\geq \mathbf x_k(\norm{y})]$. }
 \label{fig:generate-feasible-payoff}
 \end{center}
 \end{figure}
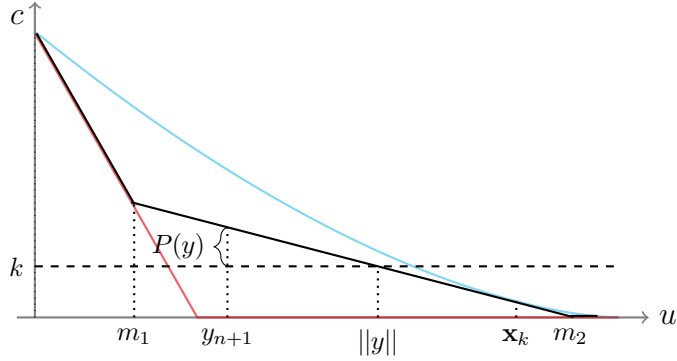

Let us now give some intuition for \autoref{lemma:pass-fail-sufficiency} and \autoref{lemma:stationary-implement} via a geometric argument. In order to generate a payoff profile $y\in\mathcal{F}(n,k)$, we must first show that it is possible to generate a total surplus  of $\norm{y}$ and give the agent only a  payoff of $y_{n+1}$ while the remaining surplus, $\norm{y}-y_{n+1}$, is captured by the brokers (collectively). A priori, it is not clear if this is possible, since an agent who earns a continuation payoff (which we have shown in \Cref{sec:single-agent} is the stopping threshold) of $y_{n+1}$ would have an incentive to stop his search earlier than an agent who earns a higher continuation payoff of $\norm{y}$. 

However, such surplus creation and division is indeed possible. In  \autoref{fig:generate-feasible-payoff}, the curves $c_F$, $c_{G^\varnothing}$, and $c_{G^{\mathbf x_k(\norm{y})}}$ are depicted in blue, red, and black respectively. It is clear from the figure that $c_{G^{\mathbf x_k(\norm{y})}}(\norm{y})=k$, which implies that $\norm{y}=\mathcal{u}_k(G^{\mathbf x_k(\norm{y})})$. In other words, if the agent observes $G^{\mathbf x_k(\norm{y})}$ for free in each period, then his continuation payoff would be $\norm{y}$, and he would optimally stop his search whenever the posterior mean quality exceeds $\norm{y}$. Notice that the pass-fail signal has binary support: a posterior mean of $m_1\coloneqq \mathbb{E}_F\big[\theta|\theta\leq \mathbf x_k(\norm{y})\big]<\norm{y}$ (``fail") and a posterior mean of $m_2\coloneqq\mathbb{E}_F\big[\theta|\theta\geq \mathbf x_k(\norm{y})\big]>\norm{y}$ (``pass"). Thus, an agent who observes $G^{\mathbf x_k(\norm{y})}$ for free in each period would stop only when he sees a ``pass" signal realization, which happens with probability $1-F\big(\mathbf x_k(\norm{y})\big)$. 

Next, suppose the agent's continuation value is lowered to $y_{n+1}\leq\norm{y}$, in which case he would optimally stop his search whenever the posterior mean quality exceeds $y_{n+1}$. However, it is still the case that the agent stops his search only when sees a ``pass," because $m_1<y_{n+1}<m_2$ (see \autoref{lemma:support}). Thus, despite having a higher incentive to quit his search, the agent's probability of stopping remains $1-F\big(\mathbf x_k(\norm{y})\big)$. 

Finally, the price is specified so that an agent who pays $P(y)$ and observes $G^{\mathbf x_k(\norm{y})}$ in each period earns a payoff of $y_{n+1}$. Since a surplus of $\norm{y}$ is generated, the remaining $\norm{y}-y_{n+1}$ is captured by the brokers. When all brokers offer $o^*(y)$, the agent is indifferent across all $n$ TIOLI offers. Hence, he can randomize his purchase decision in such a way that each broker's payoff is $y_j$, thereby generating the payoff profile $y$ as desired.

\section{Proof of main results}\label{sec:proofmain}
In this section, we offer a characterization of the equilibrium payoff set as a function of $n$ and $k$ (Propositions 1-4), and show that our main results (Theorems 1-3) follow as a consequence of this characterization.

Our first proposition fully characterizes the equilibrium payoff set for any search cost in a monopolistic setting. In particular, $\mathcal{E}(1,k)$ has a ``bang-bang" structure: either there is a unique equilibrium payoff or a folk-theorem result obtains. Furthermore, when a unique equilibrium payoff obtains, the equilibrium is efficient with the monopolist extracting the full surplus.

\begin{proposition}\label{thm:monopoly}
There exists a unique $k^*\in (0, m^\varnothing)$ such that $\mathcal{E}(1,k)=\mathcal{F}(1,k)$ for all $k\leq k^*$ and $\mathcal{E}(1,k)=\{(\bar u_k-\underline u_k,\underline u_k)\}$ for all $k>k^*$. 
\end{proposition}

The proof of the proposition, located in the Appendix, first characterizes $\mathcal{S}(1,k)$ by using a fixed-point argument, and we show that it has a bang-bang structure. Then, we show that $\mathcal{S}(1,k)$ is itself measurably self-generating. By \autoref{lemma:equilibrium-characterization}, this implies that $\mathcal{S}(1,k)=\mathcal{E}(1,k)$. Here, we provide a sketch that takes a more constructive approach. For simplicity (and just for the sketch), let us take $\mathcal{S}(1,k)=\mathcal{E}(1,k)$ as a given. 

For a given $k\in(0,m^\varnothing)$, we first define the monopolist's minimax payoff by
\begin{align*}
\nu_k\coloneqq \inf_{\beta, \psi} \hspace*{.2em} \sup_{o\in \mathcal{O}}  \hspace*{.2em} &V_1(o, \beta, \psi)\\[6pt]
\text{s.t.}\hspace*{1em} &  \beta \text{ is sequentially rational with respect to } \psi, \text{ and }\\[6pt]
& Im_\psi\subseteq \mathcal{E}(1,k).
\end{align*}
This minimax payoff diverges from the standard repeated games minimax payoff \citep{fud86} in two key ways. First, unlike repeated simultaneous-move games, here the agent moves after the brokers. Thus, the agent's action in the minimax formulation must be sequentially rational, as captured by the first constraint. Second, because the contracting problem we study is a stopping game, it lacks a well-defined stage game. As a result, $\nu_k$ is not a minimax of the stage-game payoffs but instead depends on the continuation payoffs that can be used to reward or punish in an equilibrium, as captured by the second constraint.

Despite these differences, we may still interpret $\nu_k$ as the payoff the monopolist can guarantee herself when she faces an adversarial (but sequentially rational) agent and unfavorable equilibrium continuation payoffs. We shall now show that any $y\in\mathcal{F}(1,k)$ such that $y_1\geq \nu_k$ is an equilibrium payoff. In other words, the equilibrium payoff set is characterized by the convex hull of three payoff profiles: $(\bar u_k-\underline u_k, \underline u_k)$, $(\nu_k,\underline u_k)$, and $(\nu_k, \bar u_k-\nu_k)$, which is a triangle as depicted by the shaded red set in \autoref{fig:lemma 6}.

\begin{figure}[ht]
 \begin{center}
\begin{tikzpicture}[scale=.8]
[xscale=1.75, yscale=1.5]
    \draw[->, thick] (0,0) -- (4,0) node[right] {$y_1$};
    \draw[->, thick] (0,0) -- (0,4) node[above] {$y_{n+1}$};

    \draw[thick] (0,.5) node[left] {$\underline u_k$} -- (3,.5);
    \draw[thick] (3,.5) -- (0,3.5) node[left] {$\bar u_k$};
    \draw[dotted] (3,0) node[below] {$\bar u_k - \underline u_k$} -- (3,0.4);
    \draw[thick, cyan, opacity=.15] (1,0) -- (1,4);
     \draw[thick, red, opacity=.4] (1,0)  -- (1,4);
      \draw[thick] (1,0) node[below] {$\nu_k$} -- (1,0);
    
     \draw[fill=cyan, opacity=.15] (0,.5) -- (0,3.5) -- (3,.5) --cycle;
         \draw[fill=red, opacity=.4] (3,.5) -- (1,2.5) -- (1,.5) -- cycle;

\end{tikzpicture}
\end{center}
\captionsetup{oneside,margin={0cm,0cm},justification=justified, singlelinecheck=false}
\caption{The shaded simplex (i.e. the union of the red and blue areas) is the feasible set $\mathcal{F}(1,k)$, and the red area is the closure of a hypothetical equilibrium payoff set $\mathcal{E}(1,k)$ for an arbitrary value $\nu_k$.} 
 \label{fig:lemma 6}
\end{figure}

Our approach is constructive in that for each payoff profile in the equilibrium payoff set, we provide a strategy profile that supports it. Of course, we make no claim that each outcome in the equilibrium payoff set must be supported by the strategy profile we construct. Nevertheless, the profiles we construct are qualitatively similar to one another, allowing for a ``recipe'' on achieving any payoff profile in the equilibrium payoff set.
 
To that end, how can we obtain $\hat y$ in the red triangle of \autoref{fig:lemma 6} as an equilibrium payoff profile? By \autoref{lemma:stationary-implement}, we can focus on stationary on-path strategies: each period, the monopolist proposes the TIOLI offer $o^*(\hat y)$, the agent accepts, and the continuation payoff remains $\hat y$. Thus, the agent stops searching if he observes a posterior mean $m > \hat y_{n+1}$ and continues otherwise. 

Any deviations by the agent following a proposal of $o^*(\hat y)$ are ignored. In contrast, if the monopolist deviates from proposing $o^*(\hat y)$, our construction maximally punishes her, the scope of which depends on the value of $\nu_k$. If the agent rejects the deviation offer, play transitions to a phase where the agent is rewarded with a continuation payoff profile $y^r = (\nu_k, \bar u_k - \nu_k)$ (``$r$" for reject). If the agent accepts the deviation offer, both are punished and the continuation payoff profile becomes $y^a = (\nu_k, \underline u_k)$ (``$a$" for accept). The continuation payoff profiles following deviation, $y^r$ and $y^a$, are depicted in \autoref{fig:construct-onpath}.

Of course, $y^r$ and $y^a$ themselves must also be equilibrium payoff profiles. To support $y^r$, we again use stationary on-path strategies: the monopolist repeatedly proposes $o^*(y^r)$. If she deviates and the agent rejects, play remains at $y^r$; if the agent accepts, play transitions to $y^a$ (see \autoref{fig:construct-offpath-good}). Similarly, to support $y^a$, the monopolist proposes $o^*(y^a)$ each period; if he deviates and the agent rejects, play transitions to $y^r$, while if the agent accepts, play remains at $y^a$ (see \autoref{fig:construct-offpath-bad}).

\begin{figure}[ht]
\begin{center}
\hspace*{-2em}
\captionsetup[subfigure]{oneside,margin={1cm,-1cm}}
\begin{subfigure}[t]{0.3\textwidth}
\centering
\begin{tikzpicture}[scale=.8]
    \draw[->, thick] (0,0) -- (4,0) node[right] {$y_1$};
    \draw[->, thick] (0,0) -- (0,4) node[above] {$y_{n+1}$};

    \draw[thick] (0,.5) node[left] {$\underline u_k$} -- (3,.5);
    \draw[thick] (3,.5) -- (0,3.5) node[left] {$\bar u_k$};
    \draw[dotted] (3,0) node[below] {$\bar u_k - \underline u_k$} -- (3,0.4);
    \draw[thick, cyan, opacity=.15] (1,0) -- (1,4);
     \draw[thick, red, opacity=.4] (1,0)  -- (1,4);
      \draw[thick] (1,0) node[below] {$\nu_k$} -- (1,0);
    
     \draw[fill=cyan, opacity=.15] (0,.5) -- (0,3.5) -- (3,.5) --cycle;
         \draw[fill=red, opacity=.4] (3,.5) -- (1,2.5) -- (1,.5) -- cycle;
   \filldraw[black] (1.35,1.35) circle (1pt) node[black, right]{$\hat y$};
    \filldraw[black] (1,2.5) circle (1pt) node[black, above right=-.25em]{$y^r$};
    \filldraw[black] (1,.5) circle (1pt) node[black, below right=-.25em]{$y^a$};
 \draw[->,thick]  (1.35,1.35) to  [looseness=2, out= 150, in=200](.98,2.5)  ;
  \draw[->,ultra thick,dash dot] (1.35,1.35) to [looseness=2, out= 200, in=150](.98,.57);
\end{tikzpicture}
\caption{Deviation from $o^*(\hat y)$}
\label{fig:construct-onpath}
\end{subfigure}
\hspace{2em}
\begin{subfigure}[t]{0.3\textwidth}
\centering
\begin{tikzpicture}[scale=.8]
    \draw[->, thick] (0,0) -- (4,0) node[right] {$y_1$};
    \draw[->, thick] (0,0) -- (0,4) node[above] {$y_{n+1}$};

    \draw[thick] (0,.5) node[left] {$\underline u_k$} -- (3,.5);
    \draw[thick] (3,.5) -- (0,3.5) node[left] {$\bar u_k$};
    \draw[dotted] (3,0) node[below] {$\bar u_k - \underline u_k$} -- (3,0.4);
    \draw[thick, cyan, opacity=.15] (1,0) -- (1,4);
     \draw[thick, red, opacity=.4] (1,0)  -- (1,4);
      \draw[thick] (1,0) node[below] {$\nu_k$} -- (1,0);
    
     \draw[fill=cyan, opacity=.15] (0,.5) -- (0,3.5) -- (3,.5) --cycle;
         \draw[fill=red, opacity=.4] (3,.5) -- (1,2.5) -- (1,.5) -- cycle;
    \filldraw[black] (1,2.5) circle (1pt) node[black, above right=-.25em]{$y^r$};
    \filldraw[black] (1,.5) circle (1pt) node[black, below right=-.25em]{$y^a$};
   \draw[->,thick]  (.98,2.5) to [looseness=45, out= 110, in=5] (1.1,2.45)  ;
   \draw[->,ultra thick, dash dot] (1,2.5) to[looseness=1, out= 200, in=150] (.98,.57);
\end{tikzpicture}
\caption{Deviation from $o^*(y^r)$}
\label{fig:construct-offpath-good}
\end{subfigure}
\hspace*{2em}
\begin{subfigure}[t]{0.3\textwidth}
\centering
\begin{tikzpicture}[scale=.8]
    \draw[->, thick] (0,0) -- (4,0) node[right] {$y_1$};
    \draw[->, thick] (0,0) -- (0,4) node[above] {$y_{n+1}$};

    \draw[thick] (0,.5) node[left] {$\underline u_k$} -- (3,.5);
    \draw[thick] (3,.5) -- (0,3.5) node[left] {$\bar u_k$};
    \draw[dotted] (3,0) node[below] {$\bar u_k - \underline u_k$} -- (3,0.4);
    \draw[thick, cyan, opacity=.15] (1,0) -- (1,4);
     \draw[thick, red, opacity=.4] (1,0)  -- (1,4);
      \draw[thick] (1,0) node[below] {$\nu_k$} -- (1,0);
    
     \draw[fill=cyan, opacity=.15] (0,.5) -- (0,3.5) -- (3,.5) --cycle;
         \draw[fill=red, opacity=.4] (3,.5) -- (1,2.5) -- (1,.5) -- cycle;
    \filldraw[black] (1,2.5) circle (1pt) node[black, above right=-.25em]{$y^r$};
    \filldraw[black] (1,.5) circle (1pt) node[black, below right=-.25em]{$y^a$};

    \draw[->,ultra thick, dash dot]  (1,.49) to [looseness=45, out= 240, in=125] (.98,.57)  ;
   \draw[->,thick] (1,.5) to  [looseness=1, out= 45, in=295](1.02,2.4);
\end{tikzpicture}
\caption{Deviation from $o^*(y^a)$.}
\label{fig:construct-offpath-bad}
\end{subfigure}
\end{center}
 \captionsetup{oneside,margin={0cm,0cm},justification=justified, singlelinecheck=false}
\caption{This figure illustrates the change in continuation payoff profiles following a monopolist's deviation to an off-path TIOLI offer. The solid and dashed lines indicate the change if the agent rejects or accepts the deviation offer, respectively.}
\label{fig:construct}
\end{figure}

Notice that when the equilibrium payoff profile is either  $y^r$ or $y^a$,  $\nu_k$ is the monopolist's on-path payoff as well as her punishment continuation payoff if she deviates. This dual role of $\nu_k$ is key to the fixed-point characterization of the equilibrium payoff set in our proof in the Appendix. This approach is different from the standard APS construction, which defines a fixed-point operator over subsets of $\mathcal{F}(1,k)$, while our approach seeks for a single-dimensional fixed-point $\nu_k$ over an interval.

To understand the bang-bang structure, consider the construction in \autoref{fig:construct} for an arbitrary $\nu_k > 0$. The agent can credibly discipline the monopolist only if the continuation payoff from rejecting a deviation ($y^r_{n+1}$) sufficiently exceeds the payoff from accepting it ($y^a_{n+1}$). Graphically, the agent's ``credibility'' is measured by the vertical gap between $y^r$ and $y^a$ at $\nu_k$, which is decreasing in $\nu_k$. Hence, if the agent’s threat is credible at some $\nu_k > 0$, it is strictly more credible at $\nu'_k=0$; and conversely, if his threat is non-credible at $\nu_k=0$, then it is non-credible for any other $\nu_k$. This bang-bang nature of credibility forces the solution to one of the corners: either $\nu_k=0$, in which case the entire feasible set can be supported in equilibrium, or $\nu_k=\bar u_k-\underline u_k$, in which case there is a unique equilibrium outcome that features full surplus extraction by the monopolist.

We present a graphical depiction of the equilibrium set for different values of $k$ in \autoref{fig:lemma 7}.  Intuitively, the agent has little bargaining power in a monopolistic setting: the singular broker is the only source of information, and she moves first in each period by making a TIOLI offer. The agent can threaten to reject any offer that does not guarantee him a large enough share of the surplus, but such threats are not credible unless they are sequentially rational. In particular, such a threat is credible only if the agent is willing to forgo an informative signal about the currently sampled good in exchange for a high continuation payoff.

\begin{figure}[ht]
\begin{center}
\hspace*{-5em}
\captionsetup[subfigure]{oneside,margin={1cm,-1cm}}\begin{subfigure}[t]{0.4\textwidth}
\centering
\begin{tikzpicture}[scale=.8]
    \draw[->, thick] (0,0) -- (4,0) node[right] {$y_1$};
    \draw[->, thick] (0,0) -- (0,4) node[above] {$y_{n+1}$};

    \draw[thick] (0,.2) node[left] {$\underline u_{k'}$} -- (2,.2);
    \draw[thick] (2,.2) -- (0,2) node[left] {$\bar u_{k'}$};
    \draw[dotted] (2,0) node[below] {$\bar u_{k'} - \underline u_{k'}$} -- (2,.2);
    
    \draw[dashed, fill=cyan, opacity=.15] (0,.2) -- (2,.2) -- (0,2) -- cycle;

         \fill[cyan, opacity =.15] (2, 0.2) circle (3pt);
    \fill[red, opacity =.4] (2, 0.2) circle (3pt);

\end{tikzpicture}
\hspace*{-18mm}\caption{Equilibrium set for $k'> k^*$}
\label{fig:lemma 7_a}
\end{subfigure} 
\hspace*{4em}
\captionsetup[subfigure]{oneside,margin={5cm,-6cm},justification=centering} 
\begin{subfigure}[t]{0.3\textwidth}
\centering
\raisebox{0cm}{ 
\begin{tikzpicture}[scale=1]
[xscale=1.75, yscale=1.5]
    \draw[->, thick] (0,0) -- (4,0) node[right] {$y_1$};
    \draw[->, thick] (0,0) -- (0,4) node[above] {$y_{n+1}$};

    \draw[thick] (0,.5) node[left] {$\underline u_{k''}$} -- (3,.5);
    \draw[thick] (3,.5) -- (0,3.5) node[left] {$\bar u_{k''}$};
    \draw[dotted] (3,0) node[below] {$\bar u_{k''} - \underline u_{k''}$} -- (3,0.4);
    
        \draw[fill=cyan, opacity=.15] (3,.5) -- (0,3.5) -- (0,.5) -- cycle;
    \draw[fill=red, opacity=.4] (3,.5) -- (0,3.5) -- (0,.5) -- cycle;

\end{tikzpicture}
}
\vspace*{-7mm}\caption{Equilibrium set for $k''\leq k^*$}
\label{fig:lemma 7_b}
\end{subfigure}
 \captionsetup{justification=justified,singlelinecheck=false}
 \caption{Panel (a) presents the feasible payoff set $\mathcal{F}(1,k')$ for $k'>k^*$, shaded in blue, and the equilibrium payoff set $\mathcal{E}(1,k')$ which is given by the unique broker optimal outcome $(\bar u_{k'}-\underline u_{k'},\underline u_{k'})$, shaded in red. Panel (b) presents the equilibrium payoff set $\mathcal{E}(1,k'')$ for $k''\leq k^*$, shaded in red, which is equal to the entire feasible set $\mathcal{F}(1,k'').$ Note that because $k''<k',$ $\mathcal{F}(1,k'')\neq \mathcal{F}(1,k')$ (see \autoref{lemma:surplus-compstat-cost} for details).}
\label{fig:lemma 7}
\end{center}
\end{figure}
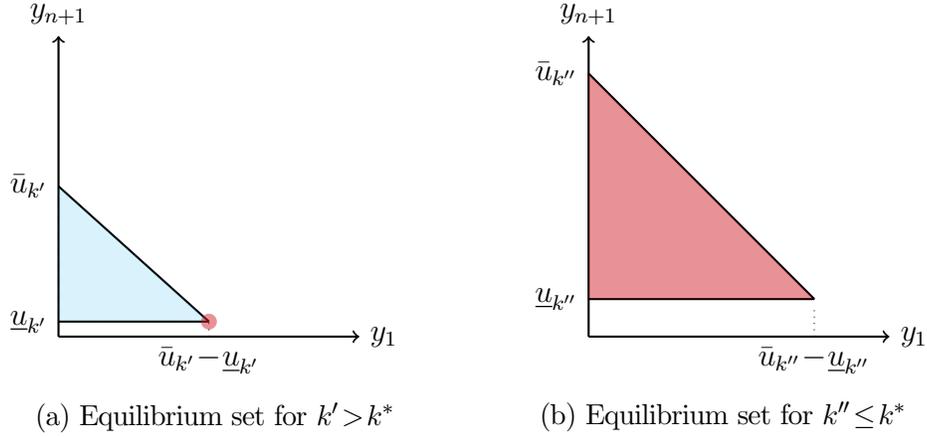


When $k$ is large (i.e., $k>k^*$), even a continuation payoff of $\bar u_k$, which is the highest continuation payoff the agent can earn, is still too low to compensate the agent for forgoing an informative signal about the currently sampled good. In this case, the agent's threat is not credible, leaving him in a weak bargaining position relative to the broker. Consequently, there is a unique equilibrium payoff that features full surplus extraction by the broker. 

Conversely, when $k$ is low (i.e., $k\leq k^*$), $\bar u_k$ is large enough to compensate the agent for forgoing an informative signal, making his threats credible. In this case, any surplus split between the agent and the broker can be sustained in equilibrium, which yields a folk-theorem result. The value $k^*$ is the unique value that balances the loss from forgoing information today with the gain from a high continuation payoff.

How does the above intuition change in a competitive setting? When  $n>1$, the agent can now acquire information from multiple sources, which puts him in a stronger bargaining position. Specifically, rejecting the TIOLI offer from one broker does not entail altogether forgoing the opportunity to learn about the currently sampled good. Hence, the threats the agent uses in a monopolistic setting become even more effective in securing him a higher payoff in a competitive setting. This improved efficacy of the agent's threats has two implications: $(i)$ There exists a Bertrand equilibrium outcome in which at least two brokers offer full information to the agent for free in each period, and $(ii)$ when a folk-theorem result obtains in a monopolistic setting, it also obtains in a competitive setting.

The next two propositions prove these two claims. With some abuse of notation, for each $k\in (0,m^\varnothing)$ and $n\geq 1$, let $y^A\in\mathcal{F}(n,k)$ denote the \emph{autarky payoff profile} where $y_i^A=0$ for all $i\in N$ and $y^A_{n+1}=\underline u_k$, and let $y^B\in\mathcal{F}(n,k)$ denote the \emph{Bertrand payoff profile} where $y_i^B=0$ for all $i\in N$ and $y^B_{n+1}=\bar u_k$.

\begin{proposition}\label{thm:bertrand}
For each $k\in (0, m^\varnothing)$ and $n>1$, the Bertrand payoff profile $y^B\in\mathcal{E}(n,k)$.  
\end{proposition}

\begin{proposition}\label{thm:competition}
Let $k^*$ be the same threshold as in \autoref{thm:monopoly}. For all $k\leq  k^*$ and all $n>1$, $\mathcal{E}(n,k)=\mathcal{F}(n,k)$.
\end{proposition}

We have thus far shown the agent's increased bargaining power when $k\leq k^*$ gives rise to a folk-theorem result, regardless of the underlying market structure. The following result considers the effects of the agent's bargaining power when the search cost satisfies $k>k^*$ and $n>1$. Even when a monopolistic setting sustains a unique equilibrium payoff, a competitive setting may sustain uncountably many equilibrium outcomes. 

To state our result, let us first define the following subset of feasible payoffs: given a constant $\epsilon>0$, let
\[
\mathcal{F}^\epsilon(n,k)\coloneqq \mathcal{F}(n,k)\,\cap\,\{y\in\mathbb{R}_+^{n+1}: \norm{y}-y_{n+1}\leq \epsilon\},
\]
which is a polytope. Because $\mathcal{F}(n,k)$ has a non-empty interior, notice that for any $\epsilon >0$, the subset $\mathcal{F}^\epsilon(n,k)$ also has a non-empty interior. Moreover, notice that both the autarky payoff profile $y^A$ and the Bertrand payoff profile $y^B$ satisfy $\norm{y^A}-y^A_{n+1}=\norm{y^B}-y^B_{n+1}=0$, and therefore, $y^A,y^B\in\mathcal{F}^\epsilon(n,k)$ for all $\epsilon > 0$. 

\begin{proposition}\label{thm:multiple}
Let $k^*$ be the same threshold as in \autoref{thm:monopoly}. There exists a threshold $k^{**}\in (k^*,m^\varnothing)$ and a constant $\epsilon>0$  such that for each $k\leq k^{**}$ and each $n>1$, $\mathcal{F}^\epsilon(n,k)\subseteq\mathcal{E}(n,k)$.
\end{proposition}

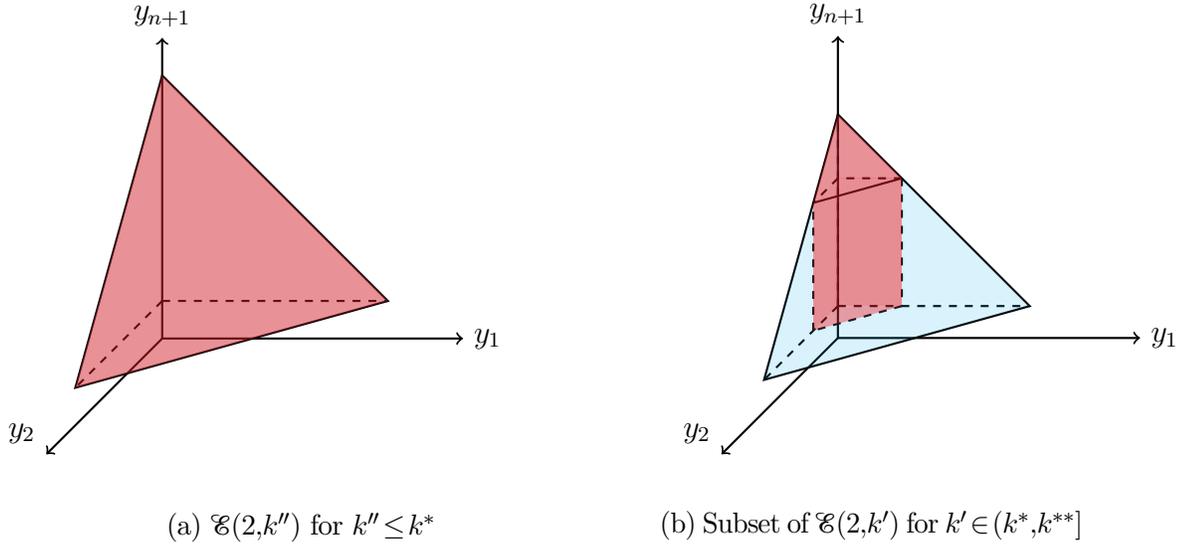
\begin{figure}[ht]
\begin{center}
\hspace*{-5.5em}
\captionsetup[subfigure]{oneside,margin={1cm,-1cm}}\begin{subfigure}[t]{0.4\textwidth}
\centering
\begin{tikzpicture}[scale=1]
     \coordinate (A) at (0, .5, 0);
     \coordinate (B) at (3, .5, 0);
     \coordinate (C) at (0, 3.5, 0);
     \coordinate (D) at (0,.5, 3);

     \draw[thick] (B) -- (C) --(D) -- cycle;

    \draw[thick, dashed] (A) -- (B);
    \draw[thick, dashed] (A) -- (D);

     \draw[->, thick] (0,0,0) -- (4,0,0) node[right] {$y_1$};
     \draw[->, thick] (0,0,0) -- (0,4,0) node[above] {$y_{n+1}$};
     \draw[->, thick] (0,0,0) -- (0,0,4) node[above left] {$y_2$};


      \fill[cyan, opacity=0.15] (B) -- (C) -- (D) -- cycle;
    \fill[red, opacity=0.4] (B) -- (C) -- (D) -- cycle;
\end{tikzpicture}
\hspace*{-18mm}\captionsetup{oneside,margin={.5cm,0cm}}\caption{$\mathcal{E}(2,k'')$ for $k''\leq k^*$}
\label{fig:comp set a}
\end{subfigure} 
\hspace*{6em}
\captionsetup[subfigure]{oneside,margin={5cm,-6cm},justification=centering} 
\begin{subfigure}[t]{0.33\textwidth}
\centering
\raisebox{0cm}{ 
\begin{tikzpicture}[scale=.85]
    \coordinate (A) at (0, .5, 0);
    \coordinate (B) at (3, .5, 0);
    \coordinate (C) at (0, 3.5, 0);
    \coordinate (D) at (0, .5, 3);
    \coordinate (E) at (1, .5, 0);
    \coordinate (F) at (0, .5, 1);
    \coordinate (A2) at (0, 2.5, 0);  
    \coordinate (E2) at (1, 2.5, 0);  
    \coordinate (F2) at (0, 2.5, 1);  

    \draw[thick] (B) -- (C) -- (D) -- cycle;

    \draw[thick, dashed] (A) -- (B);
    \draw[thick, dashed] (A) -- (D);
    \draw[thick, dashed] (A) -- (C);
     \draw[thick] (E2) -- (F2);
        \draw[thick, dashed] (E2) -- (A2);
           \draw[thick, dashed] (A2) -- (F2);
              \draw[thick, dashed] (E) -- (F);
                 \draw[thick, dashed] (E2) -- (E);
                    \draw[thick, dashed] (F) -- (F2);

    \draw[->, thick] (0,0,0) -- (4.7272,0,0) node[right] {$y_1$};
    \draw[->, thick] (0,0,0) -- (0,4.7272,0) node[above] {$y_{n+1}$};
    \draw[->, thick] (0,0,0) -- (0,0,4.7272) node[above left] {$y_2$};


    \fill[cyan, opacity=0.15] (B) -- (C) -- (D) -- cycle;

    \fill[red, opacity=0.15] (A) -- (E) -- (F) -- cycle;  
    \fill[red, opacity=0.23] (A2) -- (E2) -- (F2) -- cycle;  
    \fill[red, opacity=0.15] (A) -- (A2) -- (F2) -- (F) -- cycle;  
    \fill[red, opacity=0.15] (E) -- (E2) -- (A2) -- (A) -- cycle;  
    \fill[red, opacity=0.15] (F) -- (F2) -- (E2) -- (E) -- cycle;  

    \fill[red, opacity=0.1] (A2) -- (E2) -- (E) -- (A) -- cycle;  
    \fill[red, opacity=0.1] (F2) -- (E2) -- (E) -- (F) -- cycle;  
    \fill[red, opacity=0.1] (A2) -- (F2) -- (F) -- (A) -- cycle;  


    \fill[red, opacity=0.4] (C) -- (F2) -- (A2) -- cycle;

      \fill[red, opacity=0.4] (C) -- (E2) -- (A2) -- cycle;
      \end{tikzpicture}
}
\vspace*{0.01mm}\captionsetup{oneside,margin={-8cm,-10cm}}
\caption{Subset of $\mathcal{E}(2,k')$ for $k'\in(k^*,k^{**}]$}
\label{fig:comp set b}
\end{subfigure}
 \captionsetup{justification=justified,singlelinecheck=false}
 \caption{Panel (a) presents the equilibrium payoff set $\mathcal{E}(2,k'')$ for $k''\leq k^*$, shaded in red, which is equal to the entire feasible payoff set $\mathcal{F}(2,k'').$  Panel (b) presents the feasible payoff set $\mathcal{F}(2,k')$ for $k'\in(k^*,k^{**}]$, shaded in blue, and a subset of the equilibrium payoff set $\mathcal{F}^\epsilon(2,k')\subseteq\mathcal{E}(2,k')$ with a non-empty interior, shaded in red. Note that because $k''<k',$ $\mathcal{F}(2,k'')\neq \mathcal{F}(2,k')$ (see \autoref{lemma:surplus-compstat-cost} for details).}
\label{fig:comp set}
\end{center}
\end{figure}

\autoref{fig:comp set} presents graphical depictions of the characterizations offered in \autoref{thm:competition} and \autoref{thm:multiple} for $n=2$. For some $k''\leq k^*,$ Panel (a) displays $\mathcal{E}(2,k'')=\mathcal{F}(2,k'')$ as the shaded red area. For some $k'\in(k^*,k^{**}],$ Panel (b) displays a  subset $\mathcal{F}^\epsilon(2, k')\subseteq \mathcal{E}(2,k')$ (for some $\epsilon >0$) with a non-empty interior as the shaded red area. 

\bigskip

We now use the four propositions to establish our main results (Theorems 1-3). From \autoref{thm:monopoly} and \autoref{thm:competition}, we conclude that when $k\leq k^*$, the total surplus set and the agent's payoff set are given by $\mathcal{W}\big(\mathcal{E}(n,k)\big)=\mathcal{W}\big(\mathcal{F}(n,k)\big)=\mathcal{A}\big(\mathcal{E}(n,k)\big)=\mathcal{A}\big(\mathcal{F}(n,k)\big)=[\underline u_k, \bar u_k]$ for all $n\geq 1$, which establishes Point $(i)$ in both  \autoref{thm:ewso} and \autoref{thm:awso}.

In contrast, when $k> k^*$, \autoref{thm:monopoly} implies that $\mathcal{W}\big(\mathcal{E}(1,k)\big)=\{\bar u_k\}$, while \autoref{thm:bertrand} implies that $\sup \mathcal{W}\big(\mathcal{E}(n,k)\big)=\bar u_k$ for any $n>1$, which establishes Point $(ii)$ of \autoref{thm:ewso}. Similarly, when $k> k^*$, \autoref{thm:monopoly} implies that $\mathcal{A}\big(\mathcal{E}(1,k)\big)=\{\underline u_k\}$, whereas for any $n\geq 1$, $\inf \hspace{.2em}\mathcal{A}\big(\mathcal{E}(n,k)\big)\geq \inf\hspace{.2em} \mathcal{A}\big(\mathcal{F}(n,k)\big)=\underline u_k$, which establishes Point $(ii)$ of  \autoref{thm:awso}.

Finally, \autoref{thm:multiple} implies that there exists a constant $\epsilon>0$ such that for all $n>1$ and all $k\in(k^*, k^{**}]$, there exists a non-empty and open subset $E\subset \mathcal{F}^\epsilon(n,k)\subseteq \mathcal{E}(n,k)$. Furthermore, we can pick $E$ such that for all $y\in E$,  $\norm{y}+\epsilon<\bar u_k =\inf \hspace{.2em}\mathcal{W}\big(\mathcal{E}(1,k)\big)$, and  $y_{n+1}-\epsilon>\underline u_k=\sup  \hspace{.2em}\mathcal{A}\big(\mathcal{E}(1,k)\big)$, which establishes \autoref{thm:strict}.

\section{Conclusion}\label{sec:conclusion}

This paper studies the role of competition in persuaded search markets, where an agent engaged in sequential search acquires information on the quality of sampled goods by repeatedly purchasing signals from a finite set of information brokers. We show that greater competition among brokers expands the set of equilibria in ways that raise the agent’s payoff but, perhaps counterintuitively, reduces total surplus. 

These results stem from the interaction between competition and bargaining power in brokered search markets, where the agent's continuation value directly affects both his probability of terminating his search and also his willingness to pay for information. Under a monopoly, the agent has bargaining power if and only if he can credibly reject a lackluster offer from the broker, which in turn holds if and only if he is willing to forego an informative signal today in order to induce more informative or cheaper signals in future periods. Under competition, however, rejecting any one broker’s offer does not require foregoing information in the current period; the agent can instead acquire information from another competing broker. Competition thus provides the agent with an inside option, strengthening his bargaining position relative to the monopolistic case. This allows the agent to capture more of the surplus, but also sustains inefficient outcomes.

These findings highlight that market structure affects the creation of surplus, not merely its division. In particular, we identify an intensive-margin channel through which competition can lower total surplus by weakening brokers' incentives to provide efficient information. Of course, in richer environments with hidden agent heterogeneity, competition may also generate efficiency gains on the extensive margin by expanding participation in the market for information. Nevertheless, the intensive-margin channel we identify would still operate as a  countervailing force, and its welfare consequences can be severe: for certain search costs, our model reveals that the competitive setting supports a ``no-trade'' equilibrium, precipitating a complete information market breakdown, while monopoly does not.  

From an applied perspective, our welfare results highlight a non-trivial policy tradeoff: whether a regulator should “break up” a monopolist into rival brokers depends on whether the objective is to enhance the agent’s welfare or to maximize efficiency. Such nuanced policy implications echo recent legal decisions. Between 2013 and 2016, \emph{CARFAX} was engaged in a class-action lawsuit wherein plaintiffs alleged that \emph{CARFAX} engaged in anti-competitive practices and had become a de facto monopoly, controlling 90\% of the market for certifying pre-owned vehicles. While \emph{CARFAX} did not dispute its monopoly status, it argued that its use of short-term contracts promoted market efficiency because it disciplined \emph{CARFAX} to always provide high-quality information.\footnote{See details on the case \emph{Maxon Hyundai Mazda et al. v Carfax, Inc} and the associated arguments referenced above at \url{https://web.archive.org/web/20241119090407/https://casetext.com/case/mazda-v-carfax-inc-1}.} \emph{CARFAX} won both the lawsuit and the appeal.

More broadly, in an increasingly data-driven world, it is essential to understand the interdependence between information markets and traditional markets for goods and services. How does the sale of information shape outcomes in traditional markets? What policies best align the incentives of information providers with those of market participants? And how should regulators balance efficiency against equitable distribution of surplus when these objectives diverge across these two markets? Our paper addresses these questions in the context of search and, we hope, provides a foundation for examining them more broadly.

\appendix
\addcontentsline{toc}{section}{Appendices}
\renewcommand{\thesubsection}{Appendix \Alph{subsection}}

\section*{Appendix}\label{appendix}
 \subsection{Additional useful lemmas}

\begin{lema}\label{lemma:support}
For all $k\in(0,m^\varnothing)$ and $u\in[\underline u_k, \bar u_k]$, $\mathbb{E}_F\big[\theta|\theta\leq \mathbf x_k(u)\big]\leq \underline u_k$ and $\bar u_k< \mathbb{E}_F\big[\theta|\theta\geq \mathbf x_k(u)\big]$.
\end{lema}
\begin{proof}
Fix any $k\in (0,m^\varnothing)$ and $u\in[\underline u_k, \bar u_k]$. First, notice that $\mathbb{E}_F\big[\theta|\theta\geq \mathbf x_k(u)\big]>\mathbf x_k(u)\geq \bar u_k$, where the first inequality follows because $F$ is absolutely continuous and $\mathbf x_k(\cdot)$ is bounded from above by $\mathbf{x}_k(\underline u_k)<1$ as established in \autoref{lemma:pass-fail-sufficiency}, and the second inequality follows because $\mathbf x_k(\cdot)$ is bounded from below by $\bar u_k$, also as established in \autoref{lemma:pass-fail-sufficiency}.

Next, we argue that $\mathbb{E}_F[\theta|\theta\leq  \mathbf x_k(\underline u_k)]= \underline u_k$, which is sufficient for the desired result because $\mathbb{E}_F[\theta|\theta\leq  \mathbf x_k(u)]\leq \mathbb{E}_F[\theta|\theta\leq  \mathbf x_k(\underline u_k)]$ (because $\mathbf x_k(\cdot)$ is decreasing). To see why $\mathbb{E}_F[\theta|\theta\leq  \mathbf x_k(\underline u_k)]= \underline u_k$, notice that 
\begin{align*}
& c_{G^\varnothing}(\underline u_k)=c_{G^{\mathbf x_k(\underline u_k)}}(\underline u_k)\\[6pt]
\Longleftrightarrow  & m^\varnothing-\underline u_k=c_F\big(\mathbf x_k(\underline u_k)\big)+\big(1-F\big(\mathbf x_k(\underline u_k)\big)\big)\big(\mathbf x_k(\underline u_k)-\underline u_k\big)\\[6pt]
\Longleftrightarrow &m^\varnothing-\underline u_k=\int^1_{\mathbf x_k(\underline u_k)}\big(\theta-\mathbf x_k(\underline u_k)\big)dF(\theta)+\big(1-F\big(\mathbf x_k(\underline u_k)\big)\big)\big(\mathbf x_k(\underline u_k)-\underline u_k\big)\\[6pt]
\Longleftrightarrow &\mathbb{E}_F[\theta|\theta\leq  \mathbf x_k(\underline u_k)]=\underline u_k,
\end{align*}
where the first line follows because both terms equal $k$ by construction, the first equivalence follows because $c_{G^\varnothing}(u)=(m^\varnothing-u)^+$ and from the definition of $c_{G^x}$ in \eqref{eq:pass-fail-function}, the second equivalence follows from expressing $c_F(\cdot)$ in its integral form, and the last equivalence follows from algebraic manipulation.
\end{proof}

\begin{lema}\label{lemma:dist-function}
For all $k\in(0,m^\varnothing)$,  $u\in[\underline u_k, \bar u_k]$, and  $u'\in [\underline u_k, \bar u_k]$, $c_{G^{\mathbf{x}_k(u)}}(u')=c_F(\mathbf{x}_k(u))+(1-F(\mathbf{x}_k(u)))(\mathbf{x}_k(u)-u')$.
\end{lema}
\begin{proof}
 Fix any $k\in (0,m^\varnothing)$, and  $u, u'\in[\underline u_k, \bar u_k]$. As defined in \eqref{eq:pass-fail-function},  
 \begin{align*}
     c_{G^{\mathbf{x}_k(u)}}(u')&=\max\{c_{G^\varnothing}(u'), c_F(\mathbf{x}_k(u))+(1-F(\mathbf{x}_k(u)))(\mathbf{x}_k(u)-u')\}>0
\end{align*}
where the inequality follows because $F$ is absolutely continuous over $\Theta$ and because $u'\leq\bar u_k\leq\mathbf{x}_k(u) <1$ as stated in \autoref{lemma:pass-fail-sufficiency}. 

Recall that $c_{G^\varnothing}(u')=(m^\varnothing-u')^+$ by definition. Suppose for the sake of contradiction that $c_{G^\varnothing}(u')> c_F(\mathbf{x}_k(u))+(1-F(\mathbf{x}_k(u)))(\mathbf{x}_k(u)-u')$, which implies that $c_{G^\varnothing}(u')=m^\varnothing-u'$. We then have 
\begin{align*}
   &c_{G^\varnothing}(u')>c_F(\mathbf{x}_k(u))+(1-F(\mathbf{x}_k(u)))(\mathbf{x}_k(u)-u')\\[6pt]
   \Longleftrightarrow & m^\varnothing-u'  > \int^1_{\mathbf x_k(u)}\big(\theta-\mathbf x_k(u)\big)dF(\theta)+\big(1-F\big(\mathbf x_k(u)\big)\big)\big(\mathbf x_k(u)-u'\big)\\[6pt]
\Longleftrightarrow &\mathbb{E}_F[\theta|\theta\leq  \mathbf x_k(u)]> u',
\end{align*}
where the first equivalence follows because $c_{G^\varnothing}(u')=(m^\varnothing-u')^+$ by definition and $c_{G^\varnothing}(u')>0$ by assumption, and by expressing $c_F(\cdot)$ in its integral form, and the last equivalence follows from algebra. However, the last line contradicts \autoref{lemma:support}. Hence, we must have $c_{G^\varnothing}(u')\leq c_F(\mathbf{x}_k(u))+(1-F(\mathbf{x}_k(u)))(\mathbf{x}_k(u)-u')$ to avoid the contradiction, establishing the desired result. 
\end{proof}

\begin{lema}\label{lemma:order}
For all $k\in(0,m^\varnothing)$ and $u'\in[\underline u_k, \bar u_k]$, the mapping $u\mapsto c_{G^{\mathbf{x}_k(u)}}(u')$ is strictly increasing over the interval $[\underline u_k, \bar u_k]$.
\end{lema}
\begin{proof}
Fix any $k\in (0,m^\varnothing)$ and $u'\in[\underline u_k, \bar u_k]$. From \autoref{lemma:dist-function}, we know that for all $u\in [\underline u_k,\bar u_k]$, $c_{G^{\mathbf{x}_k(u)}}(u')=c_F(\mathbf{x}_k(u))+(1-F(\mathbf{x}_k(u)))(\mathbf{x}_k(u)-u')$, which is a continuous function of $u$. Furthermore, $c_F(\cdot)$ is differentiable, while $F(\cdot)$ and $\mathbf x_k(\cdot)$ are differentiable almost everywhere--the absolute continuity of $F(\cdot)$ implies that $F(\cdot)$ is differentiable almost everywhere on $(0,1)$ and it also implies that $c_F(\cdot)$ is differentiable everywhere on $(0,1)$, and the a.e. differentiability of $\mathbf x_k(\cdot)$ follows via implicit differentiation. The lemma is then established by taking the derivative with respect to $u$ over the open interval $(\underline u_k, \bar u_k)$, which yields a.e.
\begin{align*}
    &\frac{d }{du}c_{G^{\mathbf{x}_k(u)}}(u')=-\frac{\partial \mathbf{x}_k(u)}{\partial u}f(\mathbf{x}_k(u))(\mathbf{x}_k(u)-u')>0,
\end{align*}
where the inequality follows from the fact that $\mathbf{x}_k(\cdot)$ is strictly decreasing and that $\mathbf{x}_k(u)>\mathbf{x}_k(\bar u_k)=\bar u_k\geq u'$ for all $u\in (\underline u_k, \bar u_k)$ as established in \autoref{lemma:pass-fail-sufficiency}.
\end{proof}

\subsection{Proofs from Section~\ref{sec:proofmain}}\label{sec:appendix_proposition}

\noindent\begin{proof}[Proof of \autoref{thm:monopoly}]
Let $n=1$. Given a non-empty set $E\subseteq\mathbb R^{2}$, define
\begin{align*}
\label{eq:min-max}
\tag{7}
\underline{\mathcal{V}}(E)\coloneqq \inf_{\beta, \psi} \hspace*{.2em} \sup_{o\in \mathcal{O}}  \hspace*{.2em} &V_1(o, \beta, \psi)\\[6pt]
\text{s.t.}\hspace*{1em} &  \beta \text{ is sequentially rational with respect to } \psi, \text{ and }\\[6pt]
& Im_\psi\subseteq E.
\end{align*}
We shall refer to $\mathcal{V}(E)$ as the monopolist's minimax payoff when continuation values are constrained to a set $E$. The following lemma shows that there is a close connection between this minimax value $\underline{\mathcal{V}}(E)$ and all other payoff profiles that can be supported by continuation values within $E$.

\begin{lemma}\label{lemma:supported}
Consider a non-empty set $E\subseteq \mathbb{R}^2$. Any payoff profile $y\in\mathcal{F}(1,k)$ such that $y_1>\underline{\mathcal{V}}(E)$ is supported by a tuple $(\alpha,\beta,\psi)$ with $Im_\psi\subseteq E\cup\{y\}$.
\end{lemma}
\noindent \begin{proof}
Fix a non-empty set $E\subseteq\mathbb{R}^2$ and a payoff profile $y\in\mathcal{F}(1,k)$ with $y_1>\underline{\mathcal{V}}(E)$. By construction (from \eqref{eq:min-max}), there exists a pair $(\beta', \psi')$, such that: $\beta'$ is sequentially rational with respect to $\psi'$,  $Im_{\psi'}\subseteq E$, and $y_1\geq V_1(o, \beta', \psi')$ for all $o\in \mathcal{O}$.

Additionally, by \autoref{lemma:stationary-implement}, there exist a TIOLI offer $o^*(y)$ and a pair $(\beta'',\psi'')$ such that: 
 $\beta''$ is sequentially rational with respect to $\psi''$,
 $\psi''(o^*(y), w, m)=y$ for all $(w,m)\in W\times \Theta$, and 
$y_1=V_1(o^*(y), \beta'', \psi'')$ and $y_2=U(o^*(y), \beta'', \psi'')$.

Let us define a new tuple $(\hat\alpha, \hat\beta, \hat\psi)$, where $\hat\alpha_i$ for all $i\in N$ is the Dirac measure centered on $o^*(y)$,  the agent's simple strategy is given by 
\[
\hat\beta^{\mathcal{w}}(\cdot|o)=\begin{cases}
      \beta^{''\mathcal{w}}(\cdot|o) & \mbox{if }  o=o^*(y)\\
    \beta^{'\mathcal{w}}(\cdot|o) & \mbox{if } o\neq o^*(y)
\end{cases} \quad \text{ and } \quad \hat\beta^{\mathcal{d}}(o, w, m)=\begin{cases}
    \beta^{''\mathcal{d}}(o, w, m) & \mbox{if }  o=o^*(y)\\
    \beta^{'\mathcal{d}}(o,w,m) & \mbox{if } o\neq o^*(y)
    \end{cases}
\]
for all $(w,m)\in W\times \Theta$, and the continuation value function is given by 
\[\hat\psi(o, w, m)=\begin{cases}
    \psi^{''}(o, w, m) & \mbox{if }  o=o^*(y)\\
    \psi^{'}(o,w,m) & \mbox{if }  o\neq o^*(y).
\end{cases}
\]
Then notice that 
$(i)$ $y_1=\mathbb{E}_{\hat\alpha}[V_1(o, \hat\beta, \hat\psi)]$ and $y_2=\mathbb{E}_{\hat\alpha}[U(o, \hat\beta, \hat\psi)]$,  
$(ii)$ $y_1\geq V_1(o, \hat\beta, \hat\psi)$ for all $o\in \mathcal{O}$, and 
$(iii)$ $\hat \beta$ is sequentially rational with respect to $\hat\psi$, which correspond to Points $(a)$-$(c)$ of \autoref{def:self-generation}. Hence, $y$ is supported by $(\hat\alpha, \hat\beta, \hat\psi)$. Moreover, by construction, $Im_{\hat\psi}\subseteq E\cup\{y\}$.
\end{proof}
\bigskip

Of particular note is the minimax value $\underline{\mathcal{V}}(\mathcal{S}(1,k))$, which is well-defined because $\mathcal{S}(1,k)$ is non-empty.\footnote{Recall that $\mathcal{E}(1,k)\subseteq\mathcal{S}(1,k)$ and $\mathcal{E}(1,k)$ is non-empty; \cite{mek23} (Online Appendix B) show the existence of a stationary equilibrium in this setting.} For ease of exposition, we henceforth denote $\underline{\mathcal{V}}(\mathcal{S}(1,k))$ simply as $\nu_k$. The value $\nu_k$ is the profit guarantee that a monopolist can secure for herself when continuation payoffs are constrained to be self-generating outcomes. Hence, if a payoff profile $y\in\mathcal{S}(1,k)$, then $y_1\geq \nu_k$. We show that this is also a sufficient condition. In order to state our result, let us first define the following notation: Given a subset $X\subseteq \mathbb{R}^n$, we write $\cl(X)$ and $\conv(X)$ to denote its  closure and its convex hull, respectively.

\begin{lemma}\label{lemma:monopoly-lower-bound}
For any $k\in (0,m^\varnothing)$ and any $y\in\mathcal{F}(1,k)$, $y\in \cl(\mathcal{S}(1, k))$ if and only if $y_1\geq \nu_k$.
\end{lemma}
\noindent \begin{proof}
(``Only-if" direction:) Consider a payoff profile $y\in \cl(\mathcal{S}(1,k))$. There exists a convergent sequence of payoff profiles $(y^\ell)_{\ell\in\mathbb N}$ with $y^\ell\in\mathcal{S}(1,k)$ for each $\ell\in \mathbb{N}$ such that $\lim_{\ell\to \infty}y^\ell=y$. Since $y^\ell\in\mathcal{S}(1,k)$ for each $\ell\in\mathbb N$, there exists a tuple $(\alpha^\ell, \beta^\ell,\psi^\ell)$ with $Im_{\psi^\ell}\subseteq \mathcal{S}(1,k)$ such that the tuple supports $y^\ell$.  Then, Point $(b)$ of \autoref{def:self-generation} immediately implies $y_{1}^\ell\geq \nu_k$ for all $\ell\in\mathbb N$, giving the desired conclusion. \\

\noindent (``If" direction:) Consider a payoff profile $y\in\mathcal{F}(1,k)$ with $y_1\geq \nu_k$. First, suppose $\nu_k=\bar u_k-\underline u_k$. In this case, feasibility implies that $y=(\bar u_k-\underline u_k, \underline u_k)$. From the ``only-if" direction, we know that $\cl(\mathcal{S}(1,k))\subseteq\{y'\in\mathcal{F}(1,k):y'_1\geq \nu_k\}=\{y\}$. As $\mathcal{S}(1,k)$ is non-empty, we conclude that $\cl(\mathcal{S}(1,k))=\{y\}$.

Now suppose $\nu_k<\bar u_k-\underline u_k$. Then the set $E'\coloneqq \{y'\in\mathcal{F}(1,k):y'_1>\nu_k\}$ is non-empty. From \autoref{lemma:supported}, any $y'\in E'$ is supported by a tuple $(\alpha', \beta', \psi')$ with  $Im_{\psi'}\subseteq \mathcal{S}(1,k)\cup\{y'\}$, which implies that $\mathcal{S}(1,k)\cup E'$ is self-generating. However, $\mathcal{S}(1,k)$ is the largest self-generating set, so $E'\subseteq \mathcal{S}(1,k)$, which further implies that $\cl(E')\subseteq \cl(\mathcal{S}(1,k))$. Noticing that $y\in \cl(E')$, we conclude that $y\in \cl(\mathcal{S}(1,k))$ as desired. 
\end{proof}
\bigskip

\begin{lemma}\label{lemma:monopoly-closure}
For any $k\in (0,m^\varnothing)$, $\nu_k=\underline{\mathcal{V}}\big(\cl(\mathcal{S}(1,k))\big)$. 
\end{lemma}
\noindent\begin{proof}
Notice  $\underline{\mathcal{V}}(\cdot)$ is weakly decreasing in the set-inclusion order, i.e., $\underline{\mathcal{V}}(E)\geq \underline{\mathcal{V}}(E')$ whenever $E\subseteq E'$. Hence, $\nu_k\geq \underline{\mathcal{V}}(\cl(\mathcal{S}(1,k)))$. For the sake of contradiction, suppose $\nu_k> \underline{\mathcal{V}}(\cl(\mathcal{S}(1,k)))$. Consider any payoff profile $y\in\cl(\mathcal{S}(1,k))$. From \autoref{lemma:monopoly-lower-bound}, we have $y_1\geq \nu_k$, and therefore, $y_1>\underline{\mathcal{V}}(\cl(\mathcal{S}(1,k)))$. By \autoref{lemma:supported}, $y$ is supported by a tuple $(\alpha,\beta, \psi)$ with $Im_\psi\subseteq\cl(\mathcal{S}(1,k))\cup\{y\}=\cl(\mathcal{S}(1,k))$. Hence, $\cl(\mathcal{S}(1, k))$ is a self-generating set. However, $\mathcal{S}(1, k)$ is the largest such set, implying $\mathcal{S}(1, k)=\cl(\mathcal{S}(1,k))$, and $\nu_k=\underline{\mathcal{V}}(\cl(\mathcal{S}(1,k)))$, contradicting the assumption that $\nu_k>\underline{\mathcal{V}}(\cl(\mathcal{S}(1,k)))$.
\end{proof}
\bigskip

The last two lemmas further simplify our analysis by showing that the closure of the largest self-generating payoff set is characterized by the intersection of three half spaces: two half spaces that determine feasibility of payoffs and a third one given by \autoref{lemma:monopoly-lower-bound}. In other words,
\[
\cl(\mathcal{S}(1,k))= \conv(\{(\bar u_k-\underline u_k, \underline u_k), (\nu_k,\underline u_k), (\nu_k, \bar u_k-\nu_k)\}). 
\]
Moreover, from \autoref{lemma:monopoly-closure},   
\[
\nu_k=\inf\left\{y_1\in[0,\bar u_k-\underline u_k]:y_1= \underline{\mathcal{V}}\Big( \conv\big(\{(\bar u_k-\underline u_k, \underline u_k), (y_1,\underline u_k), (y_1, \bar u_k-y_1)\}\big)\Big)\right\},
\]
which reduces the task of characterizing the largest self-generating payoff set into finding the smallest fixed point in a single-dimensional fixed-point problem. 

\begin{lemma}\label{lemma:monopoly-fixed-point}
There exists a unique $k^*\in (0, m^\varnothing)$ such that 
\[\nu_k=\begin{cases}
 0 &\mbox{if }  k\leq k^* \\ 
 \bar u_k-\underline u_k&\mbox{if }  k> k^*.
 \end{cases}
 \]
\end{lemma}
\noindent\begin{proof}
Consider the following fixed point problem: find $y_1\in[0,\bar u_k-\underline u_k]$ such that
\[
\label{eq:fixed-pt}
\tag{A5}
y_1= \underline{\mathcal{V}}\Big( \conv\big(\{(\bar u_k-\underline u_k, \underline u_k), (y_1,\underline u_k), (y_1, \bar u_k-y_1)\}\big)\Big).
\]
To prove the statement of \autoref{lemma:monopoly-fixed-point}, it suffices to show that $y_1=0$ solves \eqref{eq:fixed-pt} for all $k\leq k^*$, and that $y_1=\bar u_k-\underline u_k$ is the unique value that solves \eqref{eq:fixed-pt} for all $k>k^*$. Our proof proceeds in four steps.\medskip

\noindent \textbf{Step 1:} Fix any $k\in(0, m^\varnothing)$. We show that $y_1=\bar u_k-\underline u_k$ is a fixed point of \eqref{eq:fixed-pt}.

When $y_1=\bar u_k-\underline u_k$, the set $\conv\big(\{(\bar u_k-\underline u_k, \underline u_k), (y_1,\underline u_k), (y_1, \bar u_k-y_1)\}\big)$ collapses into a singleton $\{(\bar u_k-\underline u_k, \underline u_k)\}$. In other words, $\psi(o,w,m)=(\bar u_k-\underline u_k, \underline u_k)$ for all $(o,w,m)\in \mathcal{O}\times W\times \Theta$. Then, $\beta$ is sequentially rational with respect to $\psi$ only if for all $o\coloneqq (p, G)\in\mathcal{O}$,
    \[
\beta^\mathcal{w}(1|o)=\begin{cases}
    0 & \mbox{if } c_{G}(\underline u_k)-p<c_{G^\varnothing}(\underline u_k)\\
    1 & \mbox{if }  c_{G}(\underline u_k)-p>c_{G^\varnothing}(\underline u_k).
    \end{cases}
    \]
    
For $\epsilon>0$, consider offer $o_\epsilon\coloneqq (p_\epsilon, G^{\mathbf x_k(\bar u_k)})$ with $p_\epsilon=c_{G^{\mathbf x_k(\bar u_k)}}(\underline u_k)-c_{G^\varnothing}(\underline u_k)-\epsilon$. Notice
\begin{align*}
    c_{G^\varnothing}(\underline u_k)=k=c_{F}(\bar u_k)=c_{G^{\mathbf x_k(\bar u_k)}}(\bar u_k)<c_{G^{\mathbf x_k(\bar u_k)}}(\underline u_k),
\end{align*}
where the first equality follows because $\underline u_k$ solves \eqref{eq:1} when $c_G=c_{G^\varnothing}$, the second equality similarly follows as $\bar u_k$ solves \eqref{eq:1} when $c_G=c_F$, the third equality follows because $\bar u_k=\mathcal{u}_k({G^{\mathbf x_k(\bar u_k)}})$ as established in \autoref{lemma:pass-fail-sufficiency}, and the inequality follows because $c_{G^{\mathbf x_k(\bar u_k)}}(\cdot)$ is strictly decreasing over the interval $(0, \mathbb{E}_F[\theta|\theta\geq \mathbf x_k(\bar u_k)])$ and $\underline u_k<\mathbb{E}_F[\theta|\theta\geq \mathbf x_k(\bar u_k)]$ by \autoref{lemma:support}. Thus, for small enough $\epsilon>0$, $p_\epsilon\in\mathbb{R}_+$ and $o_\epsilon$ is well-defined. Furthermore, if the monopolist offers $o_\epsilon$, then sequential rationality of $\beta$ with respect to $\psi$ implies that the agent accepts $o_\epsilon$, i.e., $\beta^\mathcal{w}(1|o_\epsilon)=1$.

The broker's payoff from offer $o_\epsilon$ is then given by
\small{\begin{align*}
    V_1(o_\epsilon, \beta, \psi)&=p_\epsilon + \big(\bar u_k-\underline u_k\big)\int_\Theta \beta^\mathcal{d}(o_\epsilon,1,m) dG^{\mathbf x_k(\bar u_k)}(m)\\[6pt]
    &=c_{G^{\mathbf x_k(\bar u_k)}}(\underline u_k)-c_{G^\varnothing}(\underline u_k)-\epsilon+ \big(\bar u_k-\underline u_k\big)\int_\Theta \beta^\mathcal{d}(o_\epsilon,1,m) dG^{\mathbf x_k(\bar u_k)}(m)\\[6pt]
    &\geq c_{G^{\mathbf x_k(\bar u_k)}}(\underline u_k)-c_{G^\varnothing}(\underline u_k)-\epsilon + \big(\bar u_k-\underline u_k\big)G^{\mathbf x_k(\bar u_k)}(\underline u_k)\\[6pt]
    &=c_{G^{\mathbf x_k(\bar u_k)}}(\underline u_k)-c_{G^\varnothing}(\underline u_k)-\epsilon + \big(\bar u_k-\underline u_k\big)F\big(\mathbf x_k(\bar u_k)\big)\\[6pt]
    &=c_F(\mathbf x_k(\bar u_k))+\big(1-F(\mathbf x_k(\bar u_k))\big)\big(\mathbf x_k(\bar u_k)-\underline u_k\big)-c_{G^\varnothing}(\underline u_k)-\epsilon+ \big(\bar u_k-\underline u_k\big)F\big(\mathbf x_k(\bar u_k)\big)\\[6pt]
    &=c_F(\bar u_k)+\big(1-F(\bar u_k)\big)\big(\bar u_k-\underline u_k\big)-c_{G^\varnothing}(\underline u_k)-\epsilon+ \big(\bar u_k-\underline u_k\big)F\big(\bar u_k\big)\\[6pt]
    &=\bar u_k-\underline u_k-\epsilon,
\end{align*}}
\hspace{-1mm}\normalsize
where the second equality follows by substituting in the expression for $p_\epsilon$, the inequality follows because sequential rationality of $\beta$ implies that $\beta^\mathcal{d}(o_\epsilon, 1, m)\geq \mathbbm{1}_{[m<\underline u_k]}$ for all $m\in\Theta$ and because $\underline u_k\notin\supp(G^{\mathbf x_k(\bar u_k)})$ making $G^{\mathbf x_k(\bar u_k)}(\cdot)$ continuous at $\underline u_k$, the third equality follows from \eqref{eq:pass-fail-distibution} and the fact that $\mathbb{E}_F[\theta|\theta\leq \mathbf x_k(\bar u_k)]<\underline u_k<\mathbb{E}_F[\theta|\theta\geq \mathbf x_k(\bar u_k)]$ by \autoref{lemma:support}, the fourth equality follows by \autoref{lemma:dist-function}, the fifth equality follows by the fact that $\mathbf x_k(\bar u_k)=\bar u_k$ as stated in \autoref{lemma:pass-fail-sufficiency}, and the final equality follows from the fact that $c_F(\bar u_k)=c_{G^\varnothing}(\underline u_k)=k$. 

In the limit as $\epsilon$ becomes arbitrarily small, we  conclude that $\bar u_k-\underline u_k\leq \underline{\mathcal{V}}(\{(\bar u_k-\underline u_k, \underline u_k)\})$. However, 
$\underline{\mathcal{V}}(\{(\bar u_k-\underline u_k, \underline u_k)\})$ cannot exceed $\bar u_k-\underline u_k$, which therefore implies that $\bar u_k-\underline u_k$ is a fixed-point solution to \eqref{eq:fixed-pt}. This concludes Step 1.\medskip

For the remaining steps, we define the continuous mapping $\Phi:(0,m^\varnothing)\to \mathbb{R}$ given by 
\[
\label{eq:phi}
\tag{A6}
\Phi(k)\coloneqq c_F(\underline u_k)+\underline u_k-\big( c_{G^\varnothing}(\bar u_k)+\bar u_k\big).
\]

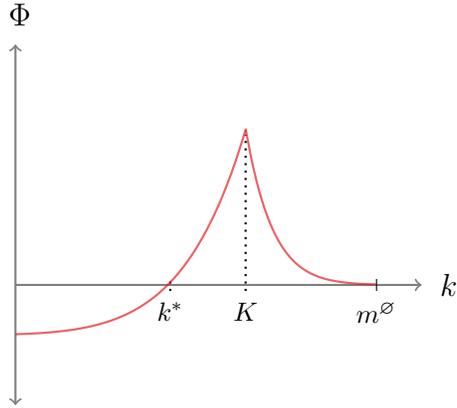
\begin{figure}[ht]
 \begin{center}
 \begin{tikzpicture}[xscale=5, yscale=4]
\draw [->, help lines,   thick] (0,0) -- (.9,0);
 \draw [<->, help lines,   thick] (0.0,-.4) --(0,.8);
 \node[right=3pt] at (.9,0) {$k$};
 \node[above=3pt] at (.01, .8) {$\Phi$};

 \draw [] (.8,-.02)node[below]{\footnotesize{$m^\varnothing$}}--(.8,.02);
  \draw [dotted, thick] (.51,-.02) node[below]{\footnotesize{$K$}}--(.51,.5);
  \draw [dotted, thick] (.343,-.02) node[below]{\footnotesize{$k^*$}}--(.343,.02);

   \draw[cyan, opacity=.15, domain=0:.51,smooth,variable=\x, thick] plot ({\x},{(.437+\x)^7-.167});
 \draw[cyan, opacity=.15, domain=.51:.8,smooth,variable=\x, thick] plot ({\x},{(1.47-\x)^16});
   \draw[red, opacity=.6, domain=0:.51,smooth,variable=\x, thick] plot ({\x},{(.437+\x)^7-.167});
 \draw[red, opacity=.6, domain=.51:.8,smooth,variable=\x, thick] plot ({\x},{(1.47-\x)^16});
\end{tikzpicture}
 \captionsetup{oneside,margin={0cm,0cm},justification=justified, singlelinecheck=false}
 \caption{The continuous mapping $\Phi:(0,m^\varnothing)\to \mathbb{R}$ is depicted by the red curve. As we will show, the figure also depicts $K$ which is the unique value such that $\bar u_{K}=m^\varnothing$,  and $k^*$ which is the unique value such that $\Phi(k^*)=0$. $\Phi(\cdot)$ is strictly increasing over the interval $(0,K)$ and strictly decreasing over the interval $(K,m^\varnothing)$.}
   \label{fig:phi}
 \end{center}
 \end{figure}
 
\noindent \textbf{Step 2:} Fix any $k\in(0, m^\varnothing)$ such that $\Phi(k)>0$. We show that $y_1=\bar u_k-\underline u_k$ is the unique fixed point of \eqref{eq:fixed-pt}.

Our proof proceeds by showing that $y_1<\underline{\mathcal{V}}\Big( \conv\big(\{(\bar u_k-\underline u_k, \underline u_k), (y_1,\underline u_k), (y_1, \bar u_k-y_1)\}\big)\Big)$ for all $y_1\in[0,\bar u_k-\underline u_k)$. To that end, fix $y_1\in[0,\bar u_k-\underline u_k)$. For some $\epsilon>0$, consider  offer $o_\epsilon\coloneqq (p_\epsilon, G^x)$ with $x=\underline u_k+y_1\mathbbm{1}_{[\bar u_k-y_1<m^\varnothing]}$ and $p_\epsilon=c_{G^x}(\underline u_k)+\underline u_k-\big(c_{G^\varnothing}(\bar u_k-y_1)+\bar u_k-y_1\big)-\epsilon$.

If $\bar u_k-y_1<m^\varnothing$, then
\begin{align*}
c_{G^x}(\underline u_k)+\underline u_k-\big(c_{G^\varnothing}(\bar u_k-y_1)+\bar u_k-y_1\big)=c_{G^x}(\underline u_k)+\underline u_k-m^\varnothing=c_{G^x}(\underline u_k)-k,
\end{align*}
where the equality follows because $c_{G^\varnothing}(u)=(m^\varnothing-u)^+$ and $\bar u_k-y_1<m^\varnothing$ by assumption, and the second equality follows because $\underline u_k=m^\varnothing-k$. Additionally, 
$
k=c_{F}(\bar u_k)<c_F(x)=c_{G^{x}}(x)\leq c_{G^{x}}(\underline u_k),
$
where the first equality follows because $\bar u_k$ solves \eqref{eq:1} when $c_G=c_F$, the first inequality follows because $c_F$ is a strictly decreasing function and because $x=\underline u_k+y_1<\bar u_k$ (recall that $y_1<\bar u_k-\underline u_k$), the second equality follows by the construction of $c_{G^x}$ as in \eqref{eq:pass-fail-function}, and the last inequality follows because $c_{G^x}$ is a weakly decreasing function. Thus, for small enough $\epsilon>0$, $p_\epsilon\in\mathbb{R}_+$ and $o_\epsilon$ is well-defined.

Similarly, if $\bar u_k-y_1\geq m^\varnothing$, then
\begin{align*}
c_{G^x}(\underline u_k)+\underline u_k-\big(c_{G^\varnothing}(\bar u_k-y_1)+\bar u_k-y_1\big)&=c_{F}(\underline u_k)+\underline u_k-\big(c_{G^\varnothing}(\bar u_k-y_1)+\bar u_k-y_1\big)\geq \Phi(k)
\end{align*}
where the equality follows because $x=\underline u_k$ and $c_{G^x}(x)=c_F(x)$ by construction, and the inequality follows because $c_{G^\varnothing}(u)+u$ is weakly increasing. As $\Phi(k)>0$ by assumption, then for small enough $\epsilon>0$, $p_\epsilon\in\mathbb{R}_+$ and $o_\epsilon$ is once again well-defined.

For any $(\beta, \psi)$ that satisfies the constraints of \eqref{eq:min-max}, the agent's payoff difference from accepting $o_\epsilon$ versus accepting the null offer is given by 
\begin{align*}
&\underbrace{\int_\Theta\max\{m, \psi_2(o_\epsilon, 1, m)\}dG^x(m)-k-p_\epsilon}_{\text{payoff from accepting } o_\epsilon}-\left(\underbrace{\int_\Theta\max\{m, \psi_2(o_\epsilon, \emptyset, m)\}dG^\varnothing(m)-k}_{\text{payoff from accepting null offer}}\right)\\[6pt]
\geq &\int_\Theta\max\{m, \underline u_k\}dG^x(m)-p_\epsilon-\left(\int_\Theta\max\{m, \bar u_k-y_1\}dG^\varnothing(m)\right)\\[6pt]
= &c_{G^x}(\underline u_k)+\underline u_k-p_\epsilon-\Big(c_{G^\varnothing}\big(\bar u_k-y_1\big)+\bar u_k-y_1\Big)\\[6pt]
=&\epsilon,
\end{align*}
where the inequality follows from the fact that $\underline u_k\leq \psi_2(o, w, m)\leq \bar u_k-y_1$ for all $(o,w, m)\in \mathcal{O}\times W\times \Theta$. Therefore, the agent strictly prefers to accept the offer $o_\epsilon$ over the null offer, i.e., $\beta^\mathcal{w}(1|o_\epsilon)=1$. 

Furthermore, for any  $(\beta, \psi)$ that satisfies the constraints of \eqref{eq:min-max}, the broker's payoff from offering $o_\epsilon$ is given by 
\begin{align*}
    V_1(o_\epsilon, \beta, \psi)&=p_\epsilon + \int_\Theta \beta^\mathcal{d}(o_\epsilon, 1, m)\psi_1(o_\epsilon, 1,m)dG^x(m)\\[6pt]
    &\geq p_\epsilon + y_1\int_\Theta \beta^\mathcal{d}(o_\epsilon, 1, m)dG^x(m)\\[6pt]
    &\geq p_\epsilon + y_1 G^x(\underline u_k)\\[6pt]
    &=c_{G^x}(\underline u_k)+\underline u_k-\big(c_{G^\varnothing}(\bar u_k-y_1)+\bar u_k-y_1\big)+ y_1 G^x(\underline u_k)-\epsilon
\end{align*}
where the first equality follows because  $\beta^\mathcal{w}(1|o_\epsilon)=1$, the first inequality follows from the fact that $\psi_1(o, w, m)\geq y_1$ for all $(o,w, m)\in \mathcal{O}\times W\times \Theta$, the second inequality follows because sequential rationality of $\beta$ implies that $\beta^\mathcal{d}(o, w, m)\geq\mathbbm{1}_{[m<\underline u_k]}$ for all $m\in\Theta$, and the last equality follows from substituting in the expression for $p_\epsilon$.

If $\bar u_k-y_1<m^\varnothing$, then 
\begin{align*}
    V_1(o_\epsilon, \beta, \psi)&\geq c_{G^x}(\underline u_k)+\underline u_k-\big(c_{G^\varnothing}(\bar u_k-y_1)+\bar u_k-y_1\big)+ y_1 G^x(\underline u_k)-\epsilon\\[6pt]
     &=c_{G^x}(\underline u_k)+\underline u_k-m^\varnothing+ y_1 G^x(\underline u_k)-\epsilon\\[6pt]
    &=c_{G^x}(\underline u_k)-k+ y_1 G^x(\underline u_k)-\epsilon\\[6pt]
    &=c_{G^x}(\underline u_k)-k+ y_1 F(x)-\epsilon\\[6pt]
    &= c_{F}(x)+\big(1-F(x)\big)(x-\underline u_k)-k+ y_1 F(x)-\epsilon\\[6pt]
    &=c_{F}(x)-k+y_1-\epsilon
    \end{align*}
where the first equality follows because $c_{G^\varnothing}(u)=(m^\varnothing-u)^+$ and $\bar u_k-y_1<m^\varnothing$ by assumption, the second equality follows because $\underline u_k=m^\varnothing-k$, the third equality follows because $G^{x}(m)=F(x)$ for all $m\in (\mathbb{E}_F[\theta|\theta\leq x], \mathbb{E}_F[\theta|\theta\geq x])$ by \eqref{eq:pass-fail-distibution} and the fact that $\mathbb{E}_F[\theta|\theta\leq x]\leq \mathbb{E}_F[\theta|\theta\leq \bar u_k]<\underline u_k<\mathbb{E}_F[\theta|\theta\geq \underline u_k]\leq \mathbb{E}_F[\theta|\theta\geq x]$ by \autoref{lemma:support}, the fourth equality follows from \autoref{lemma:dist-function}, and the final equality follows because $x=\underline u_k+y_1$ by construction. Therefore, for any $(\beta, \psi)$ that satisfies the constraints of \eqref{eq:min-max}, the broker's optimal payoff is bounded below by $c_F(x)-k+y_1$, and we can conclude that 
\begin{align*}
\underline{\mathcal{V}}\Big(\conv\big(\{(\bar u_k-\underline u_k, \underline u_k), (y_1,\underline u_k), (y_1, \bar u_k-y_1)\}\big)\Big)&\geq c_F(x)-k+y_1>y_1
\end{align*}
where the second inequality follows because $c_F(\cdot)$ is strictly decreasing with $c_F(\bar u_k)=k$ by \eqref{eq:fixed-pt}, and $x=\underline u_k+y_1$ with $y_1\in [0, \bar u_k-\underline u_k)$ by assumption. 

Similarly, if $\bar u_k-y_1\geq m^\varnothing$, then 
\begin{align*}
    V_1(o_\epsilon, \beta, \psi)&\geq c_{G^x}(\underline u_k)+\underline u_k-\big(c_{G^\varnothing}(\bar u_k-y_1)+\bar u_k-y_1\big)+ y_1 G^x(\underline u_k)-\epsilon\\[6pt]
     &=c_{G^x}(\underline u_k)+\underline u_k-\big(c_{G^\varnothing}(\bar u_k)+\bar u_k-y_1\big)+ y_1 G^x(\underline u_k)-\epsilon\\[6pt]
     &=c_{F}(\underline u_k)+\underline u_k-\big(c_{G^\varnothing}(\bar u_k)+\bar u_k-y_1\big)+ y_1 G^x(\underline u_k)-\epsilon\\[6pt]
     &=\Phi(k)+ y_1\big(1+ G^x(\underline u_k)\big)-\epsilon\end{align*}
where the first equality follows because $c_{G^\varnothing}(u)=(m^\varnothing-u)^+$ is a constant for all $u\geq m^\varnothing$ and $\bar u_k>\bar u_k-y_1\geq m^\varnothing$ by assumption, and the second equality follows because $x=\underline u_k$ and $c_{G^x}(x)=c_F(x)$ by construction. Therefore, for any $(\beta, \psi)$ satisfying the constraints of \eqref{eq:min-max}, the broker's optimal payoff is bounded below by $\Phi(k)+ y_1\big(1+F(\underline u_k)\big)$, implying
\begin{align*}
&\underline{\mathcal{V}}\Big(\conv\big(\{(\bar u_k-\underline u_k, \underline u_k), (y_1,\underline u_k), (y_1, \bar u_k-y_1)\}\big)\Big)
\geq \Phi(k)+ y_1\big(1+F(\underline u_k)\big)>y_1
\end{align*}
where the second inequality follows because $\Phi(k)>0$  by  assumption.

Consequently,  when $\Phi(k)>0$ (with either $\bar u_k-y_1<m^\varnothing$ or $\bar u_k-y_1\geq m^\varnothing$), there exists no fixed point of \eqref{eq:fixed-pt} with $y_1\in[0,\bar u_k-\underline u_k)$. By Step 1, we conclude that $y_1=\bar u_k-\underline u_k$ is the unique fixed point of \eqref{eq:fixed-pt} for any $k\in (0, m^\varnothing)$ that satisfies $\Phi(k)>0$. In other words, $\nu_k=\bar u_k-\underline u_k$ if $\Phi(k)>0$. This concludes Step 2.\medskip

\noindent \textbf{Step 3:} Fix any $k\in(0, m^\varnothing)$ such that  $\Phi(k)\leq 0$. We show $y_1=0$ is a fixed point of \eqref{eq:fixed-pt}.

When $y_1=0$, the set $\conv\big(\{(\bar u_k-\underline u_k, \underline u_k), (y_1,\underline u_k), (y_1, \bar u_k-y_1)\}\big)$ is equivalent to the feasible set $\mathcal{F}(1,k)$. Consider a pair $(\beta,\psi)$ where the reward function $\psi$ is
   \[
\psi(o,w,m)=\begin{cases}
    (0, \underline u_k) & \mbox{if }  w=1\\
    (0, \bar u_k) & \mbox{if }  w=\emptyset,
    \end{cases}
    \]
and the agent's simple strategy is given by $\beta^\mathcal{d}(o,w,m)=\mathbbm{1}_{[m\leq \psi_2(o,w,m)]}$, and for all $o\coloneqq (p, G)\in\mathcal{O}$,
    \[
\beta^\mathcal{w}(1|o)=\begin{cases}
    0 & \mbox{if }  c_{G}(\underline u_k)+\underline u_k-p\leq  c_{G^\varnothing}(\bar u_k)+\bar u_k\\
    1 & \mbox{if }  c_{G}(\underline u_k)+\underline u_k-p> c_{G^\varnothing}(\bar u_k)+\bar u_k.
    \end{cases}
    \]
Notice $Im_\psi\subseteq\mathcal{F}(1,k)$ and $\beta$ is sequentially rational with respect to $\psi$. Additionally, notice that for any TIOLI offer $o= (p,G)\in\mathcal{O}$,  
    $c_{G}(\underline u_k)+\underline u_k-p\leq  c_{F}(\underline u_k)+\underline u_k-p\leq  c_{G^\varnothing}(\bar u_k)+\bar u_k,$
where the first inequality follows because $c_G\leq c_F$ pointwise for all $G\in \mathcal{G}(F)$, and the second inequality follows by the ongoing assumption  of Step 3 that $\Phi(k)\leq 0$. Thus, the agent rejects all TIOLI offers made by the broker in this case, i.e., $\beta^\mathcal{w}(1|o)=0$ for all $o\in\mathcal{O}$. 

Given the prescribed $(\beta, \psi)$, the broker's payoff from offering any $o\in\mathcal{O}$ is 
\begin{align*}
    V_1(o, \beta, \psi)&=\int_\Theta \beta^\mathcal{d}(o, \emptyset, m)\psi_1(o, \emptyset,m)dG^\varnothing(m)=0
\end{align*}
where the first equality follows since the agent rejects all offers, and the second equality follows because $\psi_1(o,w,m)=0$ for all $(o,w,m)\in\mathcal{O}\times W\times \Theta$. 

We have thus shown that there exists a pair $(\beta, \psi)$ satisfying the constraints of \eqref{eq:min-max} with $\sup_{o\in\mathcal{O}}V_1(o,\beta, \psi)=0$. Therefore, $\underline{\mathcal{V}}\big(\mathcal{F}(1,k)\big)\leq 0$.  However, the smallest possible value for $\underline{\mathcal{V}}(\mathcal{F}(1,k))$ cannot be lower than zero, which therefore implies that $y_1=0$ is the smallest fixed-point solution to \eqref{eq:fixed-pt} when $\Phi(k)\leq 0$. In other words, $\nu_k=0$ if $\Phi(k)\leq 0$. This concludes Step 3.\medskip

\noindent \textbf{Step 4:} We finalize the proof by showing that there exists a unique $k^*\in (0, m^\varnothing)$ such that $\Phi(k)\leq 0$ for all $k\leq k^*$ and $\Phi(k)>0$ for all $k>k^*$. 

We claim that there exists a unique $K\in (0,m^\varnothing)$ such that $\bar u_{K}=m^\varnothing$. To see this, recall from \autoref{lemma:surplus-compstat-cost} that the mapping $k\mapsto \bar u_k$ is continuous and decreasing. Moreover, $\lim_{k\to 0}\bar u_k=1>m^\varnothing$ (the agent is happy to keep searching until he finds the highest quality good when searching is costless) and $\lim_{k\to m^\varnothing}\bar u_k=0<m^\varnothing$ (the agent immediately stops his search when searching is too costly). Therefore, the claim is completed by appealing to the intermediate value theorem.

Recall that $c_{G^\varnothing}(u)=(m-u)^+$. Thus,  for all $k\geq K$,  we have $\bar u_k\leq m^\varnothing$, which implies that  $c_{G^\varnothing}(\bar u_k)+\bar u_k=m^\varnothing$.  In this case $\Phi(k)=c_F(\underline u_k)+\underline u_k-m^\varnothing$. As $\underline u_k$ is strictly decreasing in $k$ and $c_F(u)+u$ is strictly increasing in $u$, $\Phi(k)$ is a strictly decreasing function when $k\in(K, m^\varnothing)$. Additionally, $\lim_{k\to m^\varnothing} \Phi(k)=c_F(0)-m^\varnothing=0$. Therefore, $\Phi(k)>0$ for all $k\in[ K, m^\varnothing)$.

In contrast, when $k<K$, then $\bar u_k> m^\varnothing$, which implies that $c_{G^\varnothing}(\bar u_k)=0$. In this case $\Phi(k)=c_F(\underline u_k)-\big(\bar u_k-\underline u_k\big)$. As $\underline u_k$ is strictly decreasing in $k$ and $c_F(u)$ is strictly decreasing in $u$, $c_F(\underline u_k)$ is strictly increasing in $k$. Furthermore, by \autoref{lemma:surplus-compstat-cost}, $\bar u_k-\underline u_k$ is strictly decreasing in $k$. Thus $\Phi(k)$ is a strictly increasing function when $k\in(0,K)$. Additionally, 
\[
\lim_{k\to 0} \Phi(k)=c_F(m^\varnothing)-(1-m^\varnothing)=\int^1_{m^\varnothing}\big(1-F(m)\big)dm-(1-m^\varnothing)<0.
\]
Since $\Phi$ is continuous and strictly increasing with $\Phi(K)>0$ and $\lim_{k\to 0}\Phi(k)<0$, there exists a unique $k^*\in (0,K)$ such that $\Phi(k^*)=0$, and $\Phi(k)>0$ for all $k>k^*$ and $\Phi(k)<0$ for all $k<k^*$ (see \autoref{fig:phi}). This concludes Step 4, and the proof of the lemma. 
\end{proof}
\bigskip

In other words, \autoref{lemma:monopoly-fixed-point} states that $\cl\big(\mathcal{S}(1,k)\big)=\{(\bar u_k-\underline u_k, \underline u_k)\}$ if $k>k^*$ and $\cl\big(\mathcal{S}(1,k)\big)=\mathcal{F}(1,k)$ if $k\leq k^*$. Moreover, the proof of \autoref{lemma:monopoly-fixed-point} shows that the inf-sup in \eqref{eq:min-max} is always attained, i.e., for any $k\in(0,m^\varnothing)$, there exists a pair $(\beta^k, \psi^k)$ with $\beta^k$ sequentially rational with respect to $\psi^k$ and $Im_{\psi^k}\subseteq \cl(\mathcal{S}(1,k))$ such that $\sup_{o\in \mathcal{O}}V_1(o, \beta^k,\psi^k)=\underline{\mathcal{V}}(\cl(\mathcal{S}(1,k)))=\nu_k$. 

As a result, $\mathcal{S}(1,k)$ is closed. To see why, consider any $y\in\cl(\mathcal{S}(1,k))$, which by \autoref{lemma:monopoly-lower-bound} satisfies $y_1\geq \nu_k$. By \autoref{lemma:stationary-implement}, there exists a tuple  $(\alpha, \beta,\psi)$ such that $\alpha$ is centered on $o^*(y)$, $\beta$ is sequentially rational with respect to $\psi$,  $\psi(o^*(y), w,m)=y\in\cl(\mathcal{S}(1,k))$ for all $(w,m)\in W\times \Theta$, and $y$ is generated by $(\alpha, \beta,\psi)$. Additionally, following an off-path offer $o\neq o^*(y)$, we can choose the agent's strategy to be $\beta=\beta^k$ and the reward function to be $\psi=\psi^k$, which yields the monopolist a payoff of at most $\nu_k$, making deviations unprofitable. Hence, $y$ is supported by $(\alpha, \beta,\psi)$, and crucially, $Im_{\psi}\subseteq\cl(\mathcal{S}(1,k))$, which makes $\cl(\mathcal{S}(1,k))$ self-generating. However, $\mathcal{S}(1,k)$ is the largest self-generating set, which immediately implies that $\cl(\mathcal{S}(1,k))=\mathcal{S}(1,k)$. Consequently, $\mathcal{S}(1,k)$ is closed for all $k\in (0,m^\varnothing)$ and, therefore, Borel.

\begin{lemma}\label{lemma:measurable}
 For all $k\in (0,m^\varnothing)$, $\mathcal{S}^m(1,k)=\mathcal{S}(1,k)$. 
\end{lemma}
\noindent\begin{proof}
It suffices to show that $\mathcal{S}(1,k)\subseteq\mathcal{S}^m(1,k)$ for all $k\in (0,m^\varnothing)$. First, suppose that $k>k^*$. Let $y^{FS}\coloneqq (\bar u_k-\underline u_k, \underline u_k)$ be the full-surplus extraction payoff profile. By \autoref{lemma:monopoly-fixed-point}, $\mathcal{S}(1,k)=\{y^{FS}\}$. Note that $y^{FS}$ is supported by the tuple $(\alpha, \beta,\psi)$ such that $\alpha$ is the Dirac measure centered on $o^*(y^{FS})$, the agent's simple strategies are
\[
\beta^\mathcal{w}(1|o)=\begin{cases}
1 & \mbox{if } c_{G}(\underline u_k)-c_{G^\varnothing}(\underline u_k)\geq p\\
0 & \mbox{if } c_{G}(\underline u_k)-c_{G^\varnothing}(\underline u_k)< p
    \end{cases}
    \]
for all $o=(p,G)\in \mathcal{O}$, and
\[
\beta^\mathcal{d}(o,w,m)=\begin{cases}
0 & \mbox{if }m\geq \underline u_k\\
1 & \mbox{if } m< \underline u_k,
    \end{cases}
    \]
and the reward function is $\psi(o,w,m)=y^{FS}$ for all $(o,w,m)\in \mathcal{O}\times W\times \Theta$. 

The functions in this tuple are Borel measurable. Indeed, $G\mapsto c_G(\underline u_k)$ is continuous under the weak-* topology, so the set $\{(p,G):c_G(\underline u_k)-c_{G^\varnothing}(\underline u_k)\ge p\}$ is Borel; the remaining rules are constant or indicators of Borel threshold
sets. Hence, the tuple is universally measurable. Since $\mathcal{S}(1,k)$ is a singleton, defining the universally measurable functions $ (\bm{\alpha},\bm{\beta}, \bm{\psi})$ with $(\bm{\alpha}[y],\bm{\beta}[y], \bm{\psi}[y])=(\alpha,\beta,\psi)$ for all $y\in \mathcal{S}(1,k)$ and noting that $Im_{\bm{\psi}}=\mathcal{S}(1,k)$ yields the desired result. 

Next, consider any $k\leq k^*$. By \autoref{lemma:monopoly-fixed-point}, $\mathcal{S}(1,k)=\mathcal{F}(1,k)$. Let $y^a\coloneqq (0,\underline u_k)$ and $y^r=(0, \bar u_k)$. Consider any $y\in \mathcal{F}(1,k)=\mathcal S(1,k)$. Clearly, $\{y, y^a, y^r\}\subseteq \mathcal{F}(1,k)$. It is also a Borel set. We proceed by showing that $\{y, y^a, y^r\}$ is measurably self-generating, and hence $\{y, y^a, y^r\}\subseteq \mathcal{S}^m(1,k)$. Since $y\in\mathcal S(1,k)$ is arbitrary, this suffices to conclude that $\mathcal{S}(1,k)\subseteq \mathcal S^m(1,k)$ as desired.

To that end, define the functions $(\bm{\alpha}, \bm{\beta}, \bm{\psi})$ such that for all $y'\in\{y, y^a, y^r\}$, the monopolist's simple strategy $\bm{\alpha}(y')$ is the Dirac measure centered on $o^*(y')$, the agent's simple strategies are
\[
\bm{\beta^\mathcal{w}}(1|y',o)=\begin{cases}
1 & \mbox{if } o=o^*(y')\\
0 & \mbox{if } o\neq o^*(y')
    \end{cases} \quad \text{ and } \quad \bm{\beta^\mathcal{d}}(y',o,w,m)=\begin{cases}
0 & \mbox{if } o=o^*(y'), w=1,  \text{ and } m\geq y'_2\\
0 & \mbox{if } o\neq o^*(y'), w=1,  \text{ and } m\geq y^a_2\\
1 & \mbox{otherwise, }
    \end{cases}
    \]
    and
    \[
\bm{\psi}(y',o,w,m)=\begin{cases}
y' & \mbox{if } o=o^*(y')\\
y^a & \mbox{if } o\neq o^*(y') \text{ and } w=1\\
y^r & \mbox{if } o\neq o^*(y') \text{ and } w=\emptyset.
    \end{cases}
    \]
The functions $(\bm{\alpha},\bm{\beta}, \bm{\psi})$ are Borel measurable as, once again, these are simple functions based on Borel-measurable sets. Hence,  $(\bm{\alpha},\bm{\beta}, \bm{\psi})$ are universally measurable. It is also straightforward to check that each $y'\in\{y,y^a,y^r\}$ is supported by $(\bm{\alpha}[y'],\bm{\beta}[y'], \bm{\psi}[y'])$ and that $Im_{\bm{\psi}}\subseteq \{y,y^a,y^r\}$, and thus, $\{y,y^a,y^r\}$ is measurably self-generating.  
\end{proof}
\bigskip

Combining \autoref{lemma:equilibrium-characterization} and \autoref{lemma:measurable}, we now conclude that $\mathcal{S}(1,k)=\mathcal{E}(1,k)$, which gives us the desired result of the proposition.
\end{proof}
\bigskip

\noindent\begin{proof}[Proof of \autoref{thm:bertrand}]
Fix any $k\in (0,m^\varnothing)$ and $n>1$. We prove the proposition by showing that $\{y^B\}$ is measurably self-generating. To that end, consider the tuple $(\alpha, \beta, \psi)$ such that
\begin{enumerate}[$(\roman*)$]
    \item $\alpha_i\big((0, F)\big)=1$ for all $i\in N$, 
    \item $\beta$ is sequentially rational with respect to $\psi$, and 
    \item $\psi\big(O, w, m)=y^B$ for all $(O, w, m)\in \mathcal{O}^n\times W\times\Theta$. 
\end{enumerate}
Notice that the specified tuple $(\alpha, \beta, \psi)$ implies the agent searches in each period by observing full information for free. As such, the tuple $(\alpha, \beta, \psi)$ generates the payoff profile $y^B$. Hence, $(\alpha, \beta,\psi)$ satisfies Points $(a)$ and $(c)$ of \autoref{def:self-generation}. Additionally, Broker $i\in N$ gets a payoff of 0 by offering the on-path contract $(0,F)$. If she instead deviates by offering $o_i\coloneqq (p_i, G_i)\neq (0,F)$ while the other $n-1$ brokers offer $(0,F)$, then it is optimal for the agent to reject Broker $i$'s contract, which again yields her a payoff of 0. Hence, no broker has a profitable deviation, i.e., Point $(b)$ of \autoref{def:self-generation} is satisfied. Thus, $y^B$ is supported by the tuple $(\alpha, \beta,\psi)$. By construction, $Im_\psi=\{y^B\}$. Hence, $\{y^B\}$ is self-generating. Finally, the tuple $(\alpha, \beta,\psi)$ is Borel measurable.\footnote{A similar argument to the one used in the proof of \autoref{lemma:measurable} for the case when $k>k^*$ establishes the measurability property.} Thus, $\{y^B\}$ is measurably self-generating, which by \autoref{lemma:equilibrium-characterization}, implies $y^B\in\mathcal{E}(n,k)$ as desired.\end{proof}
\bigskip

\noindent\begin{proof}[Proof of \autoref{thm:competition}]
Fix any $k\leq k^*$ and $n>1$. We prove \autoref{thm:competition} by showing that for any $y\in\mathcal{F}(n,k)$, the set $\{y, y^A, y^B\}$ is measurably self-generating, which suffices to establish $\mathcal{F}(n,k)\subseteq \mathcal{S}^m(n,k)$. To that end, consider any $y'\in\{y,y^A,y^B\}$ and the functions $(\bm{\alpha}, \bm{\beta}, \bm{\psi})$ such that
\begin{enumerate}[$(\roman*)$]
    \item $\bm{\alpha_i}[y']\big(o^*(y')\big)=1$ for all $i\in N$, 
    \item $\bm{\beta}[y']$ is sequentially rational with respect to $\bm{\psi}[y']$, and 
    \item $\bm{\psi}$ is given by 
    \[
    \label{eq:cont}
    \tag{8}
\bm{\psi}\big(y', O, w, m)=\begin{cases}
 y'   &  \mbox{if }  o_i=o^*(y') \text{ for all } i\in N\\
y^A   &  \mbox{if }  o_i\neq o^*(y') \text{ for some } i\in N \text{ and } w= i\\
y^B   &  \mbox{if }  o_i\neq o^*(y') \text{ for some } i\in N \text{ and } w\neq i.
\end{cases}
\]
\end{enumerate}
Notice that the specified functions $(\bm{\alpha}, \bm{\beta}, \bm{\psi})$ are Borel measurable.\footnote{A similar argument to the one used in the proof of \autoref{lemma:measurable} for the case when $k\leq k^*$ establishes the measurability property.}  By construction, $Im_{\bm{\psi}}=\{y, y^A, y^B\}$. Moreover, for all  $y'\in\{y,y^A,y^B\}$, the specified tuple $(\bm{\alpha}[y'], \bm{\beta}[y'], \bm{\psi}[y'])$ satisfies the conditions of \autoref{lemma:stationary-implement}, and therefore generates the payoff profile $y'$. Thus, Points $(a)$ and $(c)$ of \autoref{def:self-generation} are satisfied. Hence, to show that $\{y,y^A,y^B\}$ is measurably self-generating, it suffices to show that no broker has a profitable deviation (Point $(b)$ of \autoref{def:self-generation}). 

Suppose Broker $i$ deviates by offering $o_i\coloneqq (p_i, G_i)\neq o^*(y')$ while the other $n-1$ brokers offer $o^*(y')$. Given the reward function $\bm{\psi}$ specified in \eqref{eq:cont}, the agent's expected value from accepting $o_i$ is given by $\underline u_k+c_{G_i}(\underline u_k)-k-p_i.$ Conversely, if the agent optimally accepts either the null offer or $o^*(y')$ from one of the non-deviating brokers, his payoff is $\max\{\bar u_k+c_{G^{\mathbf x_k(\norm{y'})}}(\bar u_k)-P(y')-k, \bar u_k+c_{G^{\varnothing}}(\bar u_k)-k\}.$ Therefore, it is sequentially rational for the agent to accept $o_i$ only if 
\begin{align*}
\underline u_k+c_{G_i}(\underline u_k)-p_i\geq& \max\{\bar u_k+c_{G^{\mathbf{x}_k(\norm{y'})}}(\bar u_k)-P(y'), \bar u_k+c_{G^\varnothing}(\bar u_k)
\}.
\end{align*}
Indeed, the agent's willingness-to-pay for signal $G_i\in\mathcal{G}(F)$ in this setting is given by 
\[
\label{eq:wtp}
\tag{9}
\mathcal{P}_k(G_i;y')\coloneqq \big(\underline u_k-\bar u_k+c_{G_i}(\underline u_k)- \max\{c_{G^{\mathbf{x}_k(\norm{y'})}}(\bar u_k)-P(y'), c_{G^\varnothing}(\bar u_k)\}\big)^+,
\]
and it is sequentially rational for the agent to accept $o_i=(p_i, G_i)$ instead of $o^*(y')$ or the null offer only if $p_i\leq \mathcal{P}_k(G_i;y')$. Furthermore, $\mathcal{P}_k(G_i;y')$ serves as an upper bound on Broker $i$'s deviation payoff from offering $o_i$ because $\mathbb{E}_{\bm{\alpha_{-i}}[y']}\big[ V_i\left(o_i, O_{-i}, \bm{\beta}[y'],\bm{\psi}[y'] \big) \right]=p_i \bm{\beta^\mathcal{w}}[y'](i|o_i,O^*(y')_{-i})\leq \mathcal{P}_k(G_i;y').$ The following lemma allows us to further bound $\mathcal{P}_k(G;y')$ for all $G\in\mathcal{G}(F)$, thereby providing a uniform upper bound on any deviation payoff that a broker can attain.

\begin{lemma}\label{lemma:no-wtp}
For any $k\leq k^*$, $\mathcal{P}_k(G;y)=0$ for all $y\in\mathcal{F}(n,k)$ and all $G\in \mathcal{G}(F)$. 
\end{lemma}
\noindent\begin{proof}
Recall that $c_{G'}\leq c_{G''}$ pointwise for any $G', G''\in\mathcal{G}(F)$ such that $G'$ is a mean-preserving contraction of $G''$. Consequently, for all $y\in\mathcal{F}(n,k)$, $\mathcal{P}_k(G';y)\leq \mathcal{P}_k(G'';y)$. Thus, to prove the lemma, it suffices to show that $\mathcal{P}_k(F;y)=0$ for any $k\leq k^*$ and any $y\in\mathcal{F}(n,k)$. To see that this condition holds, note $\mathcal{P}_k(F;y)\geq 0$ by construction. Additionally,  
\begin{align*}
\mathcal{P}_k(F;y)&=\Big(\Phi(k)-\big(c_{G^{\mathbf{x}_k(\norm{y})}}(\bar u_k)-c_{G^\varnothing}(\bar u_k)-P(y)\big)^+\Big)^+\leq \Phi(k)^+=0
\end{align*}

\noindent where the first equality follows from 
\eqref{eq:phi}, and the last equality follows because $\Phi(k)\leq 0$ for all $k\leq k^*$ (see Step 4 in the proof of \autoref{lemma:monopoly-fixed-point}).\end{proof}\bigskip

When $k\leq k^*$, \autoref{lemma:no-wtp} implies that $\mathbb{E}_{\bm{\alpha_{-i}}[y']}\left[ V_i\big(o_i, O_{-i}, \bm{\beta}[y'],\bm{\psi}[y'] \big) \right]=0$. Thus, no broker has a profitable deviation, which implies that Point $(b)$ of \autoref{def:self-generation} is also satisfied. Consequently, $\mathcal{F}(n,k)\subseteq\mathcal{S}^m(n,k)$, which by \autoref{lemma:equilibrium-characterization}, implies $\mathcal{F}(n,k)=\mathcal{E}(n,k)$ as desired. \end{proof}
\bigskip

\noindent \begin{proof}[Proof of \autoref{thm:multiple}]
For all $k\leq k^*$ and $n>1$, the desired result trivially follows: for all $\epsilon>0$, $\mathcal{F}^\epsilon(n,k)\subseteq \mathcal{F}(n,k)$, and \autoref{thm:competition} establishes that $\mathcal{F}(n, k)=\mathcal{E}(n,k)$. 

Next, let us consider the case when the search cost exceeds $k^*$. Our proof proceeds to show that there exists a threshold $k^{**}\in(k^*, m^\varnothing)$ and a constant $\epsilon>0$ such that for each $k\in (k^*, k^{**}]$, the payoff set $\mathcal{F}^\epsilon(n,k)\subseteq\mathcal{S}^m(n,k)$, which establishes the desired result that $\mathcal{F}^\epsilon(n,k)\subseteq \mathcal{E}(n,k)$. To that end, we first make the following observation. 
\begin{lemma}\label{lemma:epsilon}
There exists a $k^{**}\in (k^*, m^\varnothing)$ and a constant $\epsilon>0$ such that for all $k\in (k^*,  k^{**}]$, for all $n>1$, for all $y\in \mathcal{F}^\epsilon(n,k)$, and for all $G\in\mathcal{G}(F)$, $\mathcal{P}_k(G;y)=0$.
\end{lemma}
\noindent \begin{proof}
For each $k\in (0,m^\varnothing)$, let $\epsilon_k\coloneqq c_{G^{\mathbf{x}_k(\underline u_k)}}(\bar u_k)$ and notice that 
\begin{align*}
 c_{G^{\mathbf{x}_k(\underline u_k)}}(\bar u_k)&=c_F(\mathbf{x}_k(\underline u_k))+(1-F(\mathbf{x}_k(\underline u_k)))(\mathbf{x}_k(\underline u_k)-\bar u_k)>c_F(\mathbf{x}_k(\underline u_k))>0,
\end{align*}
where the equality follows from \autoref{lemma:dist-function} and the  inequalities follow because $F$ is absolutely continuous and $\bar u_k<\mathbf{x}_k(\underline u_k)<1$ as stated in \autoref{lemma:pass-fail-sufficiency}. Let $\epsilon\coloneqq \inf_{k\in(k^*, m^\varnothing)}\epsilon_k$, and notice that $\epsilon>0$. 

Consider a search cost $k\in(k^*, K]$, where $K$ is the unique value of the search cost for which $\bar u_{K}=m^\varnothing$.\footnote{Step 4 of the proof of \autoref{lemma:monopoly-fixed-point} shows that a unique $K$ satisfying the desired property exists and that $k^*<K$.} For all $n>1$ and all $y\in \mathcal{F}^{\epsilon}(n,k)$, we have
\begin{align*}
\mathcal{P}_k(F;y)&=\Big(\Phi(k)-\big(c_{G^{\mathbf{x}_k(\norm{y})}}(\bar u_k)-c_{G^\varnothing}(\bar u_k)-P(y)\big)^+\Big)^+\\[6pt]
&=\Big(\Phi(k)-\big(c_{G^{\mathbf{x}_k(\norm{y})}}(\bar u_k)-P(y)\big)^+\Big)^+\\[6pt]
&\leq \Big(\Phi(k)-c_{G^{\mathbf{x}_k(\norm{y})}}(\bar u_k)+P(y)\Big)^+\\[6pt]
&=\Big(\Phi(k)-c_{G^{\mathbf{x}_k(\norm{y})}}(\bar u_k)+\big(1-F(\mathbf{x}_k(\norm{y}))\big)\big(\norm{y}-y_{n+1}\big)\Big)^+\\[6pt]
&\leq \Big(\Phi(k)-\epsilon F(\bar u_k)\Big)^+
\end{align*}
where the first equality is established in the Proof of \autoref{lemma:no-wtp}, the second equality follows because $\bar u_k\geq m^\varnothing$ for all $k\leq  K$, implying that $c_{G^\varnothing}(\bar u_k)=(m^\varnothing-\bar u_k)^+=0$, the last equality follows by the alternative characterization of $P(y)$ given in \eqref{eq:a3}, and the last inequality follows because \autoref{lemma:order} implies $c_{G^{\mathbf{x}_k(\norm{y})}}(\bar u_k)\geq c_{G^{\mathbf{x}_k(\underline u_k)}}(\bar u_k)=\epsilon_k\geq \epsilon$, because $\norm{y}-y_{n+1}\leq \epsilon$, and because $\mathbf x_k(\cdot)$ is a bounded from below by $\bar u_k$ as established in \autoref{lemma:pass-fail-sufficiency}.

The mapping $k\mapsto \Phi(k)-\epsilon F(\bar u_k)$ is continuous. Additionally, for all $k\leq k^*$, $\Phi(k)-\epsilon F(\bar u_{k})<0$ because $\Phi(k)\leq 0$ for all $k\leq k^*$ by construction (see Step 4 of the proof of \autoref{lemma:monopoly-fixed-point}) while $\epsilon F(\bar u_{k})>0$ for all $k\in (0,m^\varnothing)$. Hence, there exists a $k^{**}\in (k^*, K]$ such that $\Phi(k)-\epsilon F(\bar u_{k})\leq 0$ for all $k\leq k^{**}$. Consequently, for all $k\in(k^*, k^{**}]$, all $n>1$, and all $y\in\mathcal{F}^{\epsilon}(n,k)$, we have $\mathcal{P}_k(F;y)=0$. From the proof of \autoref{lemma:no-wtp}, $\mathcal{P}_k(F;\cdot)\geq \mathcal{P}_k(G;\cdot)\geq 0$ for all $G\in\mathcal{G}(F)$. Therefore, for all $k\in(k^*, k^{**}]$, all $n>1$, all $y\in\mathcal{F}^{\epsilon}(n,k)$, and all $G\in\mathcal{G}(F)$, we have $\mathcal{P}_k(G;y)=0$ as desired. 
\end{proof}
\bigskip

With this observation in hand, we proceed by fixing an $\epsilon>0$ such that the conclusion of \autoref{lemma:epsilon} holds. Fix some $k\in(k^*,  k^{**}]$ and some $n>1$.  For any payoff profile $y\in\mathcal{F}^\epsilon(n,k)$,  consider the set $\{y, y^A, y^B\}$. We shall show that $\{y, y^A, y^B\}$ is measurably self-generating. Since the choice of $y\in\mathcal{F}^\epsilon(n,k)$ is arbitrary, this suffices to establish that $\mathcal{F}^\epsilon(n,k)\subseteq\mathcal{S}^m(n,k)$.

To that end, take the functions $(\bm{\alpha}, \bm{\beta}, \bm{\psi})$ prescribed in the proof of \autoref{thm:competition}, which are Borel measurable. By construction, $Im_{\bm\psi}=\{y, y^A, y^B\}$. Moreover, for each $y'\in\{y,y^A,y^B\}$, the tuple $(\bm{\alpha}[y'], \bm{\beta}[y'], \bm{\psi}[y'])$ satisfies the conditions of \autoref{lemma:stationary-implement}, and therefore generates the payoff profile $y'$. Hence, Points $(a)$ and $(c)$ of \autoref{def:self-generation} are satisfied. Finally, if Broker $i$ deviates by offering $o_i=(p_i, G_i)\neq o^*(y')$ while the other $n-1$ brokers offer $o^*(y')$, then her deviation payoff is
\begin{align*}
\mathbb{E}_{\bm{\alpha_{-i}}[y']}\big[ V_i\left(o_i, O_{-i}, \bm{\beta}[y'],\bm{\psi}[y'] \big) \right]=p_i \bm{\beta^\mathcal{w}}[y'](i|o_i,O^*(y')_{-i})\leq \mathcal{P}_k(G_i;y')= 0,
\end{align*}
where the last equality follows from \autoref{lemma:epsilon}.
Thus, no broker has a profitable deviation, which implies Point $(b)$ of \autoref{def:self-generation}. Consequently, $\mathcal{F}^\epsilon(n,k)\subseteq\mathcal{S}^m(n,k)\subseteq\mathcal{E}(n,k)$.
\end{proof}

\subsection{Proofs from \Cref{prelim}}

\noindent\begin{proof}[Proof of \autoref{lemma:surplus-compstat-cost}]
To prove the first statement, fix some $G\in\mathcal{G}(F)$ and consider $k', k''\in (0, m^\varnothing )$ such that $k'<k''$. Since $c_G(\mathcal{u}_{k'}(G))=k'<k''=c_G(\mathcal{u}_{k''}(G))$ by construction, and since $c_G$ is a decreasing function, we have that $\mathcal{u}_{k'}(G)>\mathcal{u}_{k''}(G)$. Thus, for each $G\in\mathcal G(F)$, the mapping $k\to \mathcal u_k(G)$ is continuous and strictly decreasing. 

To prove the second statement, note that $\underline u_k=m^\varnothing -k$, and hence $\partial \underline u_k/\partial k=-1$. Additionally,  $c_F(\bar u_k)=k$ for all $k\in (0,m^\varnothing)$. Recall that $c_F(\cdot)$ is differentiable since $F$ is assumed to be absolutely continuous. Thus, we can implicitly differentiate to conclude 
\[
\frac{\partial \bar u_k}{\partial k}=\frac{-1}{1-F(\bar u_k)}<-1,
\]
where the inequality follows from the fact that $\bar u_k>0$ for all $k\in (0,m^\varnothing)$; otherwise, we would have $c_F(0)=k$ by \eqref{eq:1}, and $c_F(0)=m^\varnothing$ by construction, which violates the assumption that $k<m^\varnothing$. Therefore, $\bar u_k-\underline u_k$ is decreasing in $k$. \end{proof}
\bigskip

\noindent \begin{proof}[Proof of \autoref{lemma:equilibrium-characterization}]
For the entirety of this proof, fix any $n\geq 1$ and any $k\in (0,m^\varnothing)$. \medskip

\noindent \underline{$\mathcal{E}(n,k)\subseteq \mathcal{S}(n,k)$ direction:} Consider any $y\in \mathcal{E}(n,k)$. By definition, there exists some equilibrium strategy profile $\sigma$ such that  $y=((\mathbf V_{i}(\sigma))_{i\in N}, \mathbf U(\sigma))$. For each Broker $i\in N$, we can decompose her payoff, $\mathbf V_{i}(\sigma)$, into a sum of her payoff in the current period
\[
\mathbb{E}_{(\sigma_j|h^1)_{j\in N}}\left[p_{i}\sigma_A^\mathcal{w}(i|h^1, O)\right]
\]
and her expected continuation payoff 
\[
\mathbb{E}_{(\sigma_j|h^1)_{j\in N}}\left[\sum_{w\in W}\sigma_A^\mathcal{w}(w|h^1, O)\int_\Theta \mathbf V_{i 2}(\sigma|h^1\cup\{O, w, m\})\sigma_A^\mathcal{d}(h^1,O, w, m)dG_{w}(m)\right],
\]
where $\mathbb{E}_{(\sigma_j|h^1)_{j\in N}}$ is the expectation taken over $\mathcal{O}^n$ with respect to the probability induced by  $(\sigma_j(\cdot|h^1))_{j\in N}$. Similarly, we can decompose the agent's payoff, $U(\sigma)$, into his payoff in the current period 
\[
-k+\mathbb{E}_{(\sigma_j|h^1)_{j\in N}}\left[\sum_{w\in W}\sigma_A^\mathcal{w}(w|h^1, O)\left(-p_{w}+\int_\Theta m \hspace*{.15em} \big(1-\sigma^\mathcal{d}_A(h^1, O, w, m)\big)dG_{w}(m)\right)\right]
\]
and his expected continuation payoff
\[
\mathbb{E}_{(\sigma_j|h^1)_{j\in N}}\left[\sum_{w\in W}\sigma_A^\mathcal{w}(w|h^1, O)\int_\Theta \mathbf U_{2}(\sigma|h^1\cup\{O, w, m\})\sigma_A^\mathcal{d}(h^1,O, w, m)dG_{w}(m)\right].
\]

Define simple strategies $(\alpha, \beta)$ and a continuation payoff function $\psi$ as follows: for all $(O,w,m)\in\mathcal{O}^n\times W\times \Theta$, let
\begin{itemize}
    \item $\alpha_i(\cdot)=\sigma_i(\cdot|h^1)$ for all $i\in N$,
    \item $\beta^\mathcal{w}(\cdot|O)=\sigma^\mathcal{w}_A(\cdot|h^1, O)$,
    \item $\beta^\mathcal{d}(\cdot|O, w,m)=\sigma^\mathcal{d}_A(h^1, O, w,m)$,
    \item $\psi_i(O,w,m)=\mathbf V_{i 2}(\sigma|h^1\cup\{O,w,m\})$  for all $i\in N$, and 
    \item $\psi_{n+1}(O,w,m)= \mathbf U_{2}(\sigma|h^1\cup\{O,w,m\})$.
\end{itemize}

Clearly, $\mathbb{E}_{\alpha}[V_i(O, \beta, \psi)]=\mathbf V_i(\sigma)$ for all $i\in N$, and $\mathbb{E}_\alpha[U(O,\beta, \psi)]=\mathbf U(\sigma)$. Thus, $(\alpha, \beta, \psi)$ generates the payoff profile $y=((\mathbf V_{i}(\sigma))_{i\in N}, \mathbf U(\sigma))$, and the defined simple strategies and continuation payoff function satisfy Point $(a)$ of \autoref{def:self-generation}. Furthermore, taking $\psi$ as a given, the simple strategy profile $(\alpha, \beta)$ satisfies Points $(b)$ and $(c)$ of the definition; otherwise, there would be a one-shot deviation from $\sigma$ that improves the payoff of either a broker or the agent, which would contradict the assumption that $\sigma$ is an equilibrium strategy profile. Thus, $y$ is supported by the tuple $(\alpha, \beta,\psi)$. 

Thus, all that remains is to show is that the the constructed tuple $(\alpha,\beta,\psi)$ that supports $y$ is such that $Im_\psi\subseteq \mathcal{E}(n,k)$. Since $\sigma$ is a PBE, for any history $h_t$, the strategy $\sigma\vert h_t$ is an equilibrium in the game starting at $h_t$. Thus, following any $(O,w,m)\in\mathcal{O}^n\times W\times \Theta$ and the resulting period-2 history $h^2=h^1\cup\{O,w,m\}\in H^2$, the payoff profile $((\mathbf V_{i2}(\sigma|h^2))_{i\in N}, \mathbf U_{2}(\sigma|h^2))$ is in $\mathcal{E}(n,k)$. Consequently, $Im_\psi\subseteq \mathcal{E}(n,k)$. As the choice of $y\in\mathcal{E}(n,k)$ was arbitrary, we have established that $\mathcal{E}(n,k)$ satisfies \autoref{def:gen}. Hence, $\mathcal{E}(n,k)\subseteq \mathcal{S}(n,k)$, as desired.
\medskip

\noindent \underline{$\mathcal{S}^m(n,k)\subseteq \mathcal{E}(n,k)$ direction:} Consider any $y\in \mathcal{S}^m(n,k)$. By definition, there exists some Borel set $E\subseteq \mathcal{S}^m(n,k)$ that contains $y$. Moreover, since $E$ is measurably self-generating, there exist universally measurable functions $(\bm\alpha,\bm\beta,\bm\psi)$ such that
\begin{enumerate}[$(\roman*)$]
    \item Each $y'\in E$ is supported by the tuple $(\bm\alpha[y'], \bm\beta[y'], \bm\psi[y'])$, and 
    \item $Im_{\bm\psi}\subseteq E$. 
\end{enumerate}

For $t=1$, define $Y(h^1)\coloneqq y$. For each $t\geq 1$,  $h^{t}\in H$, and $(O_t, w_t,m_t)\in\mathcal{O}^n\times W\times\Theta$, define $Y(h^{t+1})\coloneqq \bm{\psi}[Y(h^{t})](O_t, w_t, m_t)$ where $h^{t+1}=h^{t}\cup\{O_t, w_t, m_t\}$. By construction, $Y(h)\in E$  for all $h\in H$ because $Im_{\bm{\psi}}\subseteq E$. Furthermore, the mapping $h\mapsto Y(h)$ is universally measurable as it is derived from countably many compositions of $\bm\psi$---a universally measurable function---with itself.

Let us now construct a strategy profile $\sigma=((\sigma_i)_{i\in N}, \sigma^{\mathcal{w}}_A, \sigma^{\mathcal{d}}_A)$ where 
\begin{itemize}
    \item For each $i\in N$,  $\sigma_i:H\to\Delta(\mathcal{O})$ is given by $\sigma_i(\cdot|h^t)=\bm{\alpha}_i[Y(h^t)](\cdot)$, 
    \item $\sigma_A^\mathcal{w}:H\times \mathcal{O}^n\to \Delta(W)$ is given by $\sigma_A^\mathcal{w}(\cdot|h^t, O_t)=\bm{\beta^\mathcal{w}}[Y(h^t)](\cdot|O_t)$, and  
    \item $\sigma_A^\mathcal{d}:H\times \mathcal{O}^n\times W\times \Theta\to [0,1]$ is given by $\sigma_A^\mathcal{d}(h^t, O_t, w_t, m_t)=\bm{\beta^\mathcal{d}}[Y(h^t)](O_t, w_t, m_t)$.
\end{itemize}
This constructed strategy profile $\sigma$ is universally measurable as, once again, it is the composition of universally measurable functions. 

Fix a non-terminal history $h^t\in H$. Since $Y(h^t)\in E$,  $Y(h^t)$ is supported by the tuple $\big((\bm{\alpha}_i[Y(h^t)])_{i\in N}, \bm{\beta}[Y(h^t)], \bm{\psi}[Y(h^t)]\big)$. In particular, players cannot improve their payoff from $Y(h^t)$ with a one-shot deviation, and $Y(h^t)$ is generated by the tuple. Given the constructed strategies, we can recursively express the payoff of Broker $i\in N$ as
\[
Y_i(h^t)=\mathbb{E}_{\sigma(h^t)}\Big[ p_{it}\mathbbm{1}_{[w_t=i]}+(1-d_t)Y_i(h^{t+1})\Big],
\]
and the agent's payoff as 
\[
Y_{n+1}(h^t)=\mathbb{E}_{\sigma(h^t)}\Big[ d_t\theta_t-k-p_{w_t t}+(1-d_t)Y_{n+1}(h^{t+1})\Big].
\]
Iterating this identity until period $T>t$ for Broker $i\in N$ gives 
\begin{equation}
\label{eq:broker-recursive}
\tag{A1}
    Y_i(h^t)=\mathbb{E}_{\sigma(h^t)}\left[ \sum_{\tau=t}^{ T\wedge \tilde T_\sigma}p_{i\tau}\mathbbm{1}_{[w_\tau=i]}+\mathbbm{1}_{[\tilde T_\sigma>T]}Y_i(h^{T+1})\right],
\end{equation}
where $\tilde T_\sigma$ is the stopping time induced under $\sigma$. Similarly, iterating the agent's payoff gives 
\begin{equation}
\label{eq:agent-recursive}
\tag{A2}
    Y_{n+1}(h^t)=\mathbb{E}_{\sigma(h^t)}\left[ \theta_{\tilde T_\sigma}\mathbbm{1}_{[\tilde T_\sigma\leq  T]}-\sum_{\tau=t}^{T\wedge \tilde T_\sigma}(k+p_{w_\tau \tau})+\mathbbm{1}_{[\tilde T_\sigma> T]}Y_{n+1}(h^{T+1})\right].
\end{equation}
In order to show that $Y_i(h^t)=\mathbf{V}_{it}(\sigma|h^t)$ as defined in \eqref{eq:broker-payoff} for all $i\in N$, and that $Y_{n+1}(h^t)=\mathbf{U}_{t}(\sigma|h^t)$ as defined in \eqref{eq:agent-payoff}, we need to show that the last terms in \eqref{eq:broker-recursive} and \eqref{eq:agent-recursive} vanish. 

Let us focus on the agent's payoff. Since $E\subseteq \mathcal{F}(n,k)$ and $Y(h)\in E$ for each non-terminal history $h\in H$, the agent's continuation payoff satisfies $\underline u_k\leq Y_{n+1}(h)\leq \bar u_k$. Moreover, the agent's match value is bounded above by $\theta=1$ and price is bounded below by zero. Hence, from \eqref{eq:agent-recursive}, we get
\begin{align*}
\underline u_k\leq Y_{n+1}(h^t)\leq 1-k\mathbb{E}_{\sigma(h^t)}\left[T\wedge \tilde T_\sigma-t+1\right]+\bar u_k\\[6pt]
\implies  \mathbb{E}_{\sigma(h^t)}\left[T\wedge \tilde T_\sigma-t+1\right]\leq \frac{1+\bar u_k-\underline u_k}{k}.
\end{align*}
Since the last inequality holds for all $T>t$, we have 
\[
\mathbb{E}_{\sigma(h^t)}\left[\tilde T_\sigma-t+1\right]\leq \frac{1+\bar u_k-\underline u_k}{k},
\]
which implies that $\mathbb{P}_{\sigma(h^t)}(\tilde T_\sigma>T)\to 0$ as $T$ grows large. Thus, the last term in \eqref{eq:agent-recursive} becomes
\[
\mathbb{E}_{\sigma(h^t)}\left[\mathbbm{1}_{[\tilde T_\sigma>T]}Y_{n+1}(h^{T+1})\right]\leq \bar u_k\cdot \mathbb{P}_{\sigma(h^t)}(\tilde T_\sigma>T),
\]
with the right-hand side of the inequality going to zero as $T\to \infty$. We can therefore conclude that $Y_{n+1}(h^t)=\mathbf{U}_{t}(\sigma|h^t)$ as defined in \eqref{eq:agent-payoff}. A symmetric argument applies for the brokers establishing  $Y_i(h^t)=\mathbf{V}_{it}(\sigma|h^t)$ as defined in \eqref{eq:broker-payoff} for all $i\in N$. 

Thus far, we have shown that the constructed profile of strategies $\sigma$ satisfies $(i)$ $y_i=\mathbf{V}_{i}(\sigma)$ for all $i\in N$, $(ii)$ $y_{n+1}=\mathbf{U}(\sigma)$,  and $(iii)$ no player can improve her payoff via a one-shot deviation, and thus, no player can improve her payoff by deviating away from $\sigma$ finitely often.

It remains to argue that, following any history, no player can improve her payoff by deviating infinitely often.\footnote{The payoffs in our model do not satisfy the usual ``continuous at infinity" assumption, so we cannot simply appeal to the one-shot deviation principle.} Suppose, for the sake of contradiction, the agent can deviate to strategy $\tilde \sigma_A$ following some history $h^t\in H$ and earn a  payoff of \[
\mathbf{U}_t((\sigma_i)_{i\in N}, \tilde \sigma_A|h^t)-\mathbf{U}_t(\sigma|h^t)>\Delta\]
for some $\Delta>0$. Let $\tilde \sigma\coloneqq ((\sigma_i)_{i\in N}, \tilde \sigma_A)$ denote the deviation strategy profile. 

Since the agent's match value is bounded above by $\theta=1$, prices are non-negative, and the agent pays $k>0$ in each search period, this deviation strategy can yield a finite non-negative payoff only if the induced stopping time $\tilde T_{\tilde \sigma}$ is finite almost surely. Hence, we can choose a finite $T$ such that 
\[
\mathbb{P}_{\tilde\sigma(h^t)}(\tilde T_{\tilde\sigma}>T)<\frac{\Delta}{2}.
\]

Let $\sigma'_A$ coincide with $\tilde\sigma_A$ through period $T$ and with $\sigma_A$ afterward. Thus, the only histories at which $\tilde\sigma_A$ and $\sigma'_A$ differ are those with $\tilde T_{\tilde\sigma}>T$. At such histories, the maximum continuation payoff under $\tilde\sigma_A$ is at most 1, while the continuation payoff under $\sigma_A$ is non-negative. Therefore,
\begin{align*}
\mathbf{U}_t(\tilde \sigma|h^t)-\mathbf{U}_t(\sigma'|h^t)\leq \mathbb{P}_{\tilde\sigma(h^t)}(\tilde T_{\tilde\sigma}>T)<\frac{\Delta}{2}.
\end{align*}
Consequently,  
\begin{align*}
\mathbf{U}_t(\sigma'|h^t)-\mathbf{U}_t(\sigma|h^t)>\mathbf{U}_t(\tilde \sigma|h^t)-\frac{\Delta}{2}-\mathbf{U}_t(\sigma|h^t)>\frac{\Delta}{2},
\end{align*}
so $\sigma_A'$ is a profitable deviation for the agent. However, $\sigma_A'$ differs from $\sigma_A$ finitely often, which contradicts the fact that no one player can improve payoffs by deviating from $\sigma$ finitely often.  A similar argument establishes that no broker can profitably deviate. 

We thus conclude that the strategy profile $\sigma$ is sequentially rational, and therefore, an equilibrium. Finally, noting that $y_i=\mathbf{V}_{i}(\sigma)$ for all $i\in N$ and $y_{n+1}=\mathbf{U}(\sigma)$, we conclude that $y\in\mathcal{E}(n,k)$. As $y\in \mathcal{S}^m(n,k)$ was arbitrary, we get the desired result that $\mathcal{S}^m(n,k)\subseteq\mathcal{E}(n,k)$.
\end{proof}
\bigskip

\noindent\begin{proof}[Proof of \autoref{lemma:pass-fail-sufficiency}]
Fix any $k\in (0,m^\varnothing )$. We show that for each $u\in [\underline u_k, \bar u_k]$, there exists a unique value $\mathbf x_k(u)\in [\bar u_k,1)$ such that $u=\mathcal{u}_k(G^{\mathbf x_k(u)} )$. By \eqref{eq:1}, this is equivalent to showing that $c_{G^{\mathbf x_k(u)}}(u)=k$. 

To that end, consider the function $L_k:\Theta\times [\underline u_k, \bar u_k]\to \mathbb{R}$ given by $L_k(x, u)\coloneqq c_F(x)+\big(1-F(x)\big)(x-u)$. Notice that $(i)$ $L_k$ is continuous in both variables, $(ii)$ for  each $x\in \Theta$, $L_k(x,\cdot)$ is strictly decreasing in its second argument, and $(iii)$ for each $u\in [\underline u_k, \bar u_k]$, $L_k(\cdot, u)$ is strictly increasing in its first argument over the interval $(0,u)$ and strictly decreasing over $(u,1)$.\footnote{ Recalling that $F(\cdot)$ and $c_F(\cdot)$ are absolutely continuous functions, it follows that $L_k(\cdot,u)$ is absolutely continuous in its first argument, and $\frac{\partial L_k(x, u)}{\partial x}=f( x)(u-x)$ almost everywhere. Absolutely continuity of $L_k(\cdot,u)$ coupled with the fact that $f(x)(u-x)$ is strictly positive for $x<u$ and strictly negative for $x>u$ implies the desired claims.} 

For each $u\in [\underline u_k, \bar u_k]$,  
\begin{align*}
L_k(\bar u_k, u)=c_F(\bar u_k)+\big(1-F(\bar u_k)\big)(\bar u_k-u)=k+\big(1-F(\bar u_k)\big)(\bar u_k-u)\geq k,
\end{align*}
where the second equality follows because $\bar u_k$ solves \eqref{eq:1} when $c_G=c_F$. Similarly, 
\begin{align*}
L_k(1, u)=c_F(1)+\big(1-F(1)\big)(1-u)=0<k.
\end{align*}
Therefore, there is a unique $\mathbf x_k(u)\in [\bar u_k,1)$ such that $L_k\big(\mathbf x_k(u), u\big)=k$. 

Notice that $\mathbf x_k(\bar u_k)=\bar u_k$ because $L(\bar u_k, \bar u_k)=c_F(\bar u_k)=k$. Furthermore, $\mathbf x_k(\cdot)$ is strictly decreasing in $u$ because $L_k$ is strictly decreasing in $(x,u)$ over the relevant region $(\bar u_k, 1)\times (\underline u_k, \bar u_k)$. Hence, $\bar u_k=\mathbf x_k(\bar u_k)<\mathbf x_k(\underline u_k)<1$.

To conclude the proof, note that for all $u\in[\underline u_k,\bar u_k]$, $k=c_{G^\varnothing }(\underline u_k)\geq c_{G^\varnothing }(u)$, where the  equality follows from the fact that $\underline u_k$ solves \eqref{eq:1} when $c_G=c_{G^\varnothing }$, and the inequality follows from the fact that $c_{G^\varnothing }(\cdot)$ is a decreasing function. Hence, 
\begin{align*}
c_{G^{\mathbf x_k(u)}}(u)=\max\{c_{G^\varnothing }(u), L_k\big(\mathbf x_k(u), u\big)\}=\max\{c_{G^\varnothing }(u), k\}=k
\end{align*}
which implies that $u=\mathcal{u}_k(G^{\mathbf x_k(u)})$, as desired.
\end{proof}
\bigskip

\noindent \begin{proof}[Proof of \autoref{lemma:stationary-implement}]
Fix any $n\geq 1$, $k\in (0,m^\varnothing)$, and $y\in \mathcal{F}(n, k)$. We first show that we can generate a surplus of $\norm{y}$ while giving the agent a payoff of $y_{n+1}$. We then show that we can distribute the remaining $\norm{y}-y_{n+1}$ surplus to the brokers in such a way that each Broker $i\in N$ earns $y_i$. 

To that end, let $(\alpha, \beta, \psi)$ be as in the statement of the lemma, i.e., $\alpha_{i}(o^*(y))=1$ for all $i\in N$, $\psi\big(O^*(y), w,m\big)=y$ for all $(w,m)\in W\times \Theta$, and $\beta=(\beta^\mathcal{w},\beta^\mathcal{d})$ is sequentially rational with respect to $\psi$. The latter necessarily implies that 
\[
\beta^\mathcal{d}\big(O^*(y), w,m\big)=\begin{cases}
    1 & \mbox{if }  m<y_{n+1}\\
    0 & \mbox{if }  m>y_{n+1}
    \end{cases}
\]
for all $(w,m)\in W\times \Theta$, and that
\[
\supp\beta^\mathcal{w}\big(\cdot|O^*(y)\big)\subseteq\argmax_{w\in W} \int_\Theta \max\left\{m, y_{n+1}\right\}dG_w(m)-k-p_w,
\]
where $(p_w, G_w)=\big(P(y), G^{\mathbf x_k(\norm{y})}\big)$ if $w\in N$ and $(p_w, G_w)=(0, G^\varnothing)$ if $w=\emptyset$.

Clearly, the agent is indifferent across the $n$ identical $o^*(y)$ offers. If he chooses one of the $o^*(y)$ offers (i.e., $w\in N$), his expected payoff is 
\begin{align*}
\int_\Theta\max\{m, y_{n+1}\}dG^{\mathbf x_k(\norm{y})}(m)-k-P(y)=&y_{n+1}+c_{G^{\mathbf x_k(\norm{y})}}(y_{n+1})-k-P(y)\\[6pt]
=&y_{n+1}+c_{G^{\mathbf x_k(\norm{y})}}(\norm{y})-k\\[6pt]
=&y_{n+1},
\end{align*}
where the second equality follows from substituting in the expression for $P(y)$ as defined in \eqref{eq:price}, and the last equality follows because $\norm{y}=\mathcal{u}_k(G^{\mathbf x_k(\norm{y})})$ (see \autoref{lemma:pass-fail-sufficiency}). 

If the agent instead chooses the null offer (i.e., $w=\emptyset$), his expected payoff is 
\begin{align*}
\int_\Theta\max\{m, y_{n+1}\}dG^{\varnothing}(m)-k=y_{n+1}+c_{G^\varnothing}(y_{n+1})-k
\leq y_{n+1}+c_{G^\varnothing}(\underline u_k)-k
=y_{n+1},
\end{align*}
where the inequality follows because $c_{G^\varnothing}(\cdot)$ is a decreasing function and $y_{n+1}\geq \underline u_k$, and the last equality follows because $\underline u_k=\mathcal{u}_k(G^\varnothing)$. Therefore, the agent (weakly) prefers one of the $n$ identical $o^*(y)$ offers instead of the null offer. Additionally, picking any one of the $n$ offers generates the agent a payoff of $y_{n+1}$. 

It remains to show that we can also generate $y_i$ for each broker $i\in N$. Let $\beta^\mathcal{d}\big(O^*(y), w, m\big)=\mathbbm{1}_{[m\leq y_{n+1}]}$ for all $(w,m)\in W\times \Theta$ so that Broker $i$'s payoff is given by 
\begin{align*}
V_i\big(O^*(y), \beta, \psi\big)=\beta^\mathcal{w}\big(i|O^*(y)\big) P(y) + G^{\mathbf x_k(\norm{y})}(y_{n+1})y_i.\label{eq:a2}
\tag{A3}
\end{align*}
Let us simplify \eqref{eq:a2} by making two observations. Our first observation is about the price $P(y)=c_{G^{\mathbf x_k(\norm{y})}}(y_{n+1})-c_{G^{\mathbf x_k(\norm{y})}}(\norm{y})$. From \autoref{lemma:dist-function}, we know that 
\[
c_{G^{\mathbf x_k(\norm{y})}}(\norm{y})=c_F\big(\mathbf x_k(\norm{y})\big)+\big(1-F\big(\mathbf x_k(\norm{y})\big)\big)\big(\mathbf x_k(\norm{y})-\norm{y}\big),
\]
and
\[
c_{G^{\mathbf x_k(\norm{y})}}(y_{n+1})=c_F\big(\mathbf x_k(\norm{y})\big)+\big(1-F\big(\mathbf x_k(\norm{y})\big)\big)\big(\mathbf x_k(\norm{y})-y_{n+1}\big).
\]
Therefore, 
\begin{align*}
P(y)=c_{G^{\mathbf x_k(\norm{y})}}(y_{n+1})-c_{G^{\mathbf x_k(\norm{y})}}(\norm{y})=\big(1-F\big(\mathbf x_k(\norm{y})\big)\big)\sum_{i\in N}y_i. \label{eq:a3}
\tag{A4}
\end{align*}

Our second observation is about $G^{\mathbf x_k(\norm{y})}(y_{n+1})$. From \autoref{lemma:support}, we know that
\[
\mathbb{E}_F[\theta|\theta\leq  \mathbf x_k(\norm{y})]\leq y_{n+1}<\mathbb{E}_F\big[\theta|\theta\geq \mathbf x_k(\norm{y})\big],
\]
which in turn implies that $G^{\mathbf x_k(\norm{y})}(y_{n+1})=F(\mathbf x_k(\norm{y}))$ by \eqref{eq:pass-fail-distibution}.

These two observations--$P(y)=\big(1-F(\mathbf x_k(\norm{y}))\big)\sum_{i\in N}y_i$ and $G^{\mathbf x_k(\norm{y})}(y_{n+1})=F(\mathbf x_k(\norm{y}))$--allow us to simplify \eqref{eq:a2}: 
\begin{align*}
V_i\big(O^*(y), \beta, \psi\big)=\beta^\mathcal{w}\big(i|O^*(y)\big)\left(\big(1-F(\mathbf x_k(\norm{y}))\big)\sum_{j\in N}y_j\right)+F(\mathbf x_k(\norm{y}))y_i.\label{eq:profit}
\tag{A5}
\end{align*}

The agent is indifferent across the $n$ identical TIOLI offers made by the brokers, and  he weakly prefers these offers to the null offer. Let 
\[
\beta^\mathcal{w}\big(i|O^*(y)\big)=\begin{cases}
 \frac{y_i}{\sum_{j\in N} y_j}    &  \mbox{if } \sum_{j\in N}y_j>0\\
   1/n &  \mbox{if }  \sum_{j\in N}y_j=0 .
\end{cases}
\]
Under this specification, Broker $i$'s payoff in \eqref{eq:profit} reduces to
$V_i\big(O^*(y), \beta, \psi\big)=y_i.$
Thus, we have shown that there exists a tuple $(\alpha, \beta, \psi)$, as specified in the statement of the lemma, that generates $y\in \mathcal{F}(n,k)$. This completes the proof.
\end{proof}

\singlespacing
\bibliographystyle{abbrvnat}
\nocite{}\bibliography{bibref}

@String{Academic = "Academic Press" }

@article{Fershtman_and_pavan,
 author = {Daniel Fershtman and Pavan, Alessandro},
 journal = {mimeo},
  title = {Searching for ``Arms'': Experimentation with Endogenous Consideration Sets},
 year = {2025}
}

@article{Wittwer2026,
 author = {Milena Wittwer},
 journal = {mimeo},
  title = {Moral Hazard and Imperfect Competition in Financial Markets},
 year = {2026}
}

@article{wei79,
  title={Optimal search for the best alternative},
  author={Weitzman, Martin L},
  journal={Econometrica},
  volume={47},
  number={3},
  pages={641--654},
  year={1979},
  publisher={JSTOR}
}

@article{mcc70,
  title={Economics of information and job search},
  author={McCall, John Joseph},
  journal={The Quarterly Journal of Economics},
  pages={113--126},
  year={1970},
  publisher={JSTOR}
}

@article{bla53,
  title={Equivalent comparisons of experiments},
  author={Blackwell, David},
  journal={The Annals of Mathematical Statistics},
  pages={265--272},
  year={1953},
  publisher={JSTOR}
}

@article{boa19,
  title={Competitive information disclosure in search markets},
  author={Board, Simon and Lu, Jay},
  journal={Journal of Political Economy},
  volume={126},
  number={5},
  pages={1965--2010},
  year={2018},
  publisher={University of Chicago Press Chicago, IL}
}

@article{au20,
  title={Competitive information disclosure by multiple senders},
  author={Au, Pak Hung and Kawai, Keiichi},
  journal={Games and Economic Behavior},
  volume={119},
  pages={56--78},
  year={2020},
  publisher={Elsevier}
}

@article{do22,
  title={Consumer search and optimal information},
  author={Dogan, Mustafa and Hu, Ju},
  journal={The RAND Journal of Economics},
  volume={53},
  number={2},
  pages={386--403},
  year={2022},
  publisher={Wiley Online Library}
}

@article{hu22,
  title={Industry-optimal Information in the Search Market},
  author={Hu, Ju},
  journal={SSRN WP 3935299},
  year={2022}
}

@article{au2023attraction,
  title={Attraction versus Persuasion: Information Provision in Search Markets},
  author={Au, Pak Hung and Whitmeyer, Mark},
  journal={Journal of Political Economy},
  volume={131},
  number={1},
  pages={202--245},
  year={2023},
  publisher={The University of Chicago Press Chicago, IL}
}

@article{kamenica2019bayesian,
  title={Bayesian persuasion and information design},
  author={Kamenica, Emir},
  journal={Annual Review of Economics},
  volume={11},
  pages={249--272},
  year={2019},
  publisher={Annual Reviews}
}

@article{kamenica2017competition,
  title={Competition in persuasion},
  author={Gentzkow, Matthew and Kamenica, Emir},
  journal={The Review of Economic Studies},
  volume={84},
  number={1},
  pages={300--322},
  year={2016},
  publisher={Review of Economic Studies Ltd}
}

@article{ber18,
  title={The design and price of information},
  author={Bergemann, Dirk and Bonatti, Alessandro and Smolin, Alex},
  journal={American economic review},
  volume={108},
  number={1},
  pages={1--48},
  year={2018},
  publisher={American Economic Association 2014 Broadway, Suite 305, Nashville, TN 37203}
}

@article{ber15,
  title={Selling cookies},
  author={Bergemann, Dirk and Bonatti, Alessandro},
  journal={American Economic Journal: Microeconomics},
  volume={7},
  number={3},
  pages={259--294},
  year={2015},
  publisher={American Economic Association 2014 Broadway, Suite 305, Nashville, TN 37203-2425}
}

@article{es07,
  title={The price of advice},
  author={Es{\H{o}}, P{\'e}ter and Szentes, Bal{\'a}zs},
  journal={The Rand Journal of Economics},
  volume={38},
  number={4},
  pages={863--880},
  year={2007},
  publisher={Wiley Online Library}
}

@article{adm86,
  title={A monopolistic market for information},
  author={Admati, Anat R and Pfleiderer, Paul},
  journal={Journal of Economic Theory},
  volume={39},
  number={2},
  pages={400--438},
  year={1986},
  publisher={Elsevier}
}

@article{adm90,
  title={Direct and indirect sale of information},
  author={Admati, Anat R and Pfleiderer, Paul},
  journal={Econometrica},
  pages={901--928},
  year={1990},
  publisher={JSTOR}
}

@article{aut08,
  title={Does job testing harm minority workers? Evidence from retail establishments},
  author={Autor, David  and Scarborough, David},
  journal={The Quarterly Journal of Economics},
  volume={123},
  number={1},
  pages={219--277},
  year={2008},
  publisher={MIT Press}
}

@article{mek23,
  title={Persuaded search},
  author={Mekonnen, Teddy and Murra-Anton, Zeky and Pakzad-Hurson, Bobak},
  journal={Journal of Political Economy},
  volume={133},
  number={10},
  pages={3167--3207},
  year={2025}
}

@article{li21,
  title={Sequential persuasion},
  author={Li, Fei and Norman, Peter},
  journal={Theoretical Economics},
  volume={16},
  number={2},
  pages={639--675},
  year={2021},
  publisher={Wiley Online Library}
}

@article{abr90,
  title={Toward a theory of discounted repeated games with imperfect monitoring},
  author={Abreu, Dilip and Pearce, David and Stacchetti, Ennio},
  journal={Econometrica},
  pages={1041--1063},
  year={1990},
  publisher={JSTOR}
}

@article{fud86,
 author = {Drew Fudenberg and Eric Maskin},
 journal = {Econometrica},
 number = {3},
 pages = {533--554},
 publisher = {[Wiley, Econometric Society]},
 title = {The Folk Theorem in Repeated Games with Discounting or with Incomplete Information},
 urldate = {2024-10-12},
 volume = {54},
 year = {1986}
}

@article{sato23,
  title={Persuasion in ordered search},
  author={Sato, Hiroto and Shirakawa, Ryo},
  journal={SSRN WP 4483732},
  year={2023}
}

@article{au24,
  title={Attraction via prices and information},
  author={Au, Pak Hung and Whitmeyer, Mark},
  journal={arXiv preprint arXiv:2402.11754},
  year={2024}
}

@article{dana2011productquality,
 author = {James D. Dana and Yuk-Fai Fong},
 journal = {International Economic Review},
 number = {4},
 pages = {1059--1076},
 title = {Product Quality, Reputation, and Market Structure},
 volume = {52},
 year = {2011}
}

@article{dia71,
  title={A model of price adjustment},
  author={Diamond, Peter A},
  journal={Journal of economic theory},
  volume={3},
  number={2},
  pages={156--168},
  year={1971},
  publisher={Academic Press}
}

@article{gen17,
  title={Bayesian persuasion with multiple senders and rich signal spaces},
  author={Gentzkow, Matthew and Kamenica, Emir},
  journal={Games and Economic Behavior},
  volume={104},
  pages={411--429},
  year={2017},
  publisher={Elsevier}
}

@article{bon24,
  title={Selling information in competitive environments},
  author={Bonatti, Alessandro and Dahleh, Munther and Horel, Thibaut and Nouripour, Amir},
  journal={Journal of Economic Theory},
  volume={216},
  pages={105779},
  year={2024},
  publisher={Elsevier}
}

@article{rod24,
  title={Strategic incentives and the optimal sale of information},
  author={Rodr{\'\i}guez Olivera, Rosina},
  journal={American Economic Journal: Microeconomics},
  volume={16},
  number={2},
  pages={296--353},
  year={2024},
  publisher={American Economic Association 2014 Broadway, Suite 305, Nashville, TN 37203-2425}
}

@article{wu23,
  title={Sequential bayesian persuasion},
  author={Wu, Wenhao},
  journal={Journal of Economic Theory},
  volume={214},
  pages={105763},
  year={2023},
  publisher={Elsevier}
}

@article{he23,
  title={Competitive information disclosure in random search markets},
  author={He, Wei and Li, Jiangtao},
  journal={Games and Economic Behavior},
  volume={140},
  pages={132--153},
  year={2023},
  publisher={Elsevier}
}

@article{li18,
  title={On Bayesian persuasion with multiple senders},
  author={Li, Fei and Norman, Peter},
  journal={Economics Letters},
  volume={170},
  pages={66--70},
  year={2018},
  publisher={Elsevier}
}

@article{coey2020discounts,
  title={Discounts and deadlines in consumer search},
  author={Coey, Dominic and Larsen, Bradley J and Platt, Brennan C},
  journal={American Economic Review},
  volume={110},
  number={12},
  pages={3748--3785},
  year={2020},
  publisher={American Economic Association 2014 Broadway, Suite 305, Nashville, TN 37203}
}

@article{aumann1961borel,
  title={Borel structures for function spaces},
  author={Aumann, Robert J},
  journal={Illinois Journal of Mathematics},
  volume={5},
  number={4},
  pages={614--630},
  year={1961},
  publisher={Duke University Press}
}

@article{bonatti2020consumer,
  author  = {Bonatti, Alessandro and Cisternas, Gonzalo},
  title   = {Consumer Scores and Price Discrimination},
  journal = {The Review of Economic Studies},
  year    = {2020},
  volume  = {87},
  number  = {2},
  pages   = {750--791}
}

@article{kayaRoy2024repeated,
  author  = {Kaya, Ay{\c c}a and Roy, Santanu},
  title   = {Repeated Trading: Transparency and Market Structure},
  journal = {American Economic Review},
  year    = {2024},
  volume  = {114},
  number  = {8},
  pages   = {2388--2435}
}

@article{dilme2025repeated,
  author  = {Dilm{\'e}, Francesc},
  title   = {Repeated Trade with Imperfect Information about Previous Transactions},
  journal = {Theoretical Economics},
  year    = {2025},
  volume  = {20},
  pages   = {209--254}
}
\end{document}